\newif\iffull
\fulltrue   

\iffull
	\documentclass[hidelinks]{article}
	\usepackage{fullpage}
	\usepackage[dvipsnames]{xcolor}
\else
	\documentclass[format=acmsmall, review=false]{acmart}
	\usepackage{acm-ec-21}
	\usepackage{booktabs} 
	\usepackage[ruled]{algorithm2e} 
	
	\SetAlFnt{\small}
	\SetAlCapFnt{\small}
	\SetAlCapNameFnt{\small}
	\SetAlCapHSkip{0pt}
	\IncMargin{-\parindent}
	\setcitestyle{acmnumeric}
	\usepackage[font=small,skip=0pt]{caption}
\fi

\usepackage{wrapfig}

\usepackage{environ}

\iffull
	\long\def\fullversion#1\fullversionend{#1}
	\long\def\shortversion#1\shortversionend{}
\else
	\long\def\fullversion#1\fullversionend{}
	\long\def\shortversion#1\shortversionend{#1}
\fi

\NewEnviron{maybeappendix}[1]
    {\iffull%
        \expandafter\newcommand\csname putmaybeappendix#1\endcsname{}\BODY
    \else
        \expandafter\global\expandafter\let\csname putmaybeappendix#1\endcsname\BODY%
    \fi}
\newcommand{\putmaybeappendix}[1]{\csname putmaybeappendix#1\endcsname}

\usepackage{mathtools,centernot,amsmath,amssymb,amsthm,array}
\usepackage{stackengine}
\usepackage[normalem]{ulem} 
\usepackage{hyperref}
\usepackage{caption} 
\usepackage{tikz}

\usepackage{adjustbox}
\usepackage{graphicx}

\usepackage{multirow}
\newcolumntype{M}[1]{>{\centering\arraybackslash}m{#1}}
\usepackage{caption}

\shortversion
	\newcommand{\figurescale}{0.24}
\shortversionend
\fullversion
	\newcommand{\figurescale}{0.3}
\fullversionend

\newcommand\decomp[3]{
	\begin{tabular}{c}
		\hline
		\multicolumn{1}{|c|}{$#2$} \\ \hline
		\multicolumn{1}{|c|}{$#3$} \\ \hline
		\\[-1.4em]
		$(#1)$                  
	\end{tabular}
}
\newcommand\bundle[2]{
	\begin{tabular}{c}
		\hline
		\multicolumn{1}{|c|}{{\multirow{2}{*}{$#2$}}} \\ 
		\multicolumn{1}{|c|}{} \\ \hline
		\\[-1.4em]
		$(#1)$                  
	\end{tabular}
}

\newcommand{\Omit}[1]{}

\newcommand{\valclass}{nice cancelable}
\newcommand{\nice}{nice}
\newcommand{\canc}{cancelable}
\newcommand{\cancbty}{cancelability}

\newcommand{\champ}[1][g]{
	\tikz[baseline=-\the\dimexpr\fontdimen22\textfont2\relax]{
		\matrix [column sep=1.3ex,
				inner sep=0,
				ampersand replacement=\&]{
			\node(a){};	\&
			\node(b)[draw, 
					inner xsep=2pt,
					inner ysep=1pt,
					minimum height=1em,
					rounded corners]{\scriptsize $#1$};\&
			\node(c){$\;\,$};\\
		};
		\draw (a) -- (b);
		\draw[-stealth] (b) -> (c.center)
	}\linebreak[0]
}

\newcommand{\nchamp}[1][g]{
	\tikz[baseline=-\the\dimexpr\fontdimen22\textfont2\relax]{
		\matrix [column sep=1.3ex,
				inner sep=0,
				ampersand replacement=\&]{
			\node(a){};	\&
			\node(b)[draw, 
					inner xsep=2pt,
					inner ysep=1pt,
					minimum height=1em,
					rounded corners]{\scriptsize $#1$};\&
			\node(c){$\;\,$};\\
		};
		\draw (a) -- (b);
		\draw[-stealth] (b) -> (c.center);
		\draw (b.north east) --(b.south west)
	}\linebreak[0]
}
\newcommand{\envies}{%
	\tikz[baseline=-\the\dimexpr\fontdimen22\textfont2\relax]{
		\matrix [column sep=1em,
				inner sep=0,
				ampersand replacement=\&]{
			\node(a){};	\&
			\node(b){$\;\,$};\\
		};
		\draw[-stealth] (a) -> (b.center)
	}\linebreak[0]
}
\newcommand{\nenvies}{\centernot{\envies}}

\newcommand{\circled}[1]{%
 \tikz[baseline=-\the\dimexpr\fontdimen22\textfont2\relax]\node[draw,circle,inner sep=0.4pt, outer sep=0pt, minimum width=1.2em](a){#1};%
 }

\newcommand{\allocs}{\mathbf{X}}
\newcommand{\allocshat}{\mathbf{X}''}
\newcommand{\hX}{X''}

\newcommand{\cupdot}{\mathbin{\mathaccent\cdot\cup}}

\newcommand{\vip}{a_{\textsf{vip}}}

\newcommand{\discard}[3][g]{D^{#1}_{#2,#3}}

\renewcommand{\succ}{\mathsf{succ}}
\newcommand{\pred}{\mathsf{pred}}

\newcommand{\g}{{\color{red}{g}}}
\newcommand{\h}{{\color{blue}{h}}}
\newcommand{\Bgone}{{\color{Green}{B^g_1}}}
\renewcommand{\b}{{\color{Green}{b}}}  
\renewcommand{\c}{{\color{Fuchsia}{c}}} 
\newcommand{\Tgone}{T^g_1}
\newcommand{\Bgtwo}{{\color{Magenta}{B^g_2}}}
\newcommand{\Tgtwo}{T^g_2}
\newcommand{\Bgthree}{{\color{Orange}{B^g_3}}}
\newcommand{\Tgthree}{T^g_3}
\newcommand{\Bgfour}{{\color{Sepia}{B^g_4}}}
\newcommand{\Tgfour}{T^g_4}
\newcommand{\Bhone}{{\color{OliveGreen}{B^h_1}}}
\newcommand{\Bhtwo}{{\color{Fuchsia}{B^h_2}}}
\newcommand{\Thtwo}{T^h_2}
\newcommand{\Bhthree}{{\color{Cyan}{B^h_3}}}
\newcommand{\Ththree}{T^h_3}
\newcommand{\Bhfour}{{\color{CadetBlue}{B^h_4}}}

\theoremstyle{definition}
\newtheorem{theorem}{Theorem}[section]
\newtheorem{lemma}[theorem]{Lemma}
\newtheorem{definition}[theorem]{Definition}
\newtheorem{observation}[theorem]{Observation}
\newtheorem{remark}[theorem]{Remark}
\newtheorem{example}[theorem]{Example}
\newtheorem{claim}[theorem]{Claim}

\newtheorem{proposition}[theorem]{Proposition}
\newtheorem{convention}[theorem]{Convention}
\newtheorem{corollary}[theorem]{Corollary}

\title{(Almost Full) EFX Exists for Four Agents (and Beyond)}

\author{Ben Berger\thanks{Supported by the European Research Council (ERC) under the European Union's Horizon 2020 research and innovation program (grant agreement No. 866132), and by the Israel Science Foundation (grant number 317/17).} \\
	Tel Aviv University\\
	\text{benberger1@tauex.tau.ac.il}
	\and
	Avi Cohen$^{\ast}$ \\
	Tel Aviv University\\
	\text{avicohen2@mail.tau.ac.il}
	\and 
	Michal Feldman$^{\ast}$\\ 
	Tel Aviv University\\ 
	\text{michal.feldman@cs.tau.ac.il}
	\and
	Amos Fiat\thanks{Supported by the Israel Science Foundation (grant number 1841/14) and the Blavatnik fund.}\\
	Tel Aviv University\\
	\text{fiat@tau.ac.il}
	}

\shortversion
\begin{abstract}
The existence of EFX allocations is a major open problem in fair division, even for additive valuations.
The current state of the art is that no setting where EFX allocations are impossible is known, and EFX is known to exist for ($i$) agents with identical valuations, ($ii$) 2 agents, ($iii$) 3 agents with additive valuations, ($iv$) agents with one of two additive valuations and ($v$) agents with two valued instances. 
It is also known that EFX exists if one can leave $n-1$ items unallocated, where $n$ is the number of agents.
	
We develop new techniques that allow us to push the boundaries of the enigmatic EFX problem beyond these known results, and, arguably,  to  simplify proofs of earlier results. 
Our main results are ($i$) every setting with 4 additive agents admits an EFX allocation that leaves at most a single item unallocated, ($ii$) every setting with $n$ additive valuations has an EFX allocation with at most $n-2$ unallocated items. 

Moreover, all of our results extend beyond additive valuations to all {\em nice cancelable} valuations (a new class, including additive, unit-demand, budget-additive and multiplicative valuations, among others). Furthermore, using our new techniques, we show that previous results for additive valuations extend to nice cancelable valuations. 
\end{abstract}
\shortversionend

\begin{document}

\shortversion
\begin{titlepage}
\maketitle
\end{titlepage}
\shortversionend

\fullversion
\maketitle
\begin{abstract}
The existence of EFX allocations is a major open problem in fair division, even for additive valuations.
The current state of the art is that no setting where EFX allocations are impossible is known, and EFX is known to exist for ($i$) agents with identical valuations, ($ii$) 2 agents, ($iii$) 3 agents with additive valuations, ($iv$) agents with one of two additive valuations and ($v$) agents with two valued instances. 
%
It is also known that EFX exists if one can leave $n-1$ items unallocated, where $n$ is the number of agents.
	
We develop new techniques that allow us to push the boundaries of the enigmatic EFX problem beyond these known results, and, arguably,  to  simplify proofs of earlier results. 
Our main results are ($i$) every setting with 4 additive agents admits an EFX allocation that leaves at most a single item unallocated, ($ii$) every setting with $n$ additive valuations has an EFX allocation with at most $n-2$ unallocated items. 

Moreover, all of our results extend beyond additive valuations to all {\em nice cancelable} valuations (a new class, including additive, unit-demand, budget-additive and multiplicative valuations, among others). Furthermore, using our new techniques, we show that previous results for additive valuations extend to nice cancelable valuations. 
\end{abstract}
\fullversionend


\section{Introduction}
\label{sec:intro}
The question of justness, fairness and division of resources and commitments dates back to Aristotle \cite{Chroust42}. {\sl Distributional justice}, the ``just" allocation of limited resources, is fundamental in the work of Rawls \cite{Rawls99}. Some evidence of the great interest in Rawls' work is
that numerous editions of his book have been cited over 100,000 times.

\shortversion
The mathematical study of fair division is due to Hugo Steinhaus, Bronislaw Knaster and Stefan Banach \cite{Steinhaus49} who considered proportional allocations, in which every one of the $n$ agents gets at least a $1/n$ fraction of her total value for all the goods.
\shortversionend

\fullversion
The mathematical study of fair division is credited to Hugo Steinhaus, Bronislaw Knaster and Stefan Banach.
The tale is that they would meet at the Scottish Caf\'e in Lvov where they wrote a book of open problems --- the ``Scottish book" --- preserved by Steinhaus throughout the war and subsequently translated by Ulam and published in the United States. Subsequent editions of this book were written following the end of the war.   In 1944 Hugo Steinhaus proposed the problem of dividing a cake into $n$ pieces so that every agent gets at least a $1/n$ fraction of her total utility (``proportional division"). Steinhaus was actively working despite then living under German occupation in fear for his life. A solution to Proportional division of a cake was credited to Banach and Knaster by Steinhaus in 1949 \cite{Steinhaus49}.
\fullversionend

A stronger notion of fairness is that of an {\em envy free} (EF) allocation --- introduced by Gamow and Stern \cite{gamow1958puzzle} for cake cutting, and in the context of general resource allocation by Foley \cite{Foley67}.
Unfortunately, if goods are indivisible, envy free allocations need not exist. Consider the trivial case of one indivisible good --- if Alice gets the good, others will be envious. Lipton {\sl et al.} \cite{Lipton04} and Budish \cite{Budish2011} consider a relaxed notion of envy freeness, namely \emph{envy freeness up to some item (EF1)} --- an allocation is EF1 if for every pair of agents Alice and Bob, there is an item that we can remove from Alice's allocation such that Bob will not be interested in swapping his allocation with what remains of Alice's allocation.

EF1 allocations always exist but
their fairness guarantees are questionable. Consider for example a setting where Alice and Bob have identical valuations over 3 items $a,b,c$ with respective values $1,1,2$. Arguably, a fair allocation would assign $a,b$ to one of the players, and $c$ to the other one, giving each a value $2$. However, the allocation that assigns $a,c$ to Alice and $b$ to Bob is also EF1.

The notion of Envy Freeness up to any item (EFX) was introduced by Caragiannis {\sl et al.} \cite{caragiannis2016unreasonable,caragiannis2019unreasonable}.
An allocation is EFX if for every pair of agents, Alice and Bob, Bob does not want to swap with what remains of Alice's allocation when {\sl any item is discarded}. {\sl I.e.}, it suffices to consider removing the item with minimal marginal value (to Bob) from Alices's allocation.
Indeed, in the example above, the only EFX allocations are those that allocate $a,b$ to one player and $c$ to the other player.

A major open problem is ``when do EFX allocations exist?". The current state of our knowledge is somewhat embarrassing. We do not know how to rule out EFX allocations in any setting, and yet, they are known to exist only in several restricted cases.
In particular, Plaut and Roughgarden, \cite{plaut2020almost}, 
prove that EFX valuations exist for 2 agents with arbitrary valuations, and for any number of agents with identical valuations.
Even for the simple case of {\em additive} valuations (where the value of a bundle of items is simply the sum of values of individual goods), EFX is only known to exist in settings with 3 agents (Chaudhury, Garg, and Mehlhorn \cite{chaudhury2020efx}), in settings with only one of two types of additive valuations (Mahara \cite{Mahara2020}),
or when the value of every agent to every item can take one of two permissible values (Amanatidis {\sl et al.} \cite{Amanatidis20}).


\noindent Indeed, Procaccia \cite{ProcacciaCACM} recently wrote:

\begin{quote}
	In my view, it (EFX existence) is the successor of envy-free cake cutting as fair division's biggest problem.
\end{quote}


Given that EFX valuations are known to exist in so few cases, the following question arises: Can one find a good {\em partial} EFX allocation? {\sl i.e.}, an EFX allocation in which only a small amount of items can be unallocated? The idea of partial allocations for EF and EFX allocations has appeared in multiple papers, e.g. \cite{Brams13,Cole2013,caragiannis2019envy}. Caragiannis, Gravin and Huang, \cite{caragiannis2019envy} show that discarding some items gives good EFX allocations for the rest (achieving 1/2 of the maximum Nash Welfare). Chaudhury {\sl et. al} \cite{chaudhury2020little} show that given $n$ agents with arbitrary valuations, there always exists an EFX allocation with at most $n-1$ unallocated items. Moreover, no agent prefers the set of unallocated items to her own allocation.


\subsection{Our Results}
\label{sec:results}

In this paper we develop new techniques, based upon ideas that appear in \cite{chaudhury2020efx,chaudhury2020little}.
\cite{chaudhury2020efx} introduced the notion of {\em champion edges} with respect to a single unallocated good, and used it to make progress with respect to the lexicographic potential function in order to eventually reach an EFX allocation.
We extend the notion of champion edges beyond a single unallocated item, to {\em sets} of items, allocated or not, and derive useful structural properties that allow us to make more aggressive progress within a graph theoretic framework.

Our techniques are powerful enough to allow us to $(i)$ push the boundaries of EFX existence beyond known results, and $(ii)$ present substantially simpler proofs for previously known results. Our results are described below, and are summarized in Table~\ref{tab:results}.

Our main result concerns EFX allocation for four agents.
Extending EFX existence from three to four agents is highly non-trivial.
Indeed, \cite{chaudhury2020example} discovered an instance with 4 additive agents in which there exists an EFX allocation with one unallocated item such that no progress can be made based on the lexicographic potential function.
We show that one unallocated item is the only possible obstacle to EFX existence in any setting with four agents.





\vspace{0.1in}
\noindent {\bf Theorem 1 (Main Result):} Every setting with 4 additive agents admits an EFX allocation with at most a single unallocated item (which is not envied by any agent).
\vspace{0.1in}

Moreover, this result extends beyond additive valuations to a new class of valuations that we term {\em \valclass} valuations (including additive, unit-demand, budget additive and more).

To prove Theorem 1, we show that for any EFX allocation with at least two unallocated items, it is possible to reshuffle bundles and reallocate them in such a way that advances the lexicographic potential function and preserves EFX.
The proof requires solving a complex puzzle, and exemplifies the extensive use of our new techniques.

The immediate open problem is whether one can go the additional mile and allocate the one item that remains.
A natural approach to solving this problem is by using a different potential function.
Notably, our new techniques are orthogonal to the choice of the potential function, and may prove useful in analyzing other potential functions.

Our second result makes an additional progress in settings with an arbitrary number of agents.

\fullversion
\vspace{0.1in}
\fullversionend
\noindent {\bf Theorem 2:}
Every setting with $n$ additive agents admits an EFX allocation with at most $n-2$ unallocated items.
\vspace{0.1in}

Here too, the unallocated items are not envied by anyone, and the result extends to all \valclass\ valuations.

To prove Theorem 2, we show that for any EFX allocation with at least $n-1$ unallocated items, one can reshuffle bundles and reallocate them in a way that results in a Pareto-dominating EFX allocation.
This means that one can find an EFX allocation with at most $n-2$ unallocated items.

In addition to these results, we establish the following extensions and simplifications.


\vspace{0.1in}
\noindent{\bf Beyond additive valuations.}
We extend the results of \cite{chaudhury2020efx} (EFX for 3 additive agents) and  \cite{Mahara2020} (EFX for $n$ agents with one of two additive valuations) beyond additive valuations to a class of valuations that we term {\sl \valclass} valuations, defined as follows: a valuation $v$ is {\em \canc} if
for any two bundles $S,T\subseteq M$ and item $g\in M\setminus \left(S\cup T\right)$ we have
$
v\left(S\cup\left\{ g\right\} \right)>v\left(T\cup\left\{ g\right\} \right)\Rightarrow v\left(S\right)>v\left(T\right).
$
It is \valclass\ if there exists a \canc\ and non-degenerate valuation ({\sl i.e.}, different bundles have different values) that respects all strong inequalities of $v$.

Nice \canc\ valuations include additive valuations but also valuations such as budget additive ($v(S)=\min(\sum_{j\in S} v(j), B)$ for some value $B$), unit demand ($v(S)=\max_{j\in S} v(j)$), multiplicative valuations ($v(S)=\prod_{j\in S} v(j)$), and others.

Moreover, one can have different  nice \canc\ valuations for different agents\footnote{ {\sl E.g.}, an EFX allocation exists for three agents when agent $a$ has a multiplicative valuation, $b$ has  a budget additive valuation, and $c$ has a unit demand valuation. Another example, fix any two nice \canc\ valuations, some agents have the 1st and others have the 2nd, an EFX allocation still exists.}.
We remark that the original proofs, as written, do not directly generalize to this more general class of valuations. It is our new techniques that allow us to generalize the results to this class.

\vspace{0.1in}
\noindent {\bf Theorem 3:}
EFX existence for 3 agents and EFX existence for two types of valuations (any number of agents) extends beyond additive valuations to all \valclass\ valuations.
\vspace{0.1in}

\noindent {\bf Simplification of proofs for known results.}
Our new techniques greatly simplify existing proofs of EFX existence for 3 agents \cite{chaudhury2020efx} and for the case of 2 types of additive valuations \cite{Mahara2020}.
Admittedly, simplicity is a matter of subjective judgment, but at least in terms of character count, the proof for the case of 3 agents with no envy in \cite{chaudhury2020efx} (some 5 pages) drops to half a page using our new techniques. Similarly, the proof for settings with 2 types of additive valuations in \cite{Mahara2020} (some 8 pages) drops to one page using our new techniques.
Moreover, in both settings, the simplified proofs apply beyond additive valuations to all \valclass\ valuations.

We believe that we have only scratched the surface of the power of our new techniques, and hope they will prove useful in making further progress on the EFX problem.

{ \begin{table} \centering
		\begin{adjustbox}{max width=\textwidth}
 	\begin{tabular}{|l|l|l|}
		\hline
		\rule[-1ex]{0pt}{2.5ex} Setting & Prior results & Our results \\
		\hline
		\rule[-1ex]{0pt}{2.5ex} EFX for 3 agents  &  Additive \cite{chaudhury2020efx}  & Beyond additive$^{\ast}$   \\
		\hline
		\rule[-1ex]{0pt}{2.5ex} \hbox{EFX for $n$ agents, one of 2 valuations} & Additive \cite{Mahara2020}  & Beyond additive$^{\ast}$\\
		\hline
		\rule[-1ex]{0pt}{2.5ex} Partial EFX for $n$ agents  & $\leq n-1$ unallocated items \cite{chaudhury2020little} & $\leq n-2$ unallocated items \\
				\hline
		\rule[-1ex]{0pt}{2.5ex} Partial EFX for 4 agents  & $\leq 3$ unallocated items \cite{chaudhury2020little} & $\leq 1$ unallocated item  \\
						\hline
	\end{tabular}
\end{adjustbox}
	\caption{All our results hold for (*){\sl \valclass}\ valuations, generalizing additive valuations}\label{tab:results}
\end{table}

\subsection{Our Techniques}
\label{sec:techniques}

Our proof techniques lie within a graph theoretic framework.
Given an EFX allocation $\allocs$, we describe a graph $M_{\allocs}$ (see Definition \ref{def:championgraph}) where vertices are associated with agents and there are three types of edges: envy edges  $i \envies j $,
champion edges $i \champ j$, where $g$ is an unallocated item, and
generalized champion edges $i \champ[H \mid S] j$, where  $H$ is some subset of items (allocated or not) and $S$ is a subset of $j$'s allocation in $\allocs$.

The use of such graphs, with envy and champion edges (but no generalized championship edges) has previously appeared in the literature and is a key component in the proof of an EFX allocation for 3 additive agents \cite{chaudhury2020efx}.
The new ingredient introduced in this paper is the notion of generalized champion edges.
We show how to find such edges (Section~\ref{sec:edge-discovery}), and use them to reach a new EFX allocation that advances the lexicographic potential function of \cite{chaudhury2020efx}.

The key idea in all our results is to reshuffle the existing allocation to obtain a new allocation with higher potential, while preserving EFX. This follows the same proof template as in Chaudhury {\sl et al.}  --- but we have more options to play with by using the generalized championship edges.

An envy edge $i \envies j$ suggests a possible reshuffling where agent $i$ gets $j$'s current allocation.
A champion edge $i \champ j$ suggests another reshuffling, where agent $i$ gets a subset of agent $j$'s current allocation, along with the currently unallocated item $g$.
A generalized champion edge $i \champ[H \mid S] j$ suggests giving agent $i$ agent $j$'s allocation along with some arbitrary set of items $H$ (that may be arbitrarily allocated among other agents, or be unallocated), while freeing up the set of items $S$.

Our proofs require solving a complex puzzle, where the goal is to find a cycle consisting of envy, champion, and generalized champion edges, such that the union of all sets $S$ freed up along with the currently unallocated items suffice for the suggested allocation along the cycle.

Finding the appropriate generalized championship edges is a major technical component of our techniques (see Section~\ref{sec:edge-discovery}).
We show how to find such edges, based on existing edges.
Then,
these edges 
allow us to reshuffle the current allocation and advance the potential.

In several of our proofs we consider the case where $M_\allocs$ has envy edges separately from the case where it has not. For our applications, it turns out that if there are no envy edges, one can Pareto improve the EFX allocation (See Section~\ref{sec:ef-case}), in particular this advances the potential function. If there are envy edges it may no longer be possible to Pareto improve (as pointed out by \cite{chaudhury2020efx}) but one can nonetheless advance the potential function itself (see Section~\ref{sec:not-ef-case}).

\subsection{Other Related Work}\label{sec:otherrelated}
\fullversion
Varian \cite{Varian74} introduced {\sl competitive equilibria from equal incomes (CEEI)} which assigns equal budgets to all agents and computes prices such that every agent gets a most favored set of (divisible) items, and the result is envy free.

Lipton {\sl et al.} \cite{Lipton04} give a greedy algorithm for producing an EF1 allocation, this adds items to an agent that is not envied, or --- alternatively --- switches bundles around an envy cycle, the so-called envy-cycles procedure.

Caragiannis {\sl et al.} \cite{caragiannis2019unreasonable} show that maximizing Nash welfare (maximizing the product of agent utilities) gives an EF1 allocation that is also Pareto optimal. Approximations for maximizing Nash welfare were presented by Cole and Gkatzelis \cite{ColeGkatzelis15} and subsequent papers.

Recall that envy free allocations need not exist --- but --- at least for two agents, there are envy free allocations that discard few items \cite{BramsTaylor2000,Brams13}. In \cite{Brams13} the resulting allocation is Pareto optimal, envy free, and maximal (no EF allocation allocates more items).

There have been several papers on truthful fair division. Mosel and Tamuz \cite{Mossel2010} and Cole, Gkatzelis and Goel \cite{Cole2013} deal with proportional fairness. The mechanism in the latter uses no money and provides good guarantees --- the mechanism discards a fraction of the of the resources to achieve truthfulness.

In \cite{caragiannis2019envy} Caragiannis, Gravin, and Huang show that for additive valuations, there is always an EFX allocation of a subset of items with a Nash welfare (product of agent utilities) that is at least half of the maximum possible Nash welfare for the original set of items. Moreover, this is the best possible Nash welfare of all EFX allocations.

Other extensions such as group EFX allocations and EF1 for matroid settings were considered respectively by \cite{Kyropoulou_2020} and \cite{Gourvs2014NearFI}.

While we concentrate on envy free allocations and relaxations thereof --- there are many other possible notions of fairness, e.g., {\sl equitable division} \cite{jones02} --- where the subjective value of every agent for her allocation is the same, and {\sl maximin share} \cite{Budish2011} where an agent gets the most preferred share she could guarantee herself if she was to suggest a partition and be allocated her least valuable bundle of this partition.
\fullversionend

\shortversion
Lipton {\sl et al.} \cite{Lipton04} give a greedy algorithm for producing an EF1 allocation, this adds items to an agent that is not envied, or --- alternately --- switches bundles around an envy cycle, the so-called envy-cycles procedure.
Caragiannis {\sl et al.} \cite{caragiannis2019unreasonable} show that maximizing Nash welfare (maximizing the product of agent utilities) gives an EF1 allocation that is also Pareto optimal. Approximations for maximizing Nash welfare were presented by Cole and Gkatzelis \cite{ColeGkatzelis15} and subsequent papers.

Varian \cite{Varian74} introduced the notion of {\sl competitive equilibria from equal incomes (CEEI)}, which guarantees envy freeness.
Envy free allocations that discard few items were considered in \cite{BramsTaylor2000,Brams13}. In \cite{Brams13} the resulting allocation is Pareto optimal, envy free, and maximal (no EF allocation allocates more items).

There have been some papers on truthful mechanisms for proportional fairness (Mosel and Tamuz \cite{Mossel2010} and Cole, Gkatzelis and Goel \cite{Cole2013}). The latter mechanism uses no money and provides good guarantees --- the mechanism discards a fraction of the resources to achieve truthfulness.
\shortversionend

\shortversion
\subsection{Paper Roadmap}

Due to space limitations, some of the proofs appear in the various appendices and not in the body of the paper.
The next section, Section \ref{sec:preliminaries},  presents the model, introduces ``nice cancelable valuations'', gives definitions from \cite{chaudhury2020efx} as well as new definitions, notation, and proofs.
Some proofs are deferred to Appendix \ref{apx:Preliminaries}.

Section \ref{sec:generalized_championship} defines ``generalized championship''. It may be helpful to consider Figure \ref{fig:example1} while reading the section. How to find generalized championship edges is described in Section \ref{sec:edge-discovery}. Some proofs are deferred to Appendix \ref{apx:techniques}.

Section \ref{sec:efxnm2} proves that it suffices to discard at most $n-2$ items and yet guarantee the existence of an EFX allocation.

In Section \ref{sec:4-agents} we show that for 4 agents it suffices to discard one item and yet still guarantee an EFX allocation. A roadmap of this proof itself appears in Figure \ref{fig:proof-structure}. The proof spans Sections \ref{sec:ef-case}, \ref{sec:not-ef-case}, and parts of the relevant proofs are deferred to Appendix \ref{apx:4-agents}.

Section \ref{sec:extensions} simplifies and extends prior results for 3 agents and for one of two additive valuations.  Some proofs are deferred to  Appendix \ref{sec:2-simple}.

A star ($\star$) indicates that a missing proof has been deferred to the appropriate appendix.
\shortversionend

\fullversion
\subsection{Paper Roadmap}

The next section, Section \ref{sec:preliminaries}, presents the model, introduces ``nice cancelable valuations'', gives definitions from \cite{chaudhury2020efx} as well as new definitions, notation, and proofs.
Useful properties of nice cancelable valuations are presented in Appendix \ref{apx:Preliminaries}.

Section \ref{sec:generalized_championship} defines ``generalized championship'' and establishes the technical framework underlying our EFX existence results.
It may be helpful to consider Figure \ref{fig:example1} while reading the section. How to find generalized championship edges is described in Section \ref{sec:edge-discovery}.

Section \ref{sec:efxnm2} proves that it suffices to discard at most $n-2$ items and yet guarantee the existence of an EFX allocation.

In Section \ref{sec:4-agents} we show that for 4 agents it suffices to discard one item and yet still guarantee an EFX allocation. A roadmap of this proof itself appears in Figure \ref{fig:proof-structure}. The proof spans Sections \ref{sec:ef-case}, \ref{sec:not-ef-case}.
For clarity of exposition, some of the more technical and repetitive proofs are deferred to Appendix \ref{apx:4-agents}.

Section \ref{sec:extensions} simplifies and extends prior results for 3 agents and for one of two additive valuations.
\fullversionend


\section{Preliminaries}\label{sec:preliminaries}


We consider a setting with $n$ agents, and a set $M$ of $m$ items.  Each agent has a valuation $v_i: 2^M \rightarrow \mathbb{R}_{\geq 0}$, which is normalized and monotone, i.e. $v(S) \leq v(T)$ whenever $S \subseteq T$ and $v(\emptyset) = 0$.  

For two sets of items $S,T \subseteq M$, we write $S <_i T$ if $v_i(S)< v_i(T)$.  Similarly we define $S>_i T$, $S\leq_i T$, $S\geq_i T$, $S=_i T$
if $v_i(S)>v_i(T)$, $v_i(S)\leq v_i(T)$, $v_i(S)\geq v_i(T)$, $v_i(S)= v_i(T)$, respectively.

We denote a valuation profile by $\mathbf{v} = (v_1,\ldots,v_n)$. An \emph{allocation} is a vector $\mathbf{X} = (X_1,\ldots,X_n)$ of disjoint bundles, where $X_i$ is the bundle allocated to agent $i$.
Given an allocation $\mathbf{X}$, 
We say that agent $i$ \emph{envies} a set of items $S$
if $X_i <_i S$. We say that agent $i$ \emph{envies} agent $j$
, denoted $i\envies j$,
if $i$ envies $X_j$.
We say that agent $i$ \emph{strongly envies} a set of items $S$
if there exists some $h\in S$ such that $i$ envies $S\setminus \{h\}$. Likewise we say that agent $i$ \emph{strongly envies agent $j$} if $i$ strongly envies $X_j$.
$\mathbf{X}$ is called \emph{envy-free} (EF) if no agent envies another.
$\mathbf{X}$ is called \emph{envy-free up to any good} (EFX) if no agent strongly envies another.
	
\fullversion
\subsection{Nice Cancelable Valuations}
\fullversionend
\shortversion
\noindent\textbf{Nice Cancelable Valuations:}
\shortversionend
We consider a class of valuation functions that can be viewed as a generalization of additive valuations.  These include additive, unit-demand and budget-additive valuations, among others.

\begin{definition}
	A valuation $v$ is {\em \canc} if for any two bundles $S,T\subseteq M$ and an item $g\in M\setminus \left(S\cup T\right)$,
\shortversion
	$
	v\left(S\cup \{g\} \right)>v\left(T\cup \{g\} \right)\Rightarrow v\left(S\right)>v\left(T\right).
	$
\shortversionend
\fullversion
	$$
	v\left(S\cup \{g\} \right)>v\left(T\cup \{g\} \right)\Rightarrow v\left(S\right)>v\left(T\right).
	$$
\fullversionend
\end{definition}

A valuation $v$ is \emph{non-degenerate} if $v\left(S\right)\neq v\left(T\right)$ for any two different bundles $S,T$.
A valuation $v'$ is said to \emph{respect} another valuation $v$ if for every two bundles $S,T\subseteq M$ such that $v(S)>v(T)$ it also holds that $v'(S)>v'(T)$.
%
\shortversion
	A \canc\ valuation $v$ is \emph{\nice} if there is a non-degenerate \canc\ valuation $v'$ that respects $v$. In particular, any non-degenerate \canc\ valuation is a \nice\ valuation (by setting $v'\equiv v$).
\shortversionend
\fullversion
\begin{definition}
	A \canc\ valuation $v$ is \emph{\nice} if there is a non-degenerate \canc\ valuation $v'$ that respects $v$. In particular, any non-degenerate \canc\ valuation is a \nice\ valuation (by setting $v'\equiv v$).
\end{definition}
\fullversionend

We can show that in order to prove the existence of an EFX allocation for a given valuation profile $\mathbf{v} = (v_1,\ldots,v_n)$ of \valclass\ valuations, it is without loss of generality to assume that all of the valuations are non-degenerate (Lemma \ref{lem:non-deg} in Appendix \ref{apx:Preliminaries}).  
Thus, for the remainder of this paper we assume that all valuations are \canc\ and non-degenerate.
Under this assumption it is easy to verify that for any valuation $v$ and bundles $S,T,R$ such that $R \subseteq M \setminus (S\cup T)$ we have
$$v\left(S\cup R \right)>v\left(T\cup R \right) \Leftrightarrow v\left(S\right)>v\left(T\right).$$

Other useful claims on \canc\ valuations are deferred to the appendix. 

\fullversion
\subsection {Potential Functions and Progress Measures}
\fullversionend
\shortversion
\noindent\textbf{Potential Functions and Progress Measures:}
\shortversionend
All our EFX existence results follow the same paradigm: 
given an arbitrary EFX allocation 
$\allocs$ with $k$ unallocated goods, construct a new partial EFX allocation that advances 
some fixed
potential function. 
Since there are finitely many allocations, there must exist an EFX allocation with at most $k-1$ unallocated items.

A natural progress measure to consider is Pareto domination. 
Given two allocations $\mathbf{X}$,$\mathbf{Y}$, we say that $\mathbf{Y}$ \emph{Pareto dominates} $\mathbf{X}$ if $Y_i \geq_i X_i$ for every $i \in [n]$, and there exists some $i$ for which the inequality is strict.
Chaudhury \emph{et al.} \cite{chaudhury2020efx} have shown that there need not exist a Pareto-dominating EFX allocation when $n=3$ and $k\geq 1$.  To overcome this obstacle they introduce an alternative ``lexicographic'' progress measure which we shall also use: 

\begin{definition}[\cite{chaudhury2020efx}]\label{def:domination}
	Fix some arbitrary ordering of the agents $a_1,\ldots, a_n$.  
	The allocation $\mathbf{Y}$ \emph{dominates} $\mathbf{X}$ if for some $k \in [n]$, we have that  
	$Y_{a_j}=_{a_j} X_{a_j}$ for all $1\leq j <k$, and $Y_{a_k} >_{a_k} X_{a_k}$.
\end{definition}
Note that Pareto-domination implies domination but not vice versa.  
\shortversion
\vspace{-5pt}
\begin{lemma}[$\star$]
\shortversionend
\fullversion
\begin{lemma}
\fullversionend
	\label{lem:dominate-implies-progress}
	If for every partial EFX allocation $\allocs$ with $k$ unallocated items, there exists a partial EFX allocation $\mathbf{Y}$ that dominates $\allocs$, then there exists a partial EFX allocation with at most $k-1$ unallocated items. Moreover, no agent envies the set of $k-1$ unallocated items.
\end{lemma}

\begin{maybeappendix}{Lemma_dominate_implies_progress_proof}
\begin{proof}

The first part of the lemma follows directly from the fact that the number of possible allocations is finite and domination is a total order relation.

	We now show that if in a given partial EFX allocation $\allocs$ some agent envies the set of unallocated items then there exists a partial EFX allocation $\mathbf{Y}$ that Pareto dominates $\allocs$. Analogous to the above, this proves the second part of the lemma.
	
If some agent envies $U$,
		then there exists a subset $T$ of $U$ that some agent $i$ envies and $T$ is a smallest subset of $U$ that some agent envies.  We obtain $\mathbf{Y}$ by replacing $X_i$ with $T$.  $\mathbf{Y}$ is EFX, since $i$ did not strongly envy anyone before and now she is better off, and no one strongly envies $i$ by minimality of $T$.
\end{proof}

\end{maybeappendix}

In fact, in our results we almost always progress via Pareto-domination.  In the few cases we do not, we find an allocation in which $a_1$ (the most important agent in the ordering) is strictly better off.  Hereinafter we denote this agent $\vip$.

\fullversion
\subsection{Most Envious Agents}
\fullversionend
\shortversion
\noindent\textbf{Most Envious Agents:}
\shortversionend
Hereinafter, we fix some partial EFX allocation $\mathbf{X}$ and some unallocated good $g$.  We denote by $U$ the set of goods that are unallocated in $\mathbf{X}$ (thus $g \in U$).  The following 
are variants of definitions from \cite{chaudhury2020efx}, \cite{chaudhury2020little}.

	We say that $i$ is {\em most envious} of a set of items $S$, if there exists a subset $T\subseteq S$, such that $i$ envies $T$ and no agent strongly envies $T$. When more than one such $T$ exists, we choose one of them arbitrarily unless stated otherwise. The set $S\setminus T$ is referred to as the corresponding {\em discard set}.
%

\fullversion
\begin{example}\label{ex:most_envious}
	Consider the valuation profile depicted in Table~\ref{tab:example}, describing a setting with 7 items $M=\{a,b,c,d,e,f,g\}$ and 3 agents with additive valuations.
	Let $\allocs$ be the allocation where $X_1=\{a,b,c\}, X_2=\{d\}, X_3=\{e,f\}$. Consider the set $S=\{c,d,e\}$. Agent 1 is most envious of $S$ since $X_1 <_1 \{d\}$ and no agent may strongly envy a singleton. The corresponding discard set is $S\setminus \{d\}= \{c,e\}$. Agent 2 is also most envious of $S$ since $X_2 <_2 \{c,e\}$, and no agent strongly envies $\{c,e\}$. The corresoponding discard set is $S\setminus \{c,e\}=\{d\}$.
\end{example}
\fullversionend

\begin{definition}[\cite{chaudhury2020efx}]\label{def:champ_g}
	We say that \emph{$i$ champions $j$ with respect to $g$}, denoted $i \champ j$,
	if $i$ is most envious of $X_{j}\cup \{g\}$.
	The corresponding discard set is denoted $\discard{i}{j}$.
	Note that $i$ envies the set $(X_{j}\cup \{g\}) \setminus \discard{i}{j}$, but no agent strongly envies it.
\end{definition}

\shortversion
An important case 
considered frequently in the paper 
is where $g \notin \discard{i}{j}$.  In this case $X_j = (X_j\setminus \discard{i}{j}) \cupdot \discard{i}{j}$.
	Following \cite{chaudhury2020efx},
	if $i\champ j$ and $g \notin \discard{i}{j}$, then we say that $i$ \emph{$g$-decomposes} $j$ into \emph{top} and \emph{bottom} \emph{half-bundles} $(X_j\setminus \discard{i}{j})$ and $\discard{i}{j}$, respectively (in short, $i$ $g$-decomposes $j$).	If there is no concern of ambiguity, then we denote the top and bottom half-bundles by $T_j$ and $B_j$, respectively (note that different $g$-decomposers of $j$ may induce different top and bottom half-bundles).
	Under this notation, we have $(X_j \cup \{g\}) \setminus \discard{i}{j} = T_j \cup \{g\}$.
\shortversionend

\fullversion
An important case 
considered frequently in the paper 
is where $g \notin \discard{i}{j}$.  In this case $X_j = (X_j\setminus \discard{i}{j}) \cupdot \discard{i}{j}$.
\begin{definition}[\cite{chaudhury2020efx}]
	If $i\champ j$ and $g \notin \discard{i}{j}$, then we say that $i$ \emph{$g$-decomposes} $j$ into \emph{top} and \emph{bottom} \emph{half-bundles} $(X_j\setminus \discard{i}{j})$ and $\discard{i}{j}$, respectively (in short, $i$ $g$-decomposes $j$).	If there is no concern of ambiguity, then we denote the top and bottom half-bundles by $T_j$ and $B_j$, respectively (note that different $g$-decomposers of $j$ may induce different top and bottom half-bundles).
	Under this notation, we have $(X_j \cup \{g\}) \setminus \discard{i}{j} = T_j \cup \{g\}$.
\end{definition}
\fullversionend



%
In the following observations from \cite{chaudhury2020efx}, $i \nchamp j$ and $i \nenvies j$ are the respective negations of $i \champ j$, $i \envies j$.
\shortversion
\begin{observation}[$\star$]\label{obs:exists-champion}
\shortversionend
\fullversion
\begin{observation}\label{obs:exists-champion}
\fullversionend
	Every agent $i$ has a champion with respect to $g$. 
\end{observation}
\begin{maybeappendix}{obs_exists-champion}
\begin{proof}
	Since valuations are non-degenerate and monotone, $X_i <_i X_i\cup \{g\} $. Since $i$ envies $X_i\cup \{g\} $, this set has a most envious agent, and by definition that agent is a champion of $i$ with respect to $g$.
\end{proof}
\end{maybeappendix}
%
%
\shortversion
\begin{observation}[$\star$]\label{obs:g_notin_bottom}
\shortversionend
\fullversion
\begin{observation}\label{obs:g_notin_bottom}
\fullversionend
	If $i\champ j$ and $i\nenvies j$, then $g\notin \discard{i}{j}$, i.e. $i$ $g$-decomposes $j$.
\end{observation}
\begin{maybeappendix}{obs_g_notin_bottom}
\begin{proof}
	By definition of $\discard{i}{j}$, agent $i$ envies $(X_j\cup \{g\})\setminus \discard{i}{j}$. If $g\in \discard{i}{j}$ then $(X_j\cup \{g\})\setminus \discard{i}{j}$ is a subset of $X_j$, implying that agent $i$ envies a subset of $X_j$ and thus the set $X_j$ as well. Hence, $i$ envies $j$, in contradiction to our assumption.
\end{proof}
\end{maybeappendix}
\shortversion
\begin{observation}[$\star$]\label{obs:non-champion-doesnt-envy-top-half}
\shortversionend
\fullversion
\begin{observation}\label{obs:non-champion-doesnt-envy-top-half}
\fullversionend
		If $i\nchamp j$ and $j$ is $g$-decomposed into $X_j = T_j\cupdot B_j$, then $X_i >_i T_j \cup \{g\}$.
\end{observation}
\begin{maybeappendix}{obs_non-champion-doesnt-envy-top-half}
\begin{proof}
	Assume that $j$ is $g$-decomposed by some agent $k$ into $X_j = T_j\cupdot B_j$.  By definition of $\discard{k}{j}$, no agent strongly envies $(X_j\cup \{g\})\setminus\discard{k}{j} = T_j \cup \{g\}$. Therefore, any agent that envies this set is a most envious agent of $X_j\cup \{g\}$, and thus a $g$-champion of $j$. Hence, agent $i$  does not envy that set.
\end{proof}
\end{maybeappendix}
\shortversion
\begin{observation}[$\star$]\label{obs:T_k<T_j}
\shortversionend
\fullversion
\begin{observation}\label{obs:T_k<T_j}
\fullversionend
	If $i$ $g$-decomposes $j$, $i\nchamp k$ and $k$ is $g$-decomposed, then $T_k <_i T_j$. 
\end{observation}

\begin{maybeappendix}{obs_T_k<T_j}
\begin{proof}
	We have,
	$%
		T_k\cup\{g\} <_i X_i <_i T_j\cup\{g\},%
	$
	where the first inequality is by Observation~\ref{obs:non-champion-doesnt-envy-top-half} since $i\nchamp k$ and the second inequality is by definition of basic championship (since $i$ $g$-decomposes $j$).
\end{proof}
\end{maybeappendix}
%
%

\section{Generalized Championship}\label{sec:generalized_championship}
A crucial component in our techniques is the extension of Definition \ref{def:champ_g} to 
an arbitrary set of items $H$.
%
%
%
It will be useful to have a notation that contains some information regarding the discarded items.
%

\begin{definition}\label{def:champ_H|S}
	\emph{$i$ champions $j$ with respect to ($H\mid S$)},
	denoted $i \champ[H \mid S]j$,
	where $H \subseteq M\setminus X_j$ and $S \subseteq X_j$,
	if $i$ is most envious of 
	$\left(X_{j} \setminus S \right) \cup H$.
	The corresponding discard set is denoted $\discard[H \mid S]{i}{j}$.
\end{definition}

As opposed to basic championship, not every agent $j$ has an ($H\mid S$)-champion (consider an extreme example where $H = \emptyset, S = X_j$).
%
%
If $i\champ[H \mid S]j$, then giving $i$ the desired bundle implied by the championship releases $S$ to be reallocated to other agents.  For example, if we also know that $k \champ[S \mid S'] \ell$, then these two champion relations can be ``used'' simultaneously in a transition to a new EFX allocation.

We say that a set of items $T$ is \emph{released} by $i\champ[H \mid S]j$ if $T\subseteq S\cup \discard[H\mid S]{i}{j}$.
We denote the negation of $i\champ[H \mid S]j$ by $i\nchamp[H \mid S]j$.




\begin{definition} \label{def:championgraph}
The \emph{champion graph} $M_{\mathbf{X}} = ([n],E)$ with respect to $\allocs$ is a labeled directed multi-graph.  The vertices correspond to the agents, and $E$ consists of the following 3 types of edges:
\begin{enumerate}
	\item Envy edges:  $i \envies j $ iff $i$ envies $j$.
	\item Champion edges:  $i \champ j$ iff $i$ champions $j$ with respect to $g$, where $g$ is an unallocated good. 
	\item Generalized champion edges:  $i \champ[H \mid S] j$ iff $i$ champions $j$ with respect to $H \mid S$.
\end{enumerate}
\end{definition}

We refer to envy and champion edges as \emph{basic edges}.
Hereinafter, the edge notations above will sometimes refer to the edges of the champion graph and sometimes refer to the statements they convey.  For example, we will sometimes refer to ``$i\champ j$'' as an edge in $M_\mathbf{X}$ and sometimes as shorthand that $i$ is a $g$-champion of $j$, and the meaning will be clear from the context.
Futhermore, it is not hard to verify that $i \champ j$ iff $i \champ[\{g\} \mid  \emptyset] j$ and that $i \envies j$ iff $i \champ[\emptyset \mid  \emptyset] j$. Thus,
we can treat every basic edge in $M_\mathbf{X}$ as a generalized champion edge.


\begin{example}
	Consider the instance given in Table \ref{tab:example}, and let $\allocs$ be the partial EFX allocation where $X_1=\{a,b,c\},\, X_2=\{d\},\, X_3=\{e,f\}$. 
	Figure~\ref{fig:example1} depicts the graph $M_{\allocs}$. 
	We haven't drawn all edges; rather, we chose a subset of the edges that illustrate the different types of edges.
	Item $g$ is unallocated in $\allocs$, thus $U=\{g\}$.  
	Since $\{a,b,c\} <_1 \{d\}$, $1 \envies 2$. Moreover, combined with the fact that no one strongly envies $\{d\}$, it also means that $1 \champ 2$. 
	Since $\{d\} <_2 \{b,g\}$ and no one strongly envies $\{b,g\}$, $2 \champ 1$. 
	Similarly, since $\{d\} <_2 \{f,g\}$ and no one strongly envies $\{f,g\}$, $2 \champ 3$, and since $\{e,f\} <_3 \{c,g\}$ and no one strongly envies $\{c,g\}$, $3 \champ 1$. 
	Finally, it holds that $2 \champ[\{a,b\} \mid \{e\}] 3$ since $\{d\} <_2 \{a,b,f\}$ and no one strongly envies $\{a,b,f\}$. 
\end{example}

\begin{table}[]
	\begin{minipage}{0.33\linewidth}
	\label{tab:example}
	\vspace{1em}
	\begin{center}
	\begin{tabular}{|l||c|c|c|c|c|c|c|}
		\hline
		& \multicolumn{1}{l|}{$\boldsymbol{a}$} & \multicolumn{1}{l|}{$\boldsymbol{b}$} & \multicolumn{1}{l|}{$\boldsymbol{c}$} & \multicolumn{1}{l|}{$\boldsymbol{d}$} & \multicolumn{1}{l|}{$\boldsymbol{e}$} & \multicolumn{1}{l|}{$\boldsymbol{f}$} & \multicolumn{1}{l|}{$\boldsymbol{g}$} \\ \hline\hline
		\textbf{1} & 1                        & 2                        & 3                        & 7                        &                          &                          &                          \\ \hline
		\textbf{2} & 1                        & 3                        &                          & 6                        & 1                        & 3                        & 4                        \\ \hline
		\textbf{3} &                          &                          & 5                        &                          & 3                        & 3                        & 2                        \\ \hline
	\end{tabular}
	\caption{A profile with 7 items and 3 additive agents. Unspecified values represent zero values.\label{tab:example}}
	\end{center}
	\end{minipage}\hfill
	\begin{minipage}{0.4\linewidth}
	\begin{center}
		\includegraphics[width=\textwidth]{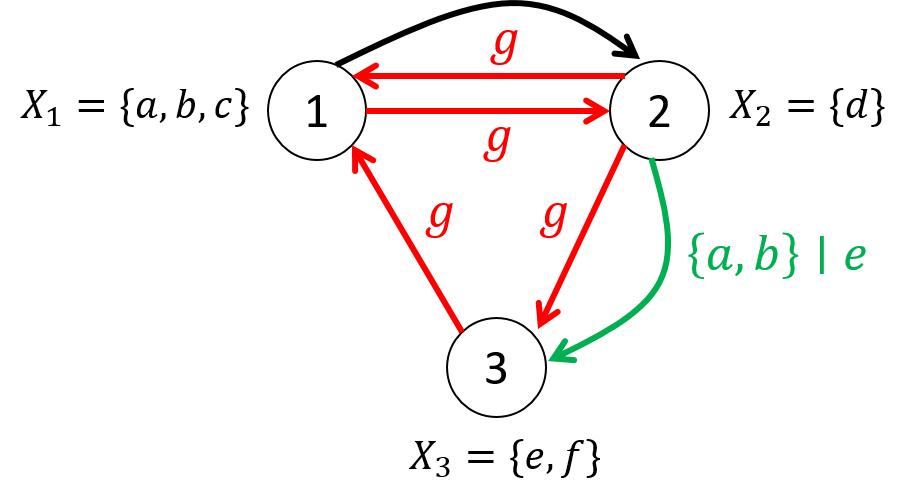}
	\end{center}	
	\end{minipage}\hfill
	\begin{minipage}{0.23\linewidth}
	\vspace{1em}
		\captionof{figure}{The graph $M_{\allocs}$ corresponding to the setting described in Table~\ref{tab:example} and the allocation $X_1=\{a,b,c\},\, X_2=\{d\},\, X_3=\{e,f\}$. Only a subset of the edges is shown.
		\label{fig:example1}}	
	\end{minipage}
\end{table}

Given a cycle $C = a_1 \rightarrow a_2 \rightarrow \cdots \rightarrow a_k \rightarrow a_1$ of edges and an agent $a_i$ in the cycle, $\succ(a_i)$ and $\pred(a_i)$ denote, respectively, the successor and predecessor of $a_i$ along the cycle.

\begin{definition}\label{def:pareto_improvable}
	A cycle $C= a_1 \champ[H_1 \mid  S_1] a_2 \champ[H_2 \mid  S_2] \cdots \champ[H_{k-1} \mid  S_{k-1}] a_k \champ[H_k \mid  S_k] a_1$ in $M_\mathbf{X}$ is called \emph{Pareto improvable} (PI) if
		for every $i,j \in [k]$ we have $H_i \cap H_j = \emptyset$, and
		either $H_i\subseteq U$ or $H_i$ is released by some edge $a_\ell\champ[H_\ell\mid S_\ell] a_{\succ(\ell)}$, for $\ell\in[k]$.
		\footnote{We could have defined a PI cycle more generally, e.g. to to allow the set $H_i$ to be a combination of unallocated goods and items released from several edges.  The proposed definition is hopefully easier to digest and suffices for our purposes.} 
		
		A PI cycle which is composed entirely of basic edges is called a \emph{basic PI cycle}.
		
\end{definition}

By definition, every agent $a_i$ along a PI cycle envies some subset  $A_i \subseteq X_{\succ(i)}$ that no agent strongly envies.  The following simple but \textbf{very (!)} useful lemma 
asserts that reallocating $A_i$ to agent $a_i$ for every $i$ produces a Pareto-dominating EFX allocation.

\shortversion
\begin{lemma}[$\star$]\label{lem:pareto-improvable}
\shortversionend
\fullversion
\begin{lemma}\label{lem:pareto-improvable}
\fullversionend
	If $M_\mathbf{X}$ contains a Pareto improvable cycle, then there exists a (partial) EFX allocation $\mathbf{Y}$ that Pareto dominates $\mathbf{X}$.
	Furthermore, every agent $i$ along the cycle satisfies $Y_i >_i X_i$.
\end{lemma}

\begin{maybeappendix}{lem_pareto-improvable}
\begin{proof}
	Let $C$ be a Pareto-improvable cycle in $M_\mathbf{x}$, and assume w.l.o.g. that $C = 1 \champ[H_1 \mid  S_1] \linebreak[1] 2 \champ[H_2 \mid  S_2] \linebreak[1]\cdots \linebreak[1] \champ[H_{k-1} \mid  S_{k-1}] \linebreak[1] k \champ[H_k \mid  S_k] 1$.
	Define the allocation $\mathbf{Y}$ as follows:  for every agent $i$,
	\[
	Y_i = \begin{cases}
	\left(\left(X_{\succ(i)}\setminus S_i \right)\cup H_i  \right) \setminus \discard[H_i \mid S_i]{i}{\succ(i)}	& i \in [k] \mbox{ (\emph{i.e.}, $i$ on cycle $C$)} \\
	X_i	& \text{otherwise}
	\end{cases}
	\] 
	
	First, note that the assumptions on the sets $H_i,S_i$ ensure that the sets $Y_i$ are disjoint.  That is, $\mathbf{Y}$ is indeed an allocation.
	For every $i \in [k]$, $i$ is the $(H_i \mid S_i)$-champion of $\succ(i)$.
	Thus $X_i <_i Y_i$, by definition of generalized championship.
Since the bundles did not change for agents outside the cycle, we conclude that $\mathbf{Y}$ Pareto-dominates $\mathbf{X}$.
	
	It remains to show that $\mathbf{Y}$ is EFX. 
Since no agent becomes worse off in the transition from $\allocs$ to $\mathbf{Y}$, it suffices to show that in allocation $\allocs$ no agent strongly envies the set $Y_i$ for all $i\in [n]$.
	Indeed, if $i$ is an agent outside the cycle, by the fact that $\allocs$ is EFX, no agent strongly envies $Y_i=X_i$ in $\allocs$.
	If $i$ is an agent in the cycle, no agent strongly envies $Y_i = \left(\left(X_{\succ(i)}\setminus S_i \right)\cup H_i  \right) \setminus \discard[H_i \mid S_i]{i}{\succ(i)} $ in $\mathbf{X}$, by definition of generalized championship.
%
\end{proof}
\end{maybeappendix}


\begin{corollary}[Following \cite{chaudhury2020little}]\label{cor:trivial_cycles}
	If $M_\allocs$ contains an envy-cycle, A self-$g$-champion (an agent $i$ satisfying $i \champ i$) or a cycle composed of envy edges and at most one $h$-champion edge for every $h \in U$, then
	there exists a (partial) EFX allocation $\mathbf{Y}$ that Pareto dominates $\mathbf{X}$.
	Note that these are exactly the basic PI cycles\footnote{Envy cycles, the simplest form of basic PI cycles, were considered in \cite{Lipton04}. Basic PI-cycles were considered in \cite{chaudhury2020little} using different terminology --- championship was only defined in \cite{chaudhury2020efx}.  Our definition of a PI-cycle captures and generalizes these notions.}.
\end{corollary}


\begin{remark}\label{rem:PI_edge_set}
Lemma \ref{lem:pareto-improvable} can be generalized to handle disjoint cycles. The fact that $C$ is a cycle is used in the proof of the lemma only to ensure that every agent
whose bundle is reallocated, is also given an alternative bundle in the new allocation. 
The same is true if $C$ is a set of vertex-disjoint cycles rather than a single cycle. We may then define $C$ as an edge set $\{a_i\champ[H_i\mid S_i] \succ(a_i)\}_{i \in [k]}$, and if the conditions in the definition of a Pareto-improvable cycle are satisfied, then Lemma~\ref{lem:pareto-improvable} still applies.
In this case we refer to $C$ as a \emph{Pareto-improvable edge set}.
\end{remark}




\subsection{New Edge Discovery}
\label{sec:edge-discovery}

In this section we describe a way to discover new generalized champion edges in $M_\allocs$.
These will almost always be of the form $k\champ[S \mid B_j]j$ where $B_j \subseteq X_j $ is some bottom half-bundle
induced by a $g$-decomposer of $j$ (see discussion below Definition \ref{def:champ_g}).
Therefore, to facilitate readability we use the following convention:

\begin{convention}
	For any two agents $k,j$ and any set $S$ disjoint from $X_j$, we write $k \champ[S \mid\circ] j$ (resp. $(S\mid \circ)$-champion) as shorthand for $k \champ[S \mid B_j]j$ (resp. $(S \mid B_j)$-champion), where the half-bundle $B_j$ will be clear from the context. 
\end{convention}

The following structure within the champion-graph is especially convenient for edge discovery.

\begin{definition}
	A cycle $C= a_1 \champ a_2 \champ \cdots \champ a_k \champ a_1$ with at least two $g$-champion edges in $M_\mathbf{X}$ is called a \emph{good $g$-cycle} if:
	\begin{enumerate}
		\item All agents along the cycle are different.
		\item There are no parallel envy edges, i.e., $a_i \nenvies \succ(a_i)$ for all $i$.
		\item There are no internal $g$-champion edges, {\sl i.e.}, for every $i,j \in [k]$, $a_i\champ a_j$ iff $a_j = \succ(a_i)$.
	\end{enumerate}
\end{definition}


\shortversion
\begin{observation}[$\star$]\label{obs:good_cycle_g_decomposition}
\shortversionend
\fullversion
\begin{observation}\label{obs:good_cycle_g_decomposition}
\fullversionend
	Agents $j$ on a good $g$-cycle
	 are $g$-decomposed by $\pred(j)$ into $X_j = T_j \cupdot B_j$.  
\end{observation}

\begin{maybeappendix}{obs_good_cycle_g_decomposition}
\begin{proof}
	This holds by Observation \ref{obs:g_notin_bottom} since $\pred(j) \champ j$ and $\pred(j) \nenvies j$ by definition of a good $g$-cycle.
\end{proof}
\end{maybeappendix}





We next show how to discover new generalized champion edges in the presence of a good $g$-cycle.  The following two observations are useful:

\shortversion
\begin{observation}[$\star$]\label{obs:no_B_j_champ_of_j}
\shortversionend
\fullversion
\begin{observation}\label{obs:no_B_j_champ_of_j}
\fullversionend
	If $i\nenvies j$ then $i\nchamp[B_j \mid \circ ] j$.
\end{observation}

\begin{maybeappendix}{obs_no_B_j_champ_of_j}
\begin{proof}
	Recall that a generalized championship relation $i \champ[H \mid S] j$ is required to satisfy $H \cap S = \emptyset$.
	Here $H = S = B_j$.  Thus, if $B_j \neq \emptyset$ then the requirement clearly does not hold.  Otherwise $B_j = \emptyset$ and $i \champ[\emptyset \mid \emptyset] j$ implies that $i \envies j$, which we assume does not hold.  In any case, the relation $i \champ[B_j \mid \circ] j$ cannot hold.
\end{proof}
\end{maybeappendix}

\shortversion
\begin{observation}[$\star$]\label{obs:pred(i)_nchamp_B(i)}
\shortversionend
\fullversion
\begin{observation}\label{obs:pred(i)_nchamp_B(i)}
\fullversionend
	For any two agents $i,j$ along a good $g$-cycle, $\pred(i)\nchamp[B_i\mid \circ] j$.	
\end{observation}

\begin{maybeappendix}{obs_pred(i)_nchamp_B(i)}
\begin{proof}
	If $i = j$ the statement holds by Observation \ref{obs:no_B_j_champ_of_j}.  Assume otherwise.  By definition of a good $g$-cycle, $\pred(i)\nenvies i$ and $\pred(i)\champ i$. Therefore,
	$
		X_{\pred(i)} >_{\pred(i)} T_i\cup B_i>_{\pred(i)} T_j \cup B_i,
	$ 
where the second inequality is by Observation~\ref{obs:T_k<T_j} and \cancbty. Thus, $\pred(i)$ does not envy $T_j \cup B_i$ and the claim follows.
\end{proof}
\end{maybeappendix}

\begin{lemma}\label{lem:alg_start}
	Let $C$ be a good $g$-cycle.  For any agent $i$ along the cycle, there exists an agent $a$ such that $a \champ[B_i \mid\circ] \succ(i)$.
\end{lemma}

\begin{proof}
	$C$ is a good cycle,
	hence $i$ $g$-decomposes ${\succ(i)}$ into $X_{\succ(i)} = T_{\succ(i)} \cup B_{\succ(i)}$.  Furthermore, $i \nchamp i$ and $i$ is $g$-decomposed into $X_i = T_i \cup B_i$.  Thus by Observation \ref{obs:T_k<T_j}
	(with $j = \succ(i)$ and $k = i$ in the Observation statement)
	we have $T_{i} <_i T_{\succ(i)}$.  Hence, by \cancbty,
	$ X_i = T_i \cup B_i <_i T_{\succ(i)} \cup B_i.$
	Since the set $T_{\succ(i)} \cup B_i$ is envied by some agent, it must have a most envious agent and the claim follows. 
\end{proof}

\begin{lemma}\label{lem:alg_step}
	Let $C$ be a good $g$-cycle.
	Let $i,j,k$ be agents along the cycle.
	If $k \champ[B_i \mid\circ]j$, then there exists an agent $a$ (not necessarily in the cycle) such that $a \champ[B_i \mid\circ]\succ(k)$.
\end{lemma}

\begin{proof}
	If $k = \pred(j)$, then $j = \succ(k)$ and we are done (take $a = k$). Assume otherwise.
	$C$ is a good cycle,
	hence $k$ $g$-decomposes ${\succ(k)}$ into $X_{\succ(k)} = T_{\succ(k)} \cup B_{\succ(k)}$.  Furthermore, $k \nchamp j$ since $k \neq \pred(j)$, and $j$ is $g$-decomposed into $X_j = T_j \cup B_j$.  Thus by Observation \ref{obs:T_k<T_j} we have $T_{j} <_i T_{\succ(k)}$.  Hence,
	$ X_k <_k T_j \cup B_i <_k T_{\succ(k)} \cup B_i, $ 
	where the first inequality holds by $k \champ[B_i \mid\circ]j$ and the second by \cancbty.  Since the set $T_{\succ(k)} \cup B_i$ is envied by some agent, it must have a most envious agent and the claim follows.  
\end{proof}

%



For every bottom half-bundle $B_i$ along a good $g$-cycle $C$, applying Lemma \ref{lem:alg_start} provides an initial $B_i$-champion edge.  If this edge is internal to the cycle, i.e., the source of the edge is in the cycle, then we can apply Lemma \ref{lem:alg_step} to discover a new $B_i$-champion edge. Once again, if the new edge is internal to the cycle, then we can reapply Lemma \ref{lem:alg_step}.  We can repeat this process to discover more and more $B_i$-champion edges, until either the new edge has already been previously discovered, or it is external (i.e., its source is outside the cycle).  

There are two particular types of internal $B_i$-champions edges that are useful to us.

\begin{definition}\label{def:good_edge} 
	Let $C$ be a good $g$-cycle. 
Let $i,j,k$ be three agents along $C$. If $i\champ[B_k\mid\circ]j$
and $k$ is on the path from $j$ to $i$ in $C$, then we say that the edge $i\champ[B_k \mid\circ]j$ is a \emph{good edge} (or good $B_k$-edge).
If $\ell\champ[B_k \mid\circ]j$ for some agent $\ell$ outside the cycle $C$, then we say that the edge $\ell\champ[B_k \mid\circ]j$ is an \emph{external edge} (or external $B_k$-edge).
\end{definition}


\shortversion
\begin{wrapfigure}[6]{r}{0.3\textwidth}
	\vspace{-22pt}
	\begin{center}
		\includegraphics[width=0.3\textwidth]{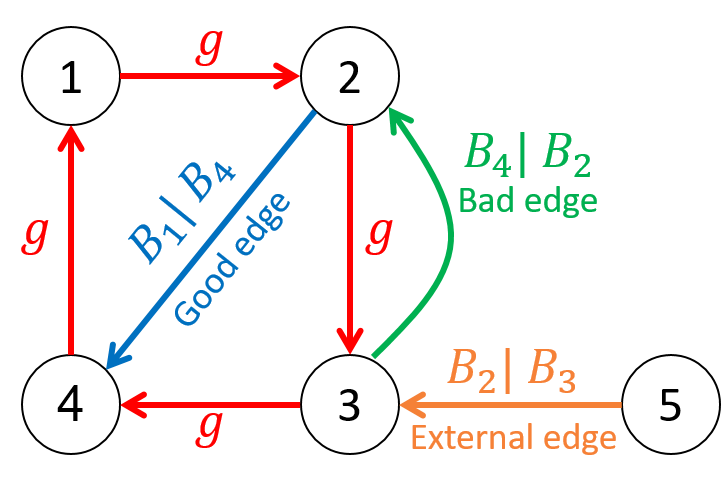}
	\end{center}
\end{wrapfigure}
\shortversionend

\fullversion
\begin{wrapfigure}[6]{r}{0.3\textwidth}
	\vspace{-30pt}
	\begin{center}
		\includegraphics[width=0.3\textwidth]{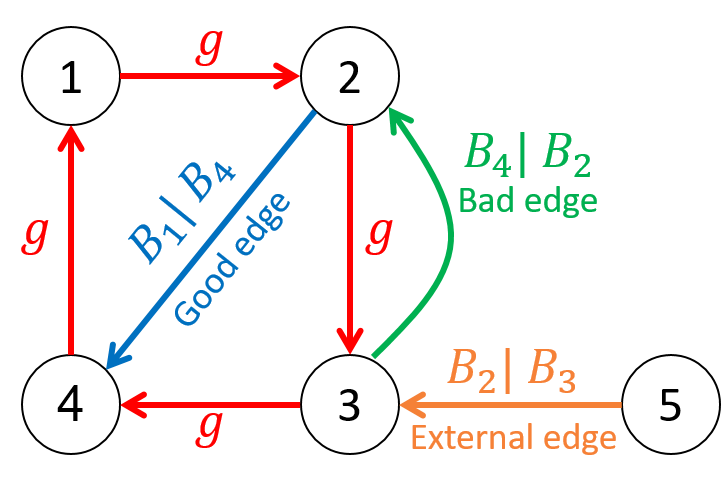}
	\end{center}
\end{wrapfigure}
\fullversionend


The figure on the right
illustrates Definition \ref{def:good_edge}. 
The red edges form a good $g$-cycle $C$ among 4 agents, $C = 1\champ 2 \champ 3 \champ 4 \champ 1$.
The edge $2\champ[B_1 \mid B_4]4$ is a \emph{good edge}, since 1 is on the path from $4$ to $2$ in $C$.
On the other hand, $3\champ[B_4 \mid B_2]2$ is not a good edge (we call it a bad edge in the figure), since 4 is not on the path from $2$ to $3$ along $C$. $5 \champ[B_2 \mid B_3] 3$ is an external edge.

\begin{theorem}
	\label{thm:good-edge-or-external}
	Let $C$ be a good $g$-cycle.  For every agent $j$ along the cycle, there exists either a good $B_j$-edge in $C$, or an external $B_j$-edge in $C$.
\end{theorem}

\begin{proof}
	Assume without loss of generality that $C=  1\champ 2 \champ \cdots \champ k \champ 1$ and $j=1$, i.e., we try to find $B_1$-champion edges.	
	By Lemma~\ref{lem:alg_start} there exists an edge $\ell_1 \champ[B_1\mid\circ] 2$ for some agent $\ell_1$. If this is an external $B_1$-edge we are done. Otherwise, $\ell_1$ is an agent along $C$, and thus by Lemma~\ref{lem:alg_step} there exists an edge $\ell_2 \champ[B_1\mid\circ] \succ(\ell_1)$. 
	As long as the result of Lemma~\ref{lem:alg_step} is not an external edge we may apply the lemma repeatedly. Hence, if no such iteration results in an external edge, we  obtain a  sequence of agents $(\ell_i)_{i=1}^\infty$ such that $\ell_{i+1} \champ[B_1\mid\circ] \succ(\ell_i)$ for every $i\geq 1$.
	
	If for some $i\geq1$, we have $\ell_{i+1} \leq \ell_i$ then the edge $\ell_{i+1} \champ[B_1\mid\circ] \succ(\ell_i)$ is a good edge (since the path from $\succ(\ell_i)$ to $\ell_{i+1} $ includes 1). Hence, if none of these edges are good, then $\ell_i < \ell_{i+1}$ for every $i\geq1$, in contradiction to $C$ being of finite length. Thus, one of these edges must be good, hence we are done.
\end{proof}

\shortversion
\begin{wrapfigure}[3]{r}{0.2\textwidth}
	\vspace{-24pt}
	\begin{center}
		\includegraphics[scale=\figurescale]{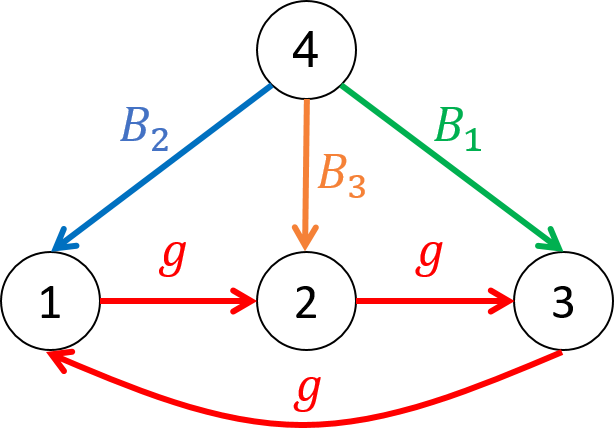}
	\end{center}
\end{wrapfigure}
\shortversionend
%
The following observation and its corollary allow us to narrow down the possible configurations of $B_j$-edges obtained from Theorem
\ref{thm:good-edge-or-external}.
\shortversion
\begin{observation}[$\star$]\label{obs:bottom_bundle_ineq}
\shortversionend
\fullversion
\begin{observation}\label{obs:bottom_bundle_ineq}
\fullversionend
	If $i\champ[B_j\mid\circ] k$ and $i\nenvies k$ then $B_k <_i B_j$.
\end{observation}
\begin{maybeappendix}{obs_bottom_bundle_ineq}
\begin{proof}
	Since $i\nenvies k$ and $i\champ[B_j\mid\circ] k$ we have
	$
		T_k\cup B_k = X_k \leq_i X_i <_i (X_k\setminus B_k)\cup B_j = T_k \cup B_j,
	$
	and by \cancbty\ this implies $B_k <_i B_j$.
\end{proof}
\end{maybeappendix}

\shortversion
\begin{corollary}[$\star$]\label{cor:no_single_external_source}	
\shortversionend
\fullversion
\begin{corollary}\label{cor:no_single_external_source}	
\fullversionend
	Let $C$ be a good $g$-cycle. 
	Consider the set of $B_j$-edges guaranteed by Theorem \ref{thm:good-edge-or-external} for every agent $j$ along the cycle.
	If all these edges are external, then they cannot all share the same source agent, unless that agent envies some agent along the cycle 
\fullversion
	(the figure below demonstrates an impossible configuration).
\fullversionend
\shortversion
	(the figure on the right demonstrates an impossible configuration).
\shortversionend
\end{corollary}
%
\fullversion
	\begin{center}
		\includegraphics[scale=\figurescale]{figures/external}
	\end{center}
\fullversionend

\begin{maybeappendix}{cor_no_single_external_source}
\begin{proof}
	Assume towards contradiction that there is some agent $a$ which is the source of all $B_i$-edges given by Theorem \ref{thm:good-edge-or-external}, and $a$ does not envy any agent along the cycle. Let $j=\arg\min_i\{v_a(B_i)\mid i \text{ is an agent along } C\}$. By assumption, there exists some $j'$ along $C$ such that $a \champ[B_j \mid\circ] j'$, and $a \nenvies j'$.  By Observation \ref{obs:bottom_bundle_ineq}, $B_j >_a B_{j'}$. Hence, $v_a(B_j)>v_a(B_{j'})$, in contradiction to the definition of $j$.
\end{proof}
\end{maybeappendix}

\section{EFX with at most $n-2$ unallocated goods}
\label{sec:efxnm2}

In this section we prove the following: 

\begin{theorem}
	\label{thm:n-2-charity}
	For every profile of $n$ additive valuations (and more generally, \valclass\ valuations), there exists an EFX allocation $\allocs$ with at most $n-2$ unallocated items.  Moreover, in $\allocs$ no agent envies the set of unallocated items.
\end{theorem}

The following lemma shows that the basic edges of the graph $M_{\allocs}$ follow a very particular structure. This lemma is used in the proof of Theorem~\ref{thm:n-2-charity}.

\begin{figure}
	\begin{center}
	\shortversion
		\includegraphics[width=0.65\textwidth]{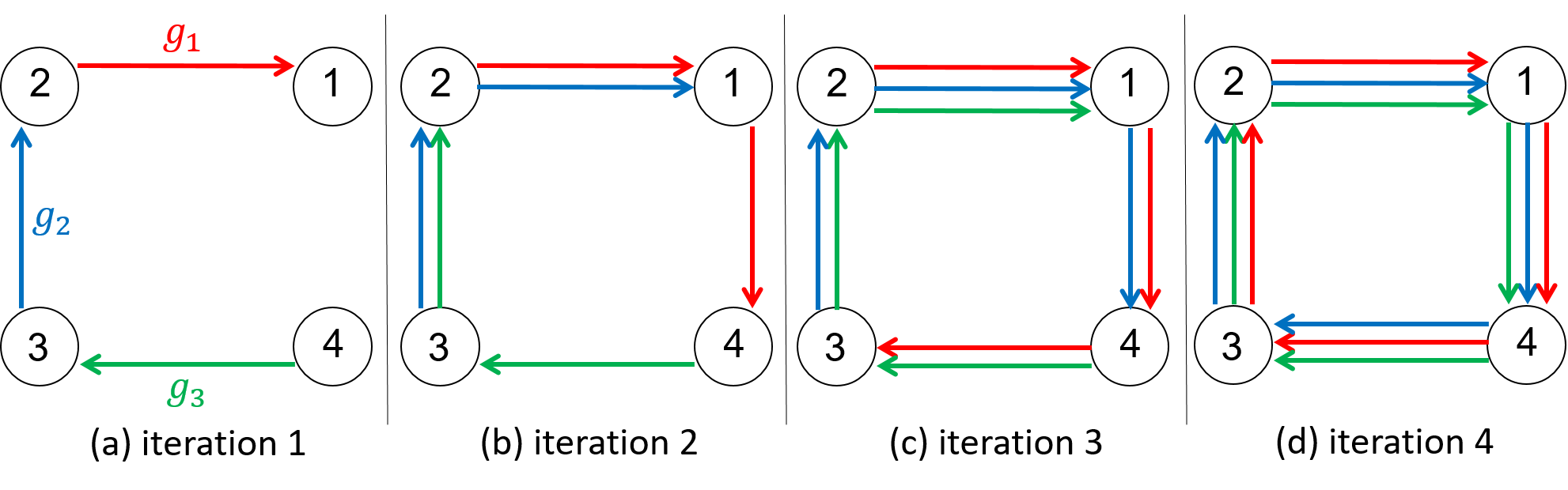}
	\shortversionend
	\fullversion
		\includegraphics[width=0.75\textwidth]{figures/parallel-ring.png}
	\fullversionend
		\caption{An illustration of the proof of Lemma~\ref{lem:parallel-rings} for $n=4$, $k=3$. Figures $(a)-(d)$ represent the 4 iterations in the proof.
		Red, blue and green edges are champion edges w.r.t. unallocated goods $g_1, g_2, g_3$, respectively. 
		\label{fig:parallel-ring}}
	\end{center}
\end{figure}

\begin{lemma}
	\label{lem:parallel-rings}
	Let $\allocs$ be a partial EFX allocation with at least $n-1$ unallocated items, and let $G$ be $M_{\allocs}$ restricted to basic edges, {\sl i.e.}, envy edges and basic champion edges as per Definition \ref{def:champ_g}. 
	
	If $G$ does not admit a basic PI cycle, then the number of unallocated goods is exactly $n-1$, and $G$ is a union of $n-1$ parallel
	Hamiltonian cycles. 
	Every such cycle consists of $g$-champion edges for some unallocated item $g$ (see Figure~\ref{fig:parallel-ring} (d) for the case of $n=4$). 
\end{lemma}

\begin{proof}
	Recall that $U$ denotes the set of unallocated items, and let $U=\{g_1, \ldots, g_k\}$ for some $k \geq n-1$. 
	Let $e_{1}^{1}$ be an arbitrary incoming $g_{1}$-champion edge of agent 1 (such an edge exists by Observation~\ref{obs:exists-champion}). 
	
	If the source of this edge is agent 1 we are done (we have a self champion). 
	Assume w.l.o.g. that the source of $e_{1}^{1}$ is agent 2, and consider its incoming $g_{2}$-champion edge, denoted $e_{2}^{2}$.
	If the source of $e_{2}^{2}$ is agent 2 or agent 1, we are done
	(in the first case we have a self $g_{2}$-champion, and in the
	second case we have a basic PI cycle of size 2).
	Assume w.l.o.g. that the
	source of $e_{2}^{2}$ is agent 3. We can continue this way to conclude
	that w.l.o.g we have a directed path $n\champ[g_{n-1}] n-1\champ[g_{n-2}]\cdots\champ[g_1] 1$ (see Figure~\ref{fig:parallel-ring} (a)). 
	If $k \geq n$, then consider the incoming $g_n$-champion edge of agent $n$.
	No matter what the source of this edge is, it must close a basic PI cycle.
	
	Ergo, $k=n-1$. Consider the incoming $g_{1}$-champion edge of agent $n$, denoted $e_{n}^{1}$.
	The source of this edge must be agent 1 (every other option closes a basic PI cycle). Now consider the incoming  $g_{2}$-champion edge of agent 1, $e_{1}^{2}$. Similarly, the source of
	this edge must be agent 2. We can again continue this way until we
	get to the incoming $g_{n-1}$-champion edge of agent $n-2$, denoted
	$e_{n-2}^{n-1}$, and conclude that its source is agent $n-1$ (see Figure~\ref{fig:parallel-ring} (b)). Now
	we consider the incoming $g_{1}$-champion edge of agent $n-1$, and
	continue with the same reasoning, to finally conclude that the source
	of the incoming $g_{n-1}$-champion edge of agent $n$ is agent 1 (see Figure~\ref{fig:parallel-ring} (d)).
	
	At this point there is a Hamiltonian cycle
	$n\champ[g_{i}] n-1\champ[g_{i}]\cdots\champ[g_{i}] 1\champ[g_{i}] n$
	for every $i$. Any other basic champion or envy edge must close a basic PI cycle
	and thus does not exist by assumption.
	This concludes the second part of the lemma.
\end{proof}

We are now ready to prove Theorem~\ref{thm:n-2-charity}.

\begin{proof}
By Lemma~\ref{lem:dominate-implies-progress}, it suffices to prove that if $\allocs$ has at least $n-1$ unallocated items, then there exists a Pareto dominating EFX allocation $\mathbf{Y}$.
By Lemma~\ref{lem:pareto-improvable}, it suffices to find a PI cycle in $M_{\allocs}$.

\shortversion
By Lemma~\ref{lem:parallel-rings}, there are exactly $n-1$ unallocated goods, and $G$ is a union of $n-1$ parallel Hamiltonian cycles $1 \champ[g_i] 2 \champ[g_i] \cdots n   \champ[g_i] 1$, one for each unallocated good $g_1, \ldots ,g_{n-1}$.
\shortversionend

\fullversion
By Lemma~\ref{lem:parallel-rings}, there are exactly $n-1$ unallocated goods, and $G$ is a union of $n-1$ parallel Hamiltonian cycles, one for each  unallocated good $g_1, \ldots, g_{n-1}$:
\begin{align*}
    1 &\champ[g_1] 2 \champ[g_1] \cdots &\cdots {n}   \champ[g_1]& 1 \\
    1 &\champ[g_2] 2 \champ[g_2]\cdots&\cdots {n} \champ[g_2]  & 1 \\
    \vdots &&      &\vdots \\
    1 &\champ[g_{n-1}] 2 \champ[g_{n-1}]\cdots& \cdots {n} \champ[g_{n-1}]  & 1
    \end{align*}
\fullversionend

\fullversion
\begin{figure}[t]
	\begin{center}
		\includegraphics[width=0.2\textwidth]{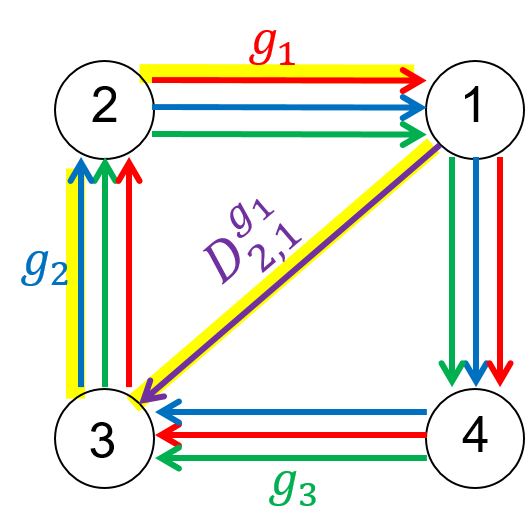}
		\caption{An illustration of the proof of Theorem~\ref{thm:n-2-charity} for the case of $n=4$. The Pareto-improvable cycle is highlighted in yellow. 
			\label{fig:p-improvable-cycle}}
	\end{center}
\end{figure}
\fullversionend

Furthermore, it follows from  Lemma~\ref{lem:parallel-rings} that we can assume that we have no other champion edges and no envy edges. Thus, all these cycles are good.

By Theorem~\ref{thm:good-edge-or-external} for all agents $k$ and every unallocated good $g$ there exists a good $\discard[g]{k-1}{k}$-champion edge, from some agent $j$ to some agent $j'$. (Indeed, since $G$ is a union of parallel Hamiltonian cycles, no external edges exists.) 
Choose some arbitrary such agent $k$ and unallocated good $g$ and corresponding agents $j, j'$, and let $Z=\discard[g]{k-1}{k}$.
We show that $j \champ[Z \mid\circ] j'$ closes a Pareto improvable cycle in $M_{\allocs}$.

By definition of a good edge, there is a unique path $P$ consisting of $g$-champion edges from agent $j'$ to $j$, passing through agent $k$. 
By Lemma~\ref{lem:parallel-rings}, for every $q\in U$, every edge in $P$ has a parallel $q$-champion edge. 
Note that $P$ has at most $n-1$ edges.

Rename the agents so that agent $j$ is agent $|P|$, agent $j'$ is agent $1$, let $k'$ be the new index for agent $k$, and let $P=1,2,\ldots, {k'-1},{k'}, \ldots {|P|}$. 
Note that $Z$ is now the discard set of the champion edge ${k'-1}\champ {k'}$.
Let $U'=U\setminus \{g\}$ ($|U'|=n-2$), and rename the items in $U'$ where $U'=\{r_1, r_2, \ldots, r_{k'-2}, r_{k'},\ldots,r_{n-1}\}$.
We now describe a Pareto-improvable cycle:
\shortversion
\begin{align*}
    1 \champ[r_1] 2 \champ[r_2] 2 \champ[r_3] \cdots  \champ[r_{k'-2}] {k'-1}\champ {k'} \champ[r_{k'}]  {k'+1}\champ[r_{k'+1}] \cdots \champ[r_{|P|-1}] {|P|}  \champ[Z \mid\circ] 1
\end{align*}
\shortversionend
\fullversion
\begin{align*}
    1 \champ[r_1] 2 \champ[r_2] 2 \champ[r_3] \cdots  \champ[r_{k'-2}] {k'-1}\champ {k'} \champ[r_{k'}]  {k'+1}\champ[r_{k'+1}] \cdots&\\ \cdots \champ[r_{|P|-1}] {|P|}  \champ[Z \mid\circ] 1
\end{align*}
\fullversionend
Since $Z$ is the discard set of the champion edge ${k'-1}\champ {k'}$, it is discarded along the cycle. 
All other edges in the cycle are with respect to distinct unallocated goods. 
Therefore, this is a Pareto-improvable cycle%
\fullversion
 (see Figure~\ref{fig:p-improvable-cycle} for an illustration)
\fullversionend%
.
\end{proof}

%
\section {EFX for 4 additive agents with 1 unallocated good}
\label{sec:4-agents}

\shortversion
\begin{figure}
	\label{fig:proof-structure}
	\begin{center}
		\includegraphics[width=0.8\linewidth]{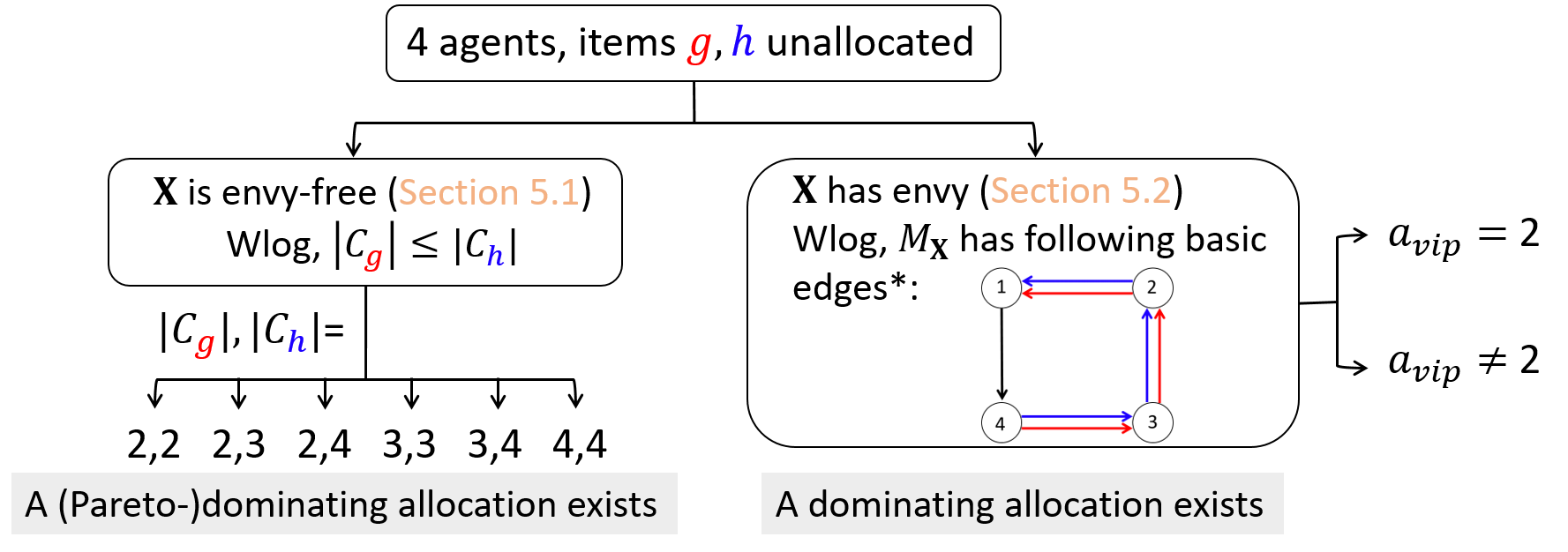}
	\end{center}
	\caption{High-level roadmap of the proof of Theorem \ref{thm:4_agent_main_result}. $|C_{\g}|$ and $|C_{\h}|$ are the lengths of some good cycles of items $\g$ and $\h$, respectively. ($*$)  The $\g$ and $\h$ champions of agent 4 can be either agent 1 or 2.\label{fig:proof-structure}}
\end{figure}
\shortversionend

In this section we prove our main result, namely that every setting with 4 additive agents admits an EFX allocation with at most one unallocated good.  We prove this theorem for the class of nice cancelable valuations which contains additive.  By Lemma \ref{lem:dominate-implies-progress} it suffices to prove:
\begin{theorem}\label{thm:4_agent_main_result}
		Let $\allocs$ be an EFX allocation on 4 agents with \valclass\ valuations, with at least two unallocated items.
		Then, there exists an EFX allocation $\mathbf{Y}$ that dominates $\allocs$.
\end{theorem}

The proof involves a rigorous case analysis, which exemplifies the extensive use of our new techniques.
We have attempted to make the proofs as accessible as possible through the use of extensive aids such as figures and colors.

By assumption, there exist two unallocated goods which we denote $\g$, $\h$.
The proof distinguishes between two main cases, namely whether $\allocs$ is envy-free (Section~\ref{sec:ef-case}) or not (Section~\ref{sec:not-ef-case}).
When $\allocs$ is envy-free, we show that a Pareto improvable (PI) cycle always exists.
This is shown via a case analysis that depends on the lengths of the good $\g$- and $\h$-cycles which must exist in the champion graph $M_\allocs$.

When $\allocs$ has envy, we argue that the only interesting case is where the basic edges follow some specific structure,
modulo permuting the agent identities (otherwise, a PI cycle exists).
Then, we show that there is an EFX allocation in which agent $\vip$ (per the lexicographic potential) is better off relative to $\allocs$.
Since $\vip$ could be any one of the agents (due to the identity permutation),
the proof splits to cases accordingly, where the case $\vip = 2$ is treated separately from the case where $\vip$ is one of the other three agents.
Our approach here is heavily inspired by and follows a similar high-level structure to that of~\cite{chaudhury2020efx} in their analysis of the envy case in their 3 agent result.
Our proof structure is depicted in Figure~\ref{fig:proof-structure}.


\fullversion
\begin{figure}
	\label{fig:proof-structure}
	\begin{center}
			\includegraphics[width=0.9\linewidth]{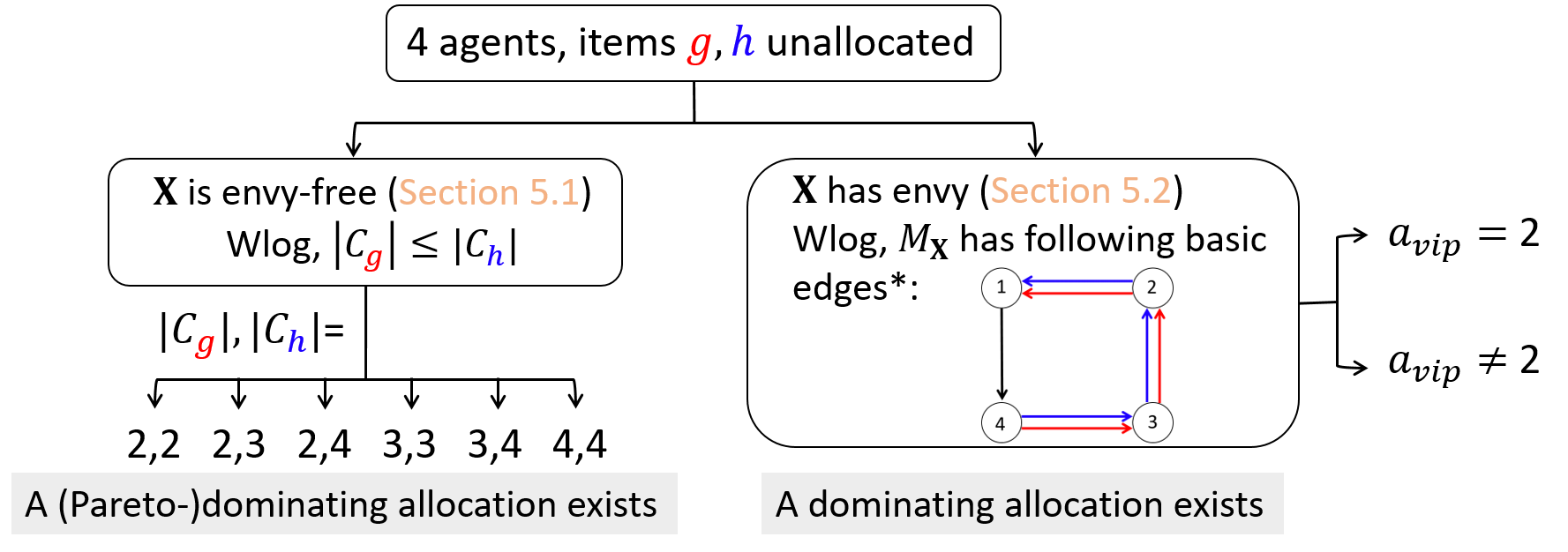}
	\end{center}
	\caption{High-level roadmap of the proof of Theorem \ref{thm:4_agent_main_result}. $|C_{\g}|$ and $|C_{\h}|$ are the lengths of some good cycles of items $\g$ and $\h$, respectively. ($*$)  The $\g$ and $\h$ champions of agent 4 can be either agent 1 or 2.\label{fig:proof-structure}}
\end{figure}
\fullversionend

\subsection{$\allocs$ is Envy-Free}
\label{sec:ef-case}

In this section we show that if $\allocs$ is envy-free, then we can always find a PI cycle or edge set in $M_\allocs$ (see Remark \ref{rem:PI_edge_set}), implying (by Lemma \ref{lem:pareto-improvable}) the existence of a Pareto-dominating EFX allocation $\mathbf{Y}$.

Recall that every agent $i$ has an incoming $\g$-champion edge and an incoming $\h$-champion edge (Observation \ref{obs:exists-champion}), and thus $M_\allocs$ contains a  $\g$-cycle and an $\h$-cycle.  If there is a self $\g$ or $\h$ champion we are done by Corollary \ref{cor:trivial_cycles}. Thus, these cycles are of size at least 2, and contain no envy edges, and are therefore good cycles. Denote them by $C_\g, C_\h$, respectively.

\fullversion
Note that if one of these has size 1 (that is, there is a self $\g$ or $\h$-champion) then we are done by Corollary \ref{cor:trivial_cycles}, thus we assume that $\left|C_\g\right| , \left|C_\h\right| \in \{2,3,4\}$.
\fullversionend

By Observation \ref{obs:good_cycle_g_decomposition}, $C_\g$ (resp. $C_\h$) induces a $\g$ (resp. $\h$)-decomposition of $X_j$ for any agent $j$ in the cycle.  In the following we denote the $\g$ (resp. $\h$) -decomposition by $X_j = T^{\g}_j\cupdot B^{\g}_j$ (resp. $X_j = T^{\h}_j\cupdot B^{\h}_j$).
We shall make repeated use of Observations \ref{obs:no_B_j_champ_of_j} and \ref{obs:pred(i)_nchamp_B(i)}. For concise presentation, we write here the implications that will be repeatedly used in this section: For every two agents $i,j$ we have:
\shortversion
\textsl{(a)} $i \nchamp[B_j^{\g} \mid\circ] j$ ; \textsl{(b)} if $i,j$ reside on the same good $\g$-cycle, then $\pred(i) \nchamp[B_i^{\g} \mid\circ] j$.
\shortversionend
\fullversion
\begin{itemize}
	\item	$i \nchamp[B_j^{\g} \mid\circ] j$
	\item	If $i$ and $j$ reside on the same good $\g$-cycle, then $\pred(i) \nchamp[B_i^{\g} \mid\circ] j$
\end{itemize}
\fullversionend
Analogous claims hold for $\h$. We remind the reader that $i \champ[B_k^{\g} \mid\circ] j$, $i \champ[B_k^{\h} \mid\circ] j$ are shorthand for $i \champ[B_k^{\g} \mid B_j^{\g} ] j$, $i \champ[B_k^{\h} \mid B_j^{\h}] j$, respectively.


Recall that each of $C_\g$, $C_\h$ can be of size 2, 3, or 4. Assume w.l.o.g. that $|C_\g|\leq |C_\h|$. Thus, there are six cases to consider.
\shortversion
Some cases have been deferred to Appendix \ref{apx:4-agents-EF} due to space limitations.
\shortversionend

\fullversion
\vspace{1ex}
\fullversionend
\noindent\textbf{Case 1: $\boldsymbol{\left|C_\g\right| = \left|C_\h\right| = 2}$.}
Assume w.l.o.g. that $C_\g = 1 \champ[\g] 2 \champ[\g] 1$.  By Theorem \ref{thm:good-edge-or-external}, there exists either a good or external $\Bgone$-edge going into agent 2, and there exists either a good or external $\Bgtwo$-edge going into agent 1.  If one of these is good we are done:  for example, the only possible good $\Bgone$-edge is $1 \champ[\Bgone \mid\circ] 2$ which closes the PI cycle $1 \champ[\Bgone\mid\circ] 2 \champ[\g] 1$ (recall that the edge $2 \champ[\g] 1$ releases $\Bgone$).

Thus both edges have to be external, {\sl i.e.}, their sources are agents 3 or 4.  Assume w.l.o.g. that $3 \champ[\Bgone \mid\circ] 2$.  We cannot also have $3 \champ[\Bgtwo \mid\circ] 1$
by Corollary \ref{cor:no_single_external_source}.
We conclude that $4 \champ[\Bgtwo \mid\circ] 1$, and we have the following structure:\vspace{-1em}
\begin{center}
	\includegraphics[scale=\figurescale]{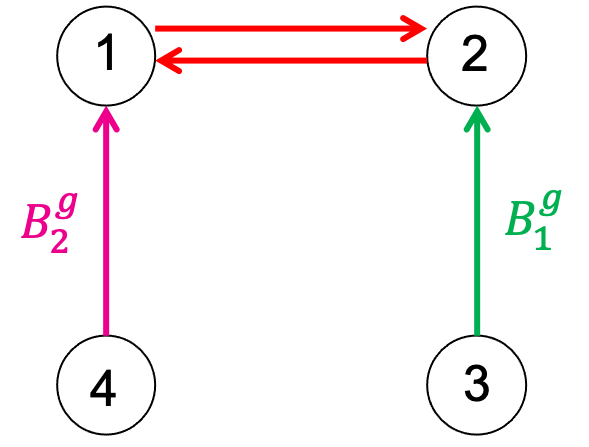}
\end{center}

Consider $C_\h$.  If $C_\h = 1 \champ[\h] 2 \champ[\h] 1$ then we are done since we get the PI cycle $1 \champ[\g] 2 \champ[\h] 1$ (see Corollary \ref{cor:trivial_cycles}).  If $C_\h = 3 \champ[\h] 4 \champ[\h] 3$, then following the analogous reasoning for $C_\g$ we can assume that we have external $\Bhthree$ and $\Bhfour$ edges, in which case we have one of the following two structures (highlighted edges are part of PI cycles or PI edge sets):
\begin{center}
	\includegraphics[scale=\figurescale]{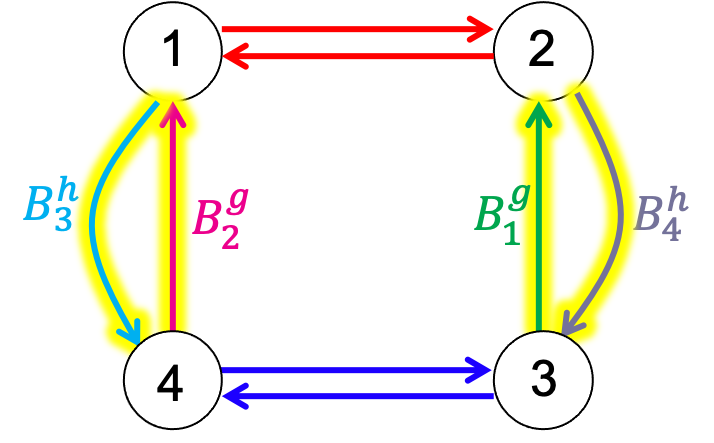} \qquad\qquad
	\includegraphics[scale=\figurescale]{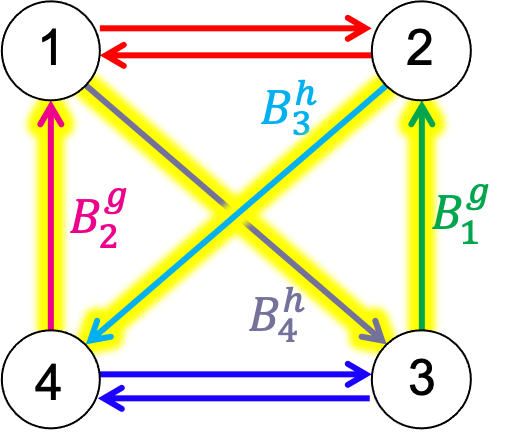}
\end{center}

In the left graph, the two cycles $1 \champ[\Bhthree \mid\circ] 4 \champ[\Bgtwo \mid\circ] 1$, $2 \champ[\Bhfour \mid\circ] 3 \champ[\Bgone \mid\circ] 2$ form a PI edge set%
, and in the right graph we have the PI-cycle $1 \champ[\Bhfour \mid\circ]\linebreak[1] 3 \champ[\Bgone \mid\circ]\linebreak[1] 2 \champ[\Bhthree \mid\circ]\linebreak[1] 4 \champ[\Bgtwo \mid\circ]\linebreak[1] 1$ (note that in both cases every one of the four used bottom half-bundles is released by a corresponding incoming edge to each one of the agents), and thus we are done.  We remark that if $M_\allocs$ contains the good $\g$-cycle $3 \champ[\g] 4 \champ[\g] 3$, then similar reasoning also shows that we have a PI cycle (or PI edge set).  In other words,
if $M_\allocs$ contains two disjoint $\g$-cycles of size 2 or two disjoint $\h$-cycles of size 2, then we are done.


It remains to consider the case where $C_\h$ intersects $C_\g$ at exactly one agent.
If $C_\h = 1 \champ[\h] 3 \champ[\h] 1$, then we get the PI cycle $1 \champ[\h] 3 \champ[\Bgone \mid\circ] 2 \champ[\g] 1$.  The case $C_\h = 2 \champ[\h] 4 \champ[\h] 2$ is symmetric.  We are left with the cases $C_\h = 2 \champ[\h] 3 \champ[\h] 2$ or $C_\h = 1 \champ[\h] 4 \champ[\h] 1$, which are also symmetric. Thus, w.l.o.g. we assume $C_\h = 2 \champ[\h] 3 \champ[\h] 2$ and we get the structure:
\begin{center}
	\includegraphics[scale=\figurescale]{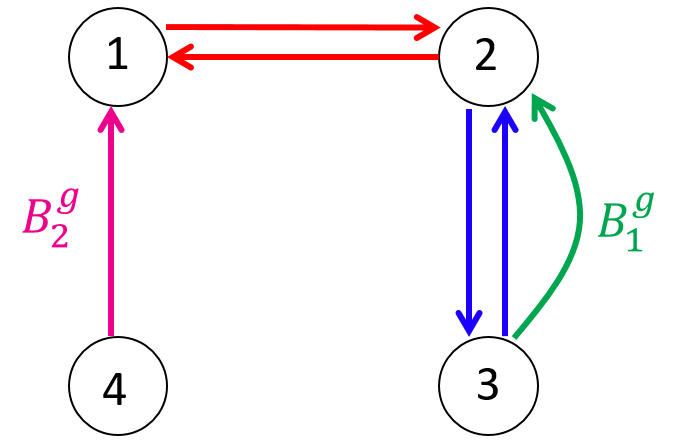}
\end{center}

As before, there must be two external $\Bhtwo$ and $\Bhthree$ edges (coming out of agents 1 and 4).  If $1 \champ[\Bhtwo \mid\circ] 3$ and $4 \champ[\Bhthree \mid\circ] 2$,
then we get the PI cycle $1 \champ[\Bhtwo \mid\circ] 3 \champ[\h] 2 \champ[\g] 1$.  Thus $1 \champ[\Bhthree \mid\circ] 2$ and $4 \champ[\Bhtwo \mid\circ] 3$, and we get the structure:
\shortversion
\vspace{-1em}
\shortversionend
\begin{center}
	\includegraphics[scale=\figurescale]{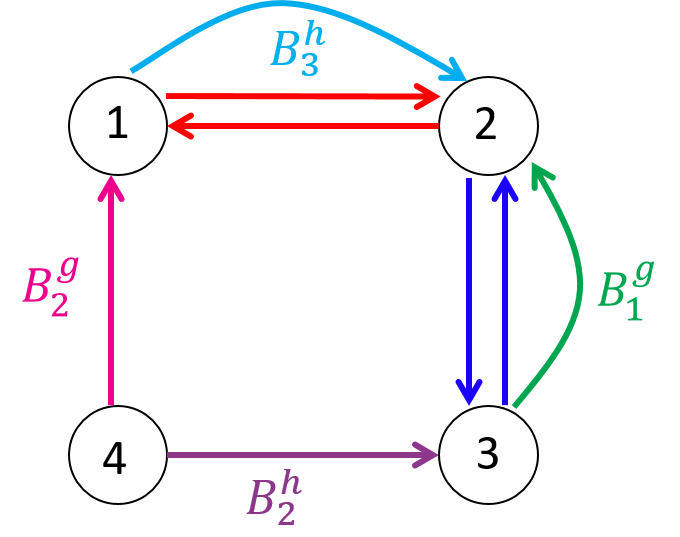}
\end{center}


We now ask who is an $\h$-champion of agent 4 (such exists by Observation \ref{obs:exists-champion}).  If $3 \champ[\h] 4$, we are done via the PI edge set $1 \champ[\Bhthree \mid\circ] 2 \champ[\g] 1$, $3 \champ[\h] 4 \champ[\Bhtwo \mid\circ] 3$.  If $2 \champ[\h] 4$, we are done: $1 \champ[\g] 2 \champ[\h] 4 \champ[\Bgtwo \mid\circ] 1$.  Thus, assume $1 \champ[\h] 4$ (recall that there are no self-champions, {\sl i.e.}, $4 \nchamp[\h] 4$).

Now we ask who is an $\h$-champion of agent 1.  If $4 \champ[\h] 1$ we are done since we have two disjoint size 2 $\h$-cycles, a situation we have already dealt with at the start of Case 1.  If $2 \champ[\h] 1$ we are done: $1 \champ[\g] 2 \champ[\h] 1$.  Therefore, $3 \champ[\h] 1$ and we have the structure:
\begin{center}
	\includegraphics[scale=\figurescale]{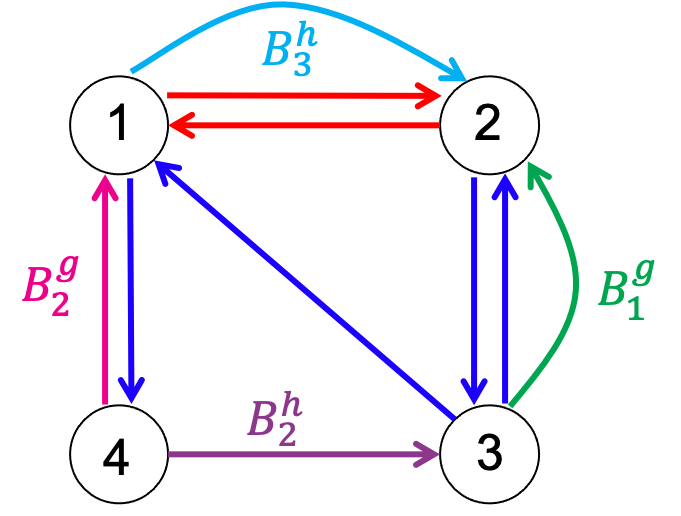}
\end{center}

Now we ask who is a $\g$-champion of 3.  If $1 \champ[\g] 3$ or $2 \champ[\g] 3$, we have a size 2 cycle with $\g$ and $\h$ edges and we are done.  Thus, $4 \champ[\g] 3$.

Finally, we ask who is a $\g$-champion of 4.  If $3 \champ[\g] 4$ then we have two disjoint $\g$-cycles of size 2 and we are done.  If $2 \champ[\g] 4$, then we are done: $2\champ[\g] 4 \champ[\Bhtwo \mid\circ] 3 \champ[\h] 2$.  Thus assume $1 \champ[\g] 4$, and we have the structure:
\shortversion
\vspace{-1em}
\shortversionend
\begin{center}
	\includegraphics[scale=\figurescale]{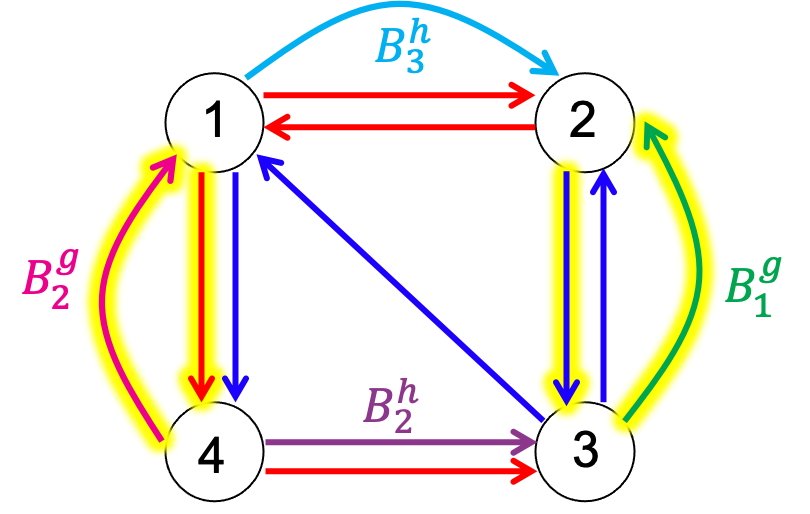}
\end{center}
In this case we are done via the PI edge set $1 \champ[\g] 4 \champ[\Bgtwo \mid\circ] 1$, $2 \champ[\h] 3 \champ[\Bgone \mid\circ] 2$.

\fullversion
\vspace{1ex}
\fullversionend
\noindent\textbf{Case 2: $\boldsymbol{\left|C_\g\right| = 2, \left|C_\h\right| = 3}$.}
Assume w.l.o.g. that $C_\g = 1 \champ[\g] 2 \champ[\g] 1$.  If $C_\h$ passes through both agents 1 and 2 then we are done since we are guaranteed to have a size 2 cycle with $\g$ and $\h$ edges.  Thus $C_\h$ passes through exactly one of them, and we can assume w.l.o.g. that $C_\h = 1 \champ[\h] 4 \champ[\h] 3 \champ[\h] 1$ (note that the reverse direction of the cycle is symmetric by switching the roles of agents 3 and 4).  We get the following structure:
\shortversion
\vspace{-1em}
\shortversionend
\begin{center}
	\includegraphics[scale=\figurescale]{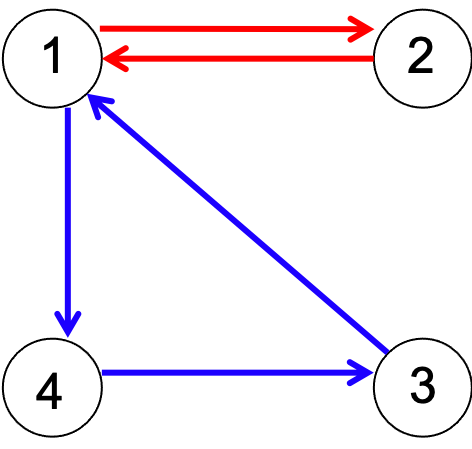}
\end{center}
As in the previous case we may assume that the $\Bgone$ and $\Bgtwo$ edges guaranteed by Theorem \ref{thm:good-edge-or-external} are external.
If we have $3\champ[\Bgtwo \mid\circ] 1$ and $4 \champ[\Bgone \mid\circ] 2$, we are done:  $1 \champ[\h] 4 \champ[\Bgone \mid\circ] 2 \champ[\g] 1$.  Thus assume we have the edges $3\champ[\Bgone \mid\circ] 2$ and $4 \champ[\Bgtwo \mid\circ] 1$.

We now ask who is a $\g$-champion of 3.  If $1 \champ[\g] 3$, we are done: $1 \champ[\g] 3 \champ[\h] 1$.  If $2 \champ[\g] 3$, we are done: $1 \champ[\h] 4 \champ[\Bgtwo] 1$, $2 \champ[\g] 3 \champ[\Bgone] 2$.  Thus $4 \champ[\g] 3$, and we have the structure:
\begin{center}
	\includegraphics[scale=\figurescale]{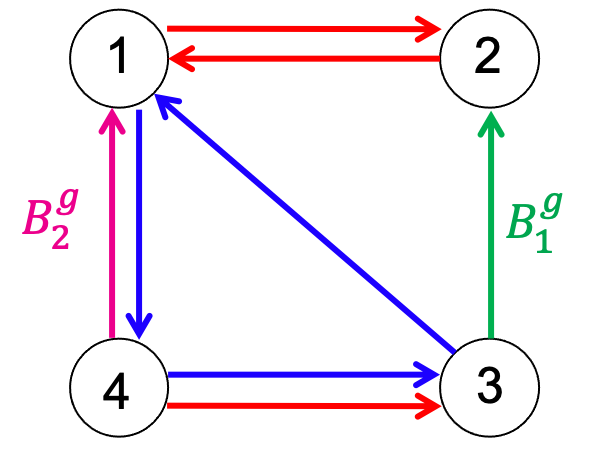}
\end{center}

Consider $C_\h$.  We ask which agent $i$ satisfies $i \champ[\Bhfour \mid\circ] 3$ (such exists by Lemma \ref{lem:alg_start}, since $3 = \succ(4)$ in $C_\h$).  We cannot have $1 \champ[\Bhfour \mid\circ] 3$ since $1 = \pred(4)$ in $C_\h$. If $4 \champ[\Bhfour \mid\circ] 3$, we are done: $1 \champ[\h] \linebreak[1] 4 \champ[\Bhfour \mid\circ] \linebreak[1] 3 \champ[\Bgone \mid\circ] \linebreak[1] 2 \champ[\g] \linebreak[1] 1$.  If $2 \champ[\Bhfour \mid\circ]3$ we are done: $1 \champ[\h] 4 \champ[\Bgtwo \mid\circ] 1$, $2 \champ[\Bhfour \mid\circ] 3 \champ[\Bgone \mid\circ] 2$.  Thus, we must have $3 \champ[\Bhfour \mid\circ] 3$.\footnote{Note that as opposed to a self $\g$-loop, $3 \champ[\Bhfour \mid\circ] 3$ is not a PI-cycle since $\Bhfour$ is not released within the cycle.}

We now ask which agent $i$ satisfies $i \champ[\Bhfour \mid\circ] 1$ (such exists by Lemma \ref{lem:alg_step} since $3 \champ[\Bhfour \mid\circ] 3$ and $1 = \succ(3)$ in $C_\h$).  We cannot have $1 \champ[\Bhfour \mid\circ] 1$, since $1 = \pred(4)$ in $C_\h$.  If $3 \champ[\Bhfour \mid\circ] 1$, we are done: $1 \champ[\h] 4 \champ[\g] 3 \champ[\Bhfour \mid\circ] 1$.  If $4 \champ[\Bhfour \mid\circ] 1$, we are done: $1 \champ[\h] 4 \champ[\Bhfour \mid\circ] 1$. Thus $2 \champ[\Bhfour \mid\circ] 1$, and we have the structure:
\shortversion
\vspace{-1em}
\shortversionend
\begin{center}
	\includegraphics[scale=\figurescale]{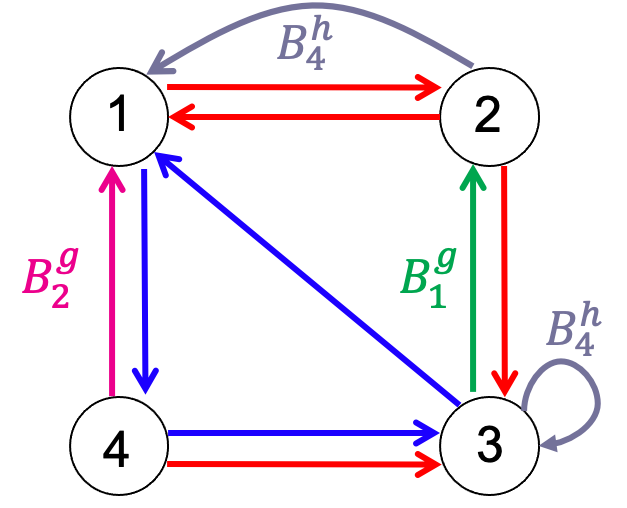}
\end{center}

Finally, we ask which agent $i$ satisfies $i \champ[\Bhone \mid\circ] 4$ (such exists by Lemma \ref{lem:alg_start}).  We cannot have $2 \champ[\Bhone \mid\circ] 4$, as otherwise, together with $2 \champ[\Bhfour \mid\circ] 1$ we have by Observation \ref{obs:bottom_bundle_ineq} that $\Bhone <_2 \Bhfour <_2 \Bhone$, contradiction.  We cannot have $3 \champ[\Bhone \mid\circ] 4$,
as $3 = \pred(1)$ in $C_\h$.  Thus, we must have $4 \champ[\Bhone \mid\circ] 4$, and we are done: $1 \champ[\g] 2 \champ[\Bhfour \mid\circ] 1$, $4 \champ[\Bhone \mid\circ] 4$.

\begin{maybeappendix}{4p-NoEnvy-Cases-3+4}
%
\fullversion
\vspace{1ex}
\fullversionend
\noindent\textbf{Case 3: $\boldsymbol{\left|C_\g\right| = 2, \left|C_\h\right| = 4}$.}
Assume (again) w.l.o.g. that $C_\g = 1 \champ[\g] 2 \champ[\g] 1$. As in Case 1 we can assume w.l.o.g. that we have the edges $3\champ[\Bgone \mid\circ] 2$, $4 \champ[\Bgtwo \mid\circ] 1$.

Consider $C_\h$.  If it contains either the edge $1 \champ[\h] 2$ or $2 \champ[\h] 1$ we are done since we are guaranteed to have a 2-cycle with $\g$ and $\h$ edges.  Thus, we have one of the following two structures:
\begin{center}
	\includegraphics[scale=\figurescale]{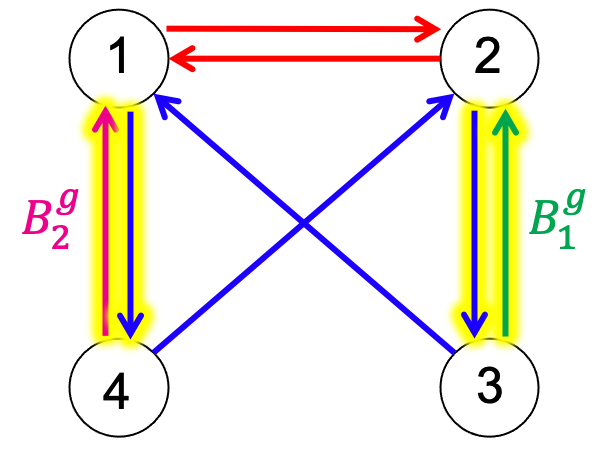} \qquad
	\includegraphics[scale=\figurescale]{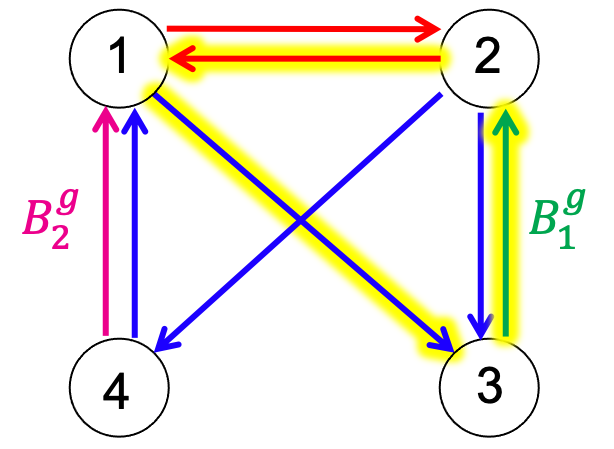}
\end{center}
In the left graph we have the PI edge set $1 \champ[\h] 4 \champ[\Bgtwo \mid\circ] 1$, $2 \champ[\h] 3 \champ[\Bgone \mid\circ] 2$, and in the right graph we have the PI cycle $1 \champ[\h] 3 \champ[\Bgone \mid\circ] 2 \champ[\g] 1$. Thus we are done.

\fullversion
\vspace{1ex}
\fullversionend
\noindent\textbf{Case 4: $\boldsymbol{\left|C_\g\right| = 3, \left|C_\h\right| = 3}$.}
Assume w.l.o.g. that $C_\g = 1 \champ[\g] 2 \champ[\g] 4 \champ[\g] 1$.  Assume first that $C_\h$ is parallel to $C_\g$, i.e., $C_\h = 1 \champ[\h] 2 \champ[\h] 4 \champ[\h] 1$.
Consider $C_\g$.  By Theorem \ref{thm:good-edge-or-external} there exists a good or external $B_j^{\g}$-edge for every $j \in \{1,2,4\}$.  By Corollary \ref{cor:no_single_external_source}, not all of these can come out of agent 3 (i.e., be external).  Thus, w.l.o.g., there is a good $\Bgone$-edge, which can be either one of
$1 \champ[\Bgone \mid\circ] 2, 1 \champ[\Bgone \mid\circ] 4, 2\champ[\Bgone \mid\circ] 4$,
and it is easy to verify that in each of the cases we obtain a PI cycle by carefully choosing complementary $\g$ and $\h$ edges from $C_\g$ and $C_\h$, respectively.
Thus we assume w.l.o.g. that $C_\h = 1 \champ[\h] 2 \champ[\h] 3 \champ[\h] 1$ (note that if this cycle is in the reverse direction we obtain the PI cycle $1 \champ[\g] 2 \champ[\h] 1$), and we have the structure
\begin{center}
	\includegraphics[scale=\figurescale]{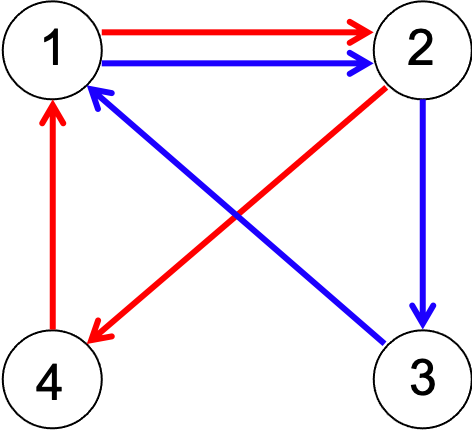}
\end{center}

We are still missing a $\g$-champion of 3 and an $\h$-champion of 4.
If agent 1 is one of them, then we immediately get a 2-cycle with $\g$ and $\h$ edges.
The same is true if both $3 \champ[\h] 4$ and $4 \champ[\g] 3$.
Thus either $2 \champ[\g] 3$ or $2 \champ[\h] 4$.
Since both cases are symmetric we may assume w.l.o.g. that $2 \champ[\h] 4$, and we have the structure
\begin{center}
	\includegraphics[scale=\figurescale]{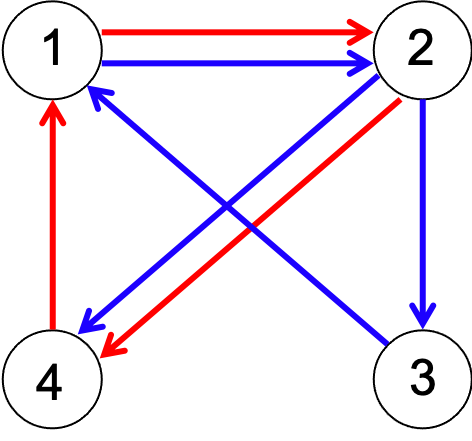}
\end{center}
Consider $C_\h$.  By Theorem \ref{thm:good-edge-or-external}, $M_\allocs$ contains a good or external $\Bhthree$-edge.  Assume first that we have a good $\Bhthree$-edge.  The only possible good $\Bhthree$-edges are $3 \champ[\Bhthree \mid\circ] 1$, $3 \champ[\Bhthree \mid\circ] 2$ and $1 \champ[\Bhthree \mid\circ] 2$.  In the first two cases we are done:  $1 \champ[\g] 2 \champ[\h] 3 \champ[\Bhthree \mid\circ] 1$ and $2 \champ[\h] 3 \champ[\Bhthree \mid\circ] 2$, respectively.  Thus we assume that we have $1 \champ[\Bhthree \mid\circ] 2$, and we get the structure:
\begin{center}
	\includegraphics[scale=\figurescale]{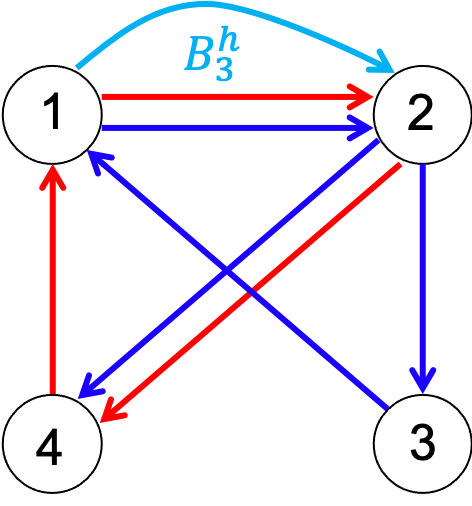}
\end{center}

By Theorem \ref{thm:good-edge-or-external}, $M_\allocs$ contains a good or external $\Bhtwo$-edge.  If it is a good $\Bhtwo$-edge, then it is one of $2 \champ[\Bhtwo \mid\circ] 3$, $2 \champ[\Bhtwo \mid\circ] 1$ and $3 \champ[\Bhtwo \mid\circ] 1$,
and in all cases we get a PI cycle:
$$1 \champ[\Bhthree \mid\circ] 2 \champ[\Bhtwo \mid\circ] 3 \champ[\h] 1,\; 1 \champ[\h] 2 \champ[\Bhtwo \mid\circ] 1\;\; \text{and}\;\; 1 \champ[\Bhthree \mid\circ] 2 \champ[\h] 3 \champ[\Bhtwo \mid\circ] 1,$$
respectively.  If it is an external $\Bhtwo$-edge, then it is one of $4 \champ[\Bhtwo \mid\circ] 1$ and $4 \champ[\Bhtwo \mid\circ] 3$ (recall that there cannot be a $\Bhtwo$-edge going into 2), and these cases admit the respective PI cycles:
$$1 \champ[\h] 2 \champ[\g] 4 \champ[\Bhtwo \mid\circ] 1 \;\; \text{and}\;\; 1 \champ[\Bhthree \mid\circ] 2 \champ[\g] 4 \champ[\Bhtwo \mid\circ] 3 \champ[\h] 1.$$
Thus, we are done with the case that $M_\allocs$ contains a good $\Bhthree$-edge.

Assume now that $M_\allocs$ contains an external $\Bhthree$-edge.
If $M_\allocs$ contains an external $\Bhtwo$-edge,
then it must be $4 \champ[\Bhtwo \mid\circ] 3$, as the other possibility $4 \champ[\Bhtwo \mid\circ] 1$ again closes the PI cycle $1 \champ[\h] 2 \champ[\g] 4 \champ[\Bhtwo \mid\circ] 1$, and we get the following structure (the assumed external $\Bhthree$-edge is not drawn):
\begin{center}
	\includegraphics[scale=\figurescale]{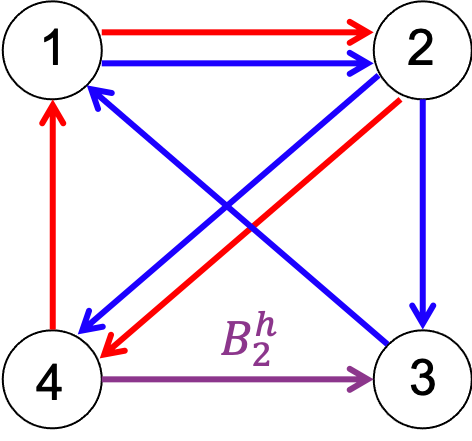}
\end{center}

By Corollary \ref{cor:no_single_external_source}, since there are external $\Bhthree$ and $\Bhtwo$ edges, there cannot be an external $\Bhone$-edge.
Thus, the $\Bhone$-edge guaranteed by Theorem \ref{thm:good-edge-or-external}
must be a good edge, which can be one of $1 \champ[\Bhone \mid\circ] 3$, $1 \champ[\Bhone \mid\circ] 2$ and $2 \champ[\Bhone \mid\circ] 3$. These cases admit the respective PI cycles:
$$1 \champ[\Bhone \mid\circ] 3 \champ[\h] 1,\; 1 \champ[\Bhone \mid\circ] 2 \champ[\g] 4 \champ[\Bhtwo \mid\circ] 3 \champ[\h] 1 \;\; \text{and}\;\; 1 \champ[\g] 2 \champ[\Bhone \mid\circ] 3 \champ[\h] 1.$$
Thus, we may assume that $M_\allocs$ does not contain an external $\Bhtwo$-edge.

We ask which agent $i$ satisfies $i \champ[\Bhtwo \mid\circ] 3$ (such exists by Lemma \ref{lem:alg_start}).
We cannot have $1 \champ[\Bhtwo \mid\circ] 3$ since $1 =\pred(2)$ in $C_\h$.  We cannot have $4 \champ[\Bhtwo \mid\circ] 3$ since we assume there are no external $\Bhtwo$-edges.
Thus $2 \champ[\Bhtwo \mid\circ] 3$  or $3 \champ[\Bhtwo \mid\circ] 3$.  If $3 \champ[\Bhtwo \mid\circ] 3$, then we get the following structure (again, the assumed external $\Bhthree$-edge is not drawn):
\begin{center}
	\includegraphics[scale=\figurescale]{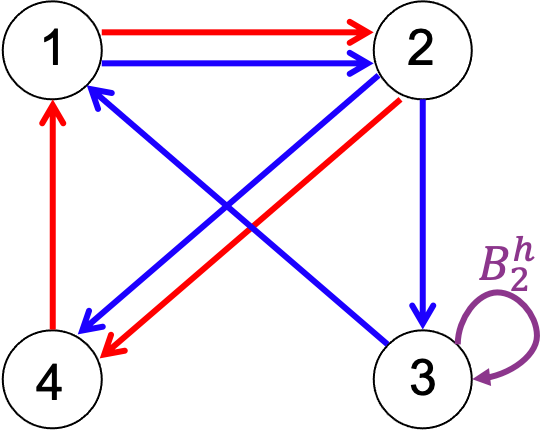}
\end{center}

The assumed external $\Bhthree$-edge can be either $4 \champ[\Bhthree] 1$ or $4 \champ[\Bhthree] 2$, and in both cases we are done via PI edge sets:
$$1 \champ[\h] 2 \champ[\g] 4 \champ[\Bhthree \mid\circ] 1,\, 3 \champ[\Bhtwo \mid\circ] 3 \quad\;\text{or}\quad\; 2 \champ[\g] 4 \champ[\Bhthree \mid\circ] 2,\, 3 \champ[\Bhtwo \mid\circ] 3,$$
respectively.  Thus $2 \champ[\Bhtwo \mid\circ] 3$, and we have the following structure (again, the assumed external $\Bhthree$-edge is not drawn):
\begin{center}
	\includegraphics[scale=\figurescale]{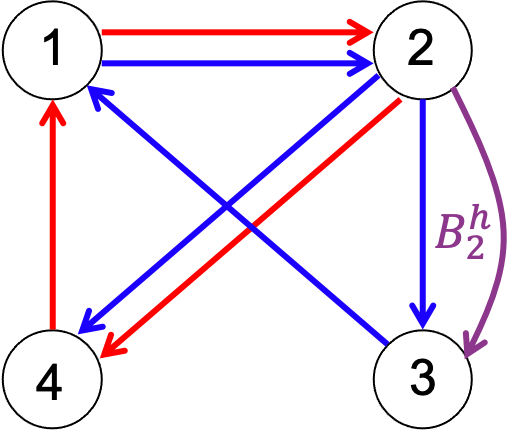}
\end{center}
It is easy to verify that if agent $2$ has other $\g$ or $\h$ champions except 1, then either we get an immediate PI cycle or we get a size 2 good  $\h$-cycle, which along with the size 3 $\g$-cycle that we already have is a case we already solved (Case 2).  Thus 1 is the unique champion of 2 (w.r.t to both $\g$,$\h$), and recall that $1 \nenvies 2$.

For the next step in the proof we shall need the following claim which holds in general.
\fullversion
Its proof is deferred to Appendix \ref{apx:4-agents}.
\fullversionend
\begin{claim}\label{claim:B_g_subseteq_B_h}
	Let $i,j$ be agents such that $i \nenvies j$, and $i$ is the unique $\g$ and $\h$ champion of $j$.  If $\h <_i \g$, then there exists a choice of $\discard[\h]{i}{j}$ and $\discard[\g]{i}{j}$ such that $\discard[\h]{i}{j} \subseteq \discard[\g]{i}{j}$.
\end{claim}
%
\shortversion
\begin{proof}
	Choose some arbitrary discard set $\discard[\h]{i}{j}$.  By Observation \ref{obs:g_notin_bottom}, $\h \notin \discard[\h]{i}{j}$ and thus
	$$X_i <_i (X_j \cup \{\h\}) \setminus \discard[\h]{i}{j} = (X_j \setminus \discard[\h]{i}{j}) \cup \{\h\} <_i (X_j \setminus \discard[\h]{i}{j}) \cup \{\g\},$$
	where the first inequality holds since $1 \champ[\h] 2$, the equality holds since $\h \notin \discard[\h]{i}{j}$ and the second inequality is by \cancbty.  Therefore, since $i$ is the unique $\g$-champion of $j$, it follows that $i$ is the most envious agent of $(X_j \setminus \discard[\h]{i}{j}) \cup \{\g\}$ (otherwise that subset has another most envious agent and consequently $X_j \cup \{\g\}$ has other most envious agents except $i$).
	Denote the corresponding minimally envied subset by $S$.  Thus, we can choose $\discard[\g]{i}{j}$ to be $X_j\setminus S$. Clearly $\discard[\h]{i}{j} \subseteq \discard[\g]{i}{j}$.
\end{proof}
\shortversionend

Assume that $\h <_1 \g$. By Claim \ref{claim:B_g_subseteq_B_h} there exists a choice of $\discard[\h]{1}{2}$ and $\discard[\g]{1}{2}$ such that $\discard[\h]{1}{2} \subseteq \discard[\g]{1}{2}$.  Note that $\Bhtwo$ and $\Bgtwo$ are arbitrary bottom half-bundles (\textsl{i.e.}, discard sets that do not contain $\h$ and $\g$, respectively) whose specific identity had no impact on the proof thus far.  Thus we can redefine $\Bhtwo$ and $\Bgtwo$ to be this choice of $\discard[\h]{1}{2}$ and $\discard[\g]{1}{2}$, respectively (note that these do not contain $\h$ and $\g$ by Observation \ref{obs:g_notin_bottom}), and have $\Bhtwo \subseteq \Bgtwo$.
Thus, the cycle $1 \champ[\g] 2 \champ[\Bhtwo \mid\circ] 3 \champ[\h] 1$ is a PI cycle since the $\g$-edge releases $\Bgtwo$ which contains $\Bhtwo$.

We can therefore assume that $\g <_1 \h$, which leads by symmetric reasoning to the assumption $\Bgtwo \subseteq \Bhtwo$.  Consider $C_\g$. By Theorem \ref{thm:good-edge-or-external}, $M_\allocs$ contains a good or external $\Bgone$-edge.  If it is a good $\Bgone$-edge, then it is one of $1 \champ[\Bgone \mid\circ] 2$, $1 \champ[\Bgone \mid\circ] 4$, $2 \champ[\Bgone \mid\circ] 4$,
and in all cases we get a PI cycle:
$$1 \champ[\Bgone \mid\circ] 2 \champ[\h] 4 \champ[\g] 1, 1 \champ[\Bgone \mid\circ] 4 \champ[\g] 1, 1 \champ[\h] 2 \champ[\Bgone \mid\circ] 4 \champ[\g] 1$$
respectively.  Thus we can assume that there is an external $\Bgone$-edge.

By Theorem \ref{thm:good-edge-or-external}, $M_\allocs$ contains a good or external $\Bgtwo$-edge.  If it is a good $\Bgtwo$-edge, then it is one of $2 \champ[\Bgtwo \mid\circ] 4$, $2 \champ[\Bgtwo \mid\circ] 1$ and $4 \champ[\Bgtwo \mid\circ] 1$,
and in all cases we get a PI cycle:
$$1 \champ[\h] 2 \champ[\Bgtwo \mid\circ] 4 \champ[\g] 1,\; 1 \champ[\g] 2 \champ[\Bgtwo \mid\circ] 1,\;\;\text{and}\;\; 1 \champ[\g] 2 \champ[\h] 4 \champ[\Bgtwo \mid\circ] 1,$$
respectively, where the first cycle is indeed PI since the $\h$-edge releases $\Bhtwo$ which contains $\Bgtwo$.  Thus, we can assume that there is an external $\Bgtwo$-edge.

By Corollary \ref{cor:no_single_external_source}, since there are external $\Bgone$ and $\Bgtwo$ edges, there cannot be an external $\Bgfour$-edge.  Thus, the $\Bgfour$-edge guaranteed by Theorem \ref{thm:good-edge-or-external} must be good, i.e. it is one of $4 \champ[\Bgfour \mid\circ] 1$, $4 \champ[\Bgfour \mid\circ] 2$, $1 \champ[\Bgfour \mid\circ] 2$.  In the first two cases we immediately get a PI cycle:
$$1 \champ[\h] 2 \champ[\g] 4 \champ[\Bgfour \mid\circ] 1, 2 \champ[\g] 4 \champ[\Bgfour \mid\circ] 2$$
and thus we assume $1 \champ[\Bgfour \mid\circ] 2$, and we have the following structure (the assumed $\Bgone$, $\Bgtwo$ and $\Bhthree$ external edges are not drawn):
\begin{center}
	\includegraphics[scale=\figurescale]{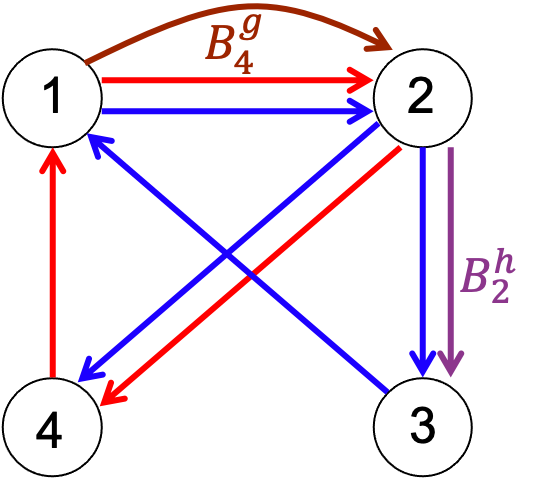}
\end{center}

Consider the external $\Bgone$ edge, which can be one of $3 \champ[\Bgone \mid\circ] 2$ and $3 \champ[\Bgone \mid\circ] 4$.  If it is $3 \champ[\Bgone \mid\circ] 4$, then we are done:  $1 \champ[\Bgfour \mid\circ] 2 \champ[\h] 3 \champ[\Bgone \mid\circ] 4 \champ[\g] 1$.  Thus we have the edge $3 \champ[\Bgone \mid\circ] 2$, and the structure is
\begin{center}
	\includegraphics[scale=\figurescale]{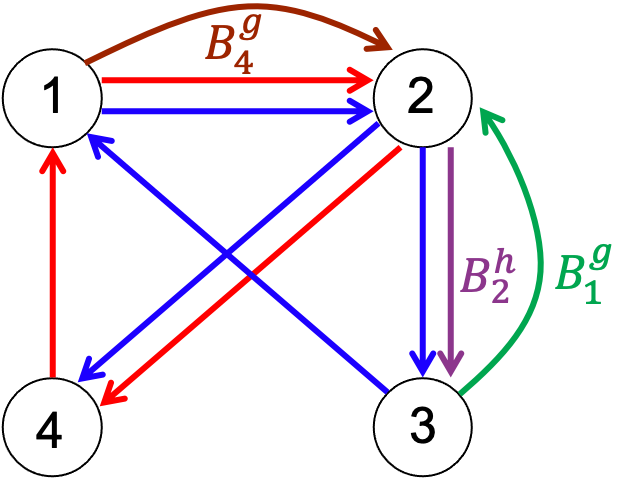}
\end{center}

Finally, we consider the external $\Bgtwo$-edge.  It cannot be $3 \champ[\Bgtwo \mid\circ] 1$, as otherwise, together with $3 \champ[\Bgone \mid\circ] 2$ we get by Observation \ref{obs:bottom_bundle_ineq} that $\Bgone <_3 \Bgtwo <_3 \Bgone$, contradiction.  Hence it must be $3 \champ[\Bgtwo \mid\circ] 4$, and we are done: $1 \champ[\Bgfour \mid\circ] 2 \champ[\h] 3 \champ[\Bgtwo \mid\circ] 4 \champ[\g] 1$.
\end{maybeappendix}

\fullversion
\vspace{1ex}
\fullversionend
\noindent\textbf{Case 5: $\boldsymbol{\left|C_\g\right| = 3, \left|C_\h\right| = 4}$.}
Assume w.l.o.g. that $C_\g = 1 \champ[\g] 3 \champ[\g] 4 \champ[\g] 1$ and that $C_\h = 1 \champ[\h]\linebreak[1] 2 \champ[\h]\linebreak[1] 3 \champ[\h]\linebreak[1] 4 \champ[\h]\linebreak[1] 1$ (note that if $C_\h$ is in the opposite direction we immediately get a PI cycle).  We have the structure:
\shortversion
\vspace{-1em}
\shortversionend
\begin{center}
	\includegraphics[scale=\figurescale]{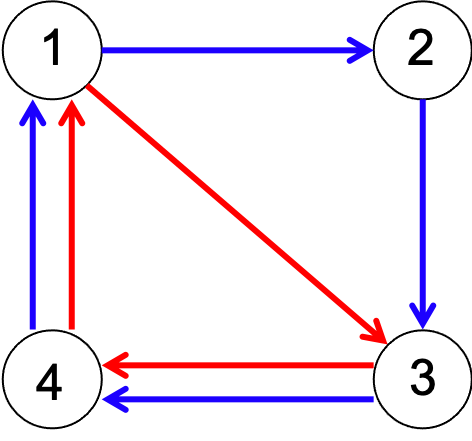}
\end{center}

Consider $C_\g$.
In what follows we reason about possible $\Bgone, \Bgthree$ and $\Bgfour$ edges,
starting with $\Bgthree$.  By Theorem \ref{thm:good-edge-or-external}, $M_\allocs$ contains a good or external $\Bgthree$-edge.  If it is a good $\Bgthree$-edge, then it is one of $3 \champ[\Bgthree \mid\circ] 4$, $3 \champ[\Bgthree \mid\circ] 1$, $4 \champ[\Bgthree \mid\circ] 1$,
and in all cases we get a PI cycle:
$$1 \champ[\g] 3 \champ[\Bgthree \mid\circ] 4 \champ[\h] 1, 1 \champ[\g] 3 \champ[\Bgthree \mid\circ] 1, 1 \champ[\g] 3 \champ[\h] 4 \champ[\Bgthree \mid\circ] 1$$
respectively.  Thus, we can assume that there is an external $\Bgthree$-edge, which can be either $2 \champ[\Bgthree \mid\circ] 1$ or $2 \champ[\Bgthree \mid\circ] 4$, and thus the structure is one of the following:
\begin{center}
	\includegraphics[scale=\figurescale]{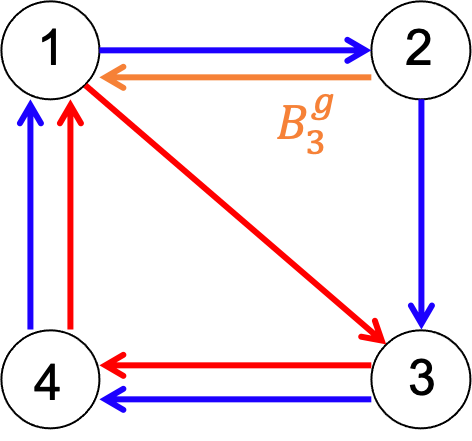} \qquad\qquad
	\includegraphics[scale=\figurescale]{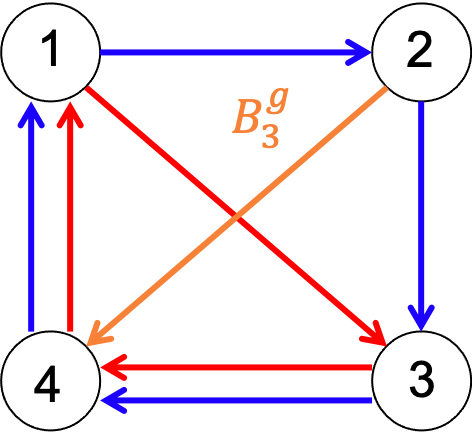}
\end{center}

We now ask which agent $i$ satisfies $i \champ[\Bgone \mid\circ] 3$ (such exists by Lemma \ref{lem:alg_start}, since $3=\succ(1)$ in $C_\g$). We cannot have $4 \champ[\Bgone \mid\circ] 3$, since $4 = \pred(1)$ in $C_\g$. If $1 \champ[\Bgone \mid\circ] 3$, we are done: $1 \champ[\Bgone \mid\circ] 3 \champ[\h] 4 \champ[\g] 1$.  If $3 \champ[\Bgone \mid\circ]3$, we are done regardless of whether $2 \champ[\Bgthree \mid\circ] 1$ or $2 \champ[\Bgthree \mid\circ] 4$; in the first case we have the PI edge set $1 \champ[\h] 2 \champ[\Bgthree \mid\circ] 1$, $3 \champ[\Bgone \mid\circ] 3$, and in the second case we have the PI edge set $1 \champ[\h] 2 \champ[\Bgthree \mid\circ] 4 \champ[\g] 1$, $3 \champ[\Bgone \mid\circ] 3$.
Thus, we must have $2 \champ[\Bgone \mid\circ] 3$.

Therefore, the external $\Bgthree$-edge cannot be $2 \champ[\Bgthree \mid\circ] 1$, as otherwise by Observation \ref{obs:bottom_bundle_ineq} we get $\Bgone <_2 \Bgthree <_2 \Bgone$, contradiction.  Hence we have $2 \champ[\Bgthree \mid\circ] 4$, and we obtain the following structure:
\begin{center}
	\includegraphics[scale=\figurescale]{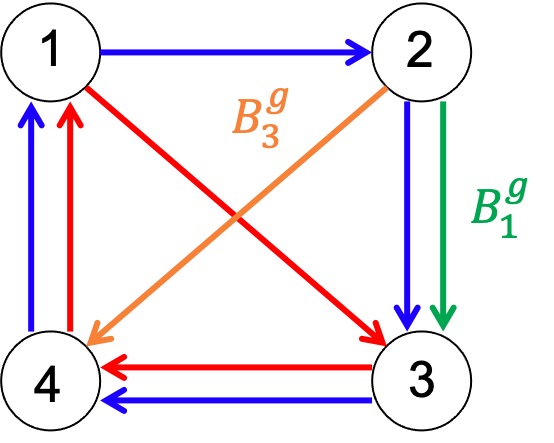}
\end{center}

We now ask which agent $i$ satisfies $i \champ[\Bgfour \mid\circ] 1$ (such exists by Lemma \ref{lem:alg_start}, since $1=\succ(4)$ in $C_\g$). We cannot have $3 \champ[\Bgfour \mid\circ] 1$, since $3 = \pred(4)$ in $C_\g$. We cannot have $2 \champ[\Bgfour \mid\circ] 1$, as otherwise we get a contradiction to Corollary \ref{cor:no_single_external_source}, since there are already external $\Bgone$ and $\Bgthree$ edges.  If $4 \champ[\Bgfour \mid\circ] 1$, we are done:  $1 \champ[\h] 2 \champ[\Bgone \mid\circ] 3 \champ[\g] 4 \champ[\Bgfour \mid\circ] 1$.  Thus we must have $1 \champ[\Bgfour \mid\circ] 1$ and we have the structure:
\shortversion
\vspace{-1em}
\shortversionend
\begin{center}
	\includegraphics[scale=\figurescale]{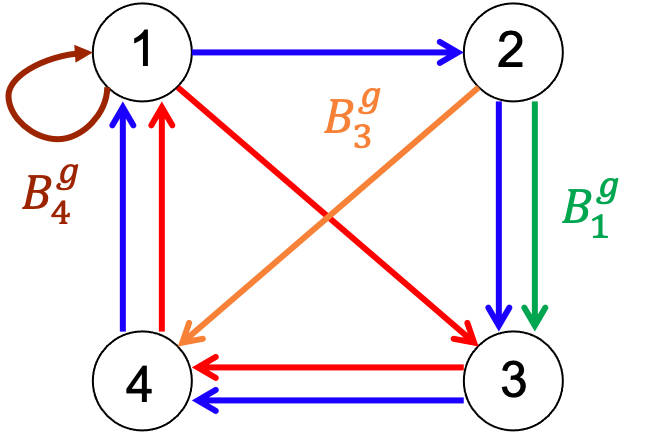}
\end{center}

By Lemma \ref{lem:alg_step}, $1\champ[\Bgfour\mid\circ]1$ implies that there is an agent $i$ that satisfies $i \champ[\Bgfour \mid\circ] 3$ ($3 = \succ(1)$ in $C_\g$).  We cannot have  $2 \champ[\Bgfour \mid\circ] 3$, as again we would get a contradiction to Corollary \ref{cor:no_single_external_source}.  We cannot have  $3 \champ[\Bgfour \mid\circ] 3$ since $3 = \pred(4)$ in $C_\g$.  If  $1 \champ[\Bgfour \mid\circ] 3$ we are done:  $1 \champ[\Bgfour \mid\circ] 3 \champ[\g] 4 \champ[\h] 1$.  Thus we must have $4 \champ[\Bgfour \mid\circ] 3$, and we are done: $3 \champ[\g] 4 \champ[\Bgfour \mid\circ] 3$.

\begin{maybeappendix}{4p-NoEnvy-Case-6}

\fullversion
\vspace{1ex}
\fullversionend
\noindent\textbf{Case 6: $\boldsymbol{\left|C_\g\right| = 4, \left|C_\h\right| = 4}$.}
It is easy to verify that if $C_\g$ and $C_\h$ are not parallel to each other then there must exist a 2-cycle with $\g$ and $\h$ edges.  Thus we may assume w.l.o.g. that $C_\g = 1 \champ[\g] 2 \champ[\g] 3 \champ[\g] 4 \champ[\g] 1$ and $C_\h = 1 \champ[\h] 2 \champ[\h] 3 \champ[\h] 4 \champ[\h] 1$, i.e., we have the following structure:
\begin{center}
	\includegraphics[scale=\figurescale]{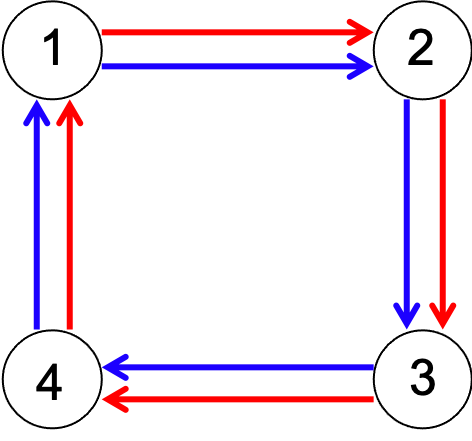}
\end{center}

Consider $C_\g$.  By Theorem \ref{thm:good-edge-or-external}, $M_\allocs$ contains a good $\Bgone$-edge (note that there are no external edges as both cycles pass through all agents), which can either be parallel to one of the edges in $C_\g$ ($1 \champ[\Bgone \mid\circ] 2$, $2 \champ[\Bgone \mid\circ] 3$, or $3 \champ[\Bgone \mid\circ] 4$), diagonal ($1 \champ[\Bgone \mid\circ] 3$ or $2 \champ[\Bgone \mid\circ] 4$) or going back ($1 \champ[\Bgone \mid\circ] 4$).  It is easy to verify that in the last two cases we immediately get a PI cycle by choosing a proper $\g$-edge and maybe also an $\h$-edge.

Since all half-bundles are symmetric, we may thus assume that all good edges with respect to any bottom half-bundle are parallel to $C_\g$ and $C_\h$.  Assume w.l.o.g. that one of these is $1 \champ[B_i^{\g} \mid\circ] 2$ for some $i$, and note that $i \neq 2$.  Considering the 3 possibilities for the good $\Bgtwo$-edge gives us one of the following 3 structures:
\begin{center}
\begin{tabular}{p{3cm}p{3cm}p{3cm}}
	\vspace{0pt}
	\includegraphics[scale=\figurescale]{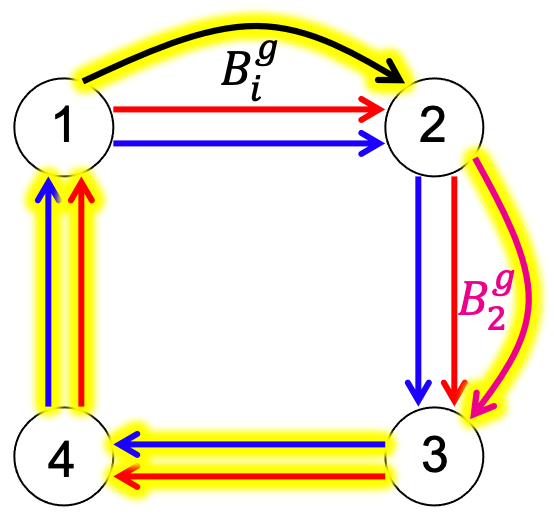} &
	\vspace{0pt}
	\includegraphics[scale=\figurescale]{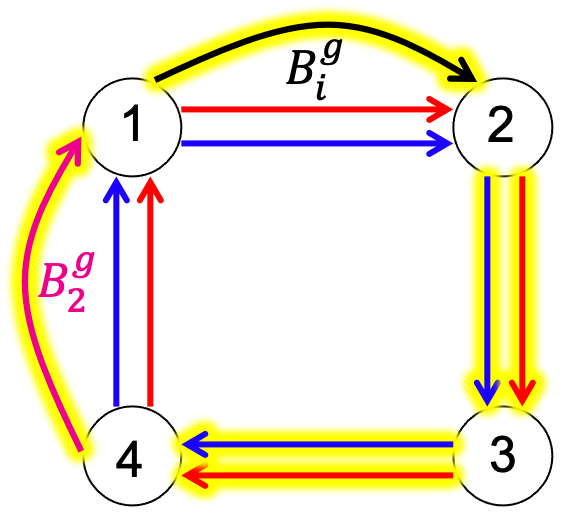} &
	\vspace{0pt}
	\includegraphics[scale=\figurescale]{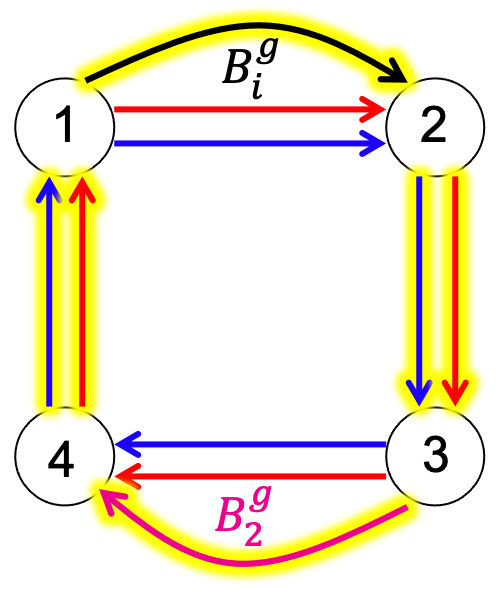}
\end{tabular}
\end{center}
and in any case we can form a PI cycle by taking the $B_i^{\g}$ and $\Bgtwo$ edges, and among the two options for completing with one $\g$ and one $\h$ edge, choose the option so that the $\h$ edge does not go into $i$.
\end{maybeappendix}


\fullversion
\fullversion
\subsection{$\allocs$ is Not Envy-Free}
\label{sec:not-ef-case}
We first consider the structure of the envy edges and champion edges w.r.t. $\g$ and $\h$. We show that there is only one non-trivial such structure in the following lemma, whose proof is deferred to Appendix~\ref{apx:4-agents}.

\begin{lemma}\label{lem:envy_unique_structure}
	 Assuming $M_\allocs$ contains no PI cycle, and up to renaming of the agents, the possible $\g$ and $\h$ champions of 4 are only agents 1 and 2, and all other basic edges in $M_\allocs$ are exactly as depicted in the following figure:
	 \begin{center}
		\includegraphics[scale=\figurescale]{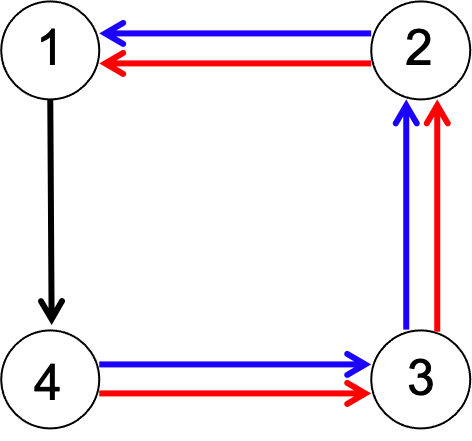}
	\end{center}
\end{lemma}

The above lemma and Observation~\ref{obs:g_notin_bottom} imply:
\begin{remark}
agents 1,2 and 3 are $\g$ and $\h$ decomposed by their unique champion, 2,3 and 4, respectively.
\end{remark}

In the following we denote the $\g$ (resp., $\h$) -decomposition by $X_j = T^g_j\cupdot B^g_j$ (resp. $X_j = T^h_j\cupdot B^h_j$), for $j\in\{1,2,3\}$.
Without loss of generality, we assume that $\g <_4 \h$ (otherwise, switch the names of $\h$ and $\g$ and the champion graph described above remains the same).
By Claim \ref{claim:B_g_subseteq_B_h}, we may redefine $\Bgthree$ and $\Bhthree$ such that they also satisfy
\begin{equation}\label{eq:B_3subsetB'_3}
	\Bgthree \subseteq \Bhthree.
\end{equation}

From this point on we don't always find a PI cycle.
Instead, we show that there is a dominating EFX allocation in which the agent $\vip$ strictly improves (while other agents may become worse-off), unless we find a PI cycle in $\allocs$.
Since we renamed the agents depending on the graph structure, agent $\vip$ can be any one of the agents. Hence we split into cases according to the identity of $\vip$.
\fullversionend
\subsubsection{$\boldsymbol{\vip \neq 2}$}
\shortversion
To start, we require the following assumption. Without loss of generality, we assume that $\g <_4 \h$ (otherwise, switch the names of $\h$ and $\g$ and the champion graph described above remains the same).
By Claim \ref{claim:B_g_subseteq_B_h}, we may redefine $\Bgthree$ and $\Bhthree$ such that they also satisfy
%
\begin{equation}\label{eq:B_3subsetB'_3}
	\Bgthree \subseteq \Bhthree.
\end{equation}
\shortversionend
As a first attempt consider the following allocation, denoted $\allocs'$:

\begin{center}
\renewcommand{\arraystretch}{1.2}
	$\allocs'$=\bundle{1}{X_4}\quad\decomp{2}{\Tgone}{\Bgtwo}\quad\decomp{3}{\Tgtwo}{\g}\quad\decomp{4}{\Ththree}{\h}
\end{center}


\noindent\textbf{All agents are better off in $\boldsymbol{\allocs'}$:}
Agent 1 is better off, since $X_1<_1 X_4$.
Agent 2 is better off, since $X_2 = \Tgtwo \cup \Bgtwo <_2 \Tgone \cup \Bgtwo$, by Observation~\ref{obs:T_k<T_j} (since $2\champ[\g]1$ and $2\nchamp[\g]2$)
and \cancbty.
Agent 3 is better off, since $X_3 <_3 \Tgtwo \cup \{\g\}$ (since $3\champ[\g]2$).
Agent 4 is better off, since $X_4 <_4 \Ththree \cup \{\h\}$ (since $4\champ[\h]3$).

\noindent\textbf{The only strong envy in $\mathbf{X'}$ is from 1 to 2:}
No agent strongly envies the sets $X'_1 = X_4, X'_3 = \Tgtwo \cup \{\g\}, X'_4 = \Ththree \cup \{\h\}$ as all agents are better off and no agent strongly envied any of them in $\allocs$ (by definition of basic championship).

It remains to show that agents 3 and 4 do not strongly envy $X'_2$. Agent 3 does not envy $X'_2$, since
$
X'_3 >_3 X_3 >_3 X_2 =\Tgtwo\cup\Bgtwo >_3 \Tgone \cup \Bgtwo = X'_2%
$
, where the first inequality is due to $3\nenvies2$  and the second inequality is due to Observation~\ref{obs:T_k<T_j} and \cancbty.

If agent 4 envies $X'_2$, then
$
X_4 <_4 X'_4 <_4 \Tgone \cup \Bgtwo <_4 \Tgthree \cup \Bgtwo
$
, where the third inequality is by Observation~\ref{obs:T_k<T_j} and \cancbty.
The following claim shows that there must now exist a PI-cycle.
\fullversion
The proof of the claim is deferred to Appendix \ref{apx:4-agents}.
\fullversionend

\begin{claim}\label{clm:X_4<_4T_3+B_2}
	If $X_4 <_4 \Tgthree \cup \Bgtwo$, then there exists a PI cycle in $M_{\allocs}$.
\end{claim}

\shortversion
\begin{proof}
	Since 4 envies $\Tgthree \cup \Bgtwo$, there exists a most envious agent of this bundle. That is, there exists an agent $i$ such that $i\champ[\Bgtwo\mid\circ]3$.
	$i$ cannot be 3,
	since $X_3 >_3 X_2 = \Tgtwo \cup \Bgtwo >_3 \Tgthree \cup \Bgtwo$, where the first inequality is by $3\nenvies2$ and the second inequality is by Observation~\ref{obs:T_k<T_j}
	and \cancbty.

	In the remaining cases, $i = 1$, $i=2$, $i=4$, we obtain the respective PI cycles:
	$$1\champ[\Bgtwo\mid\circ]3\champ[\g]2\champ[\h]1,~2\champ[\Bgtwo\mid\circ]3\champ[\g]2,~4\champ[\Bgtwo\mid\circ]3\champ[\g]2\champ[\h]1\envies4.$$
\end{proof}
\shortversionend
By Claim~\ref{clm:X_4<_4T_3+B_2}, we may assume that agent 4 does not strongly envy $X'_2$.  It remains to consider agent 1. If she does not strongly envy $X'_2$ then $\allocs'$ is an EFX allocation that Pareto-dominates $\allocs$.  Thus we assume that the only strong envy in $\mathbf{X'}$ is from 1 to 2 as we wanted to show.

\vspace{0.1in}
Now, notice that $\Tgone \cup \Bgtwo >_2 \max_{2} \left\{X_4,\, \Tgtwo \cup \g,\, \Ththree \cup \h,
\, X_3\right\}$ in $\mathbf{X'}$, since agent 2 doesn't envy agents 3 and 4 in both $\allocs$ and $\allocs'$, and agent 2 is better off in $\allocs'$. Let $Z\subseteq \Tgone \cup \Bgtwo$ be a set of minimal size satisfying
$
	Z>_2 \max_{2} \left\{X_4,\, \Tgtwo \cup \g,\, \Ththree \cup \h,
\, X_3\right\},
$ {\sl i.e.}, the inequality no longer holds after removing any item from $Z$.
If $X_4>_1 Z$ then it is easy to find a dominating allocation as shown in the following lemma.
\fullversion
The proof of the lemma is deferred to Appendix \ref{apx:4-agents}.
\fullversionend
\begin{lemma}\label{lem:X_4>_1_Z}
	If $X_4>_1 Z$ then there is an EFX allocation $\mathbf{Y}$ that dominates $\allocs$.
\end{lemma}
\shortversion
\begin{proof}

	Consider the allocation $\mathbf{Y}$ obtained from $\allocs'$ by replacing $X'_2$ with $Z$:


	\begin{center}
		\renewcommand{\arraystretch}{1.2}
		$\mathbf{Y}$=\bundle{1}{X_4}\quad\bundle{2}{Z}\quad\decomp{3}{\Tgtwo}{\g}\quad\decomp{4}{\Ththree}{\h}
	\end{center}
	%
	%
	We claim this allocation is EFX.

	\noindent\textbf{No agent envies agent 2:}
	Agents 3 and 4 do not envy agent 2 since they did not envy her in $\allocs'$ and we only removed items from $X'_2$ in the transition to $\mathbf{Y}$. Agent 1 does not envy agent 2, since we assume $X_4>_1 Z$.

	\noindent\textbf{Agent 2 envies no other agent:}
	This follows from the definition of $Z$.

	Since only the bundle of agent 2 has been changed in the transition from $\allocs'$ to $\mathbf{Y}$ we conclude that there is no strong envy that does not include agent 2 (as there wasn't any in $\allocs'$), and consequently $\mathbf{Y}$ is EFX.
	Moreover, agents 1,3, and 4 are better off in allocation $\mathbf{Y}$ relative to $\allocs$. Since we assumed that agent $\vip$ is not agent 2, it follows that $\mathbf{Y}$ dominates $\allocs$ and we are done.
\end{proof}
\shortversionend

%
By Lemma \ref{lem:X_4>_1_Z} we may assume for the rest of the proof that $X_4 <_1 Z$.  Consider the allocation $\allocs''$ defined as follows:

\begin{center}
\begin{minipage}{\textwidth}
\renewcommand{\arraystretch}{1.2}
\begin{center}
$\allocs''$=
	\bundle{1}{Z}\quad
	\bundle{2}{\max_2\{X_4,\, \Tgtwo \cup \g,\, \Ththree\cup \h,\, X_3 \}}\quad
	\bundle{3}{\begin{array}{c}
		\Tgtwo \cup \g \\
		\mbox{or}^\ast\; X_3
		\end{array} }\quad
	\bundle{4}{\begin{array}{c}
		X_4\mbox{ or}^\dagger\\
		\Ththree \cup \h
		\end{array} }\quad
\end{center}
\noindent{\scriptsize $\ast$ $X''_3=\Tgtwo \cup \{\g\}$ unless $X''_2=\Tgtwo \cup \{\g\}$, in which case $X''_3=X_3$;}\hfill
\noindent{\scriptsize $\dagger$ $X''_4=X_4$ unless $X''_2=X_4$, in which case $X''_4=\Ththree \cup \{\h\}$.}
\end{minipage}
\end{center}


\begin{claim}
	$\allocshat$ is EFX.
\end{claim}
\begin{proof}

Agent 1 envies no one since $Z >_1 X_4$ and she did not envy $X_3, \Tgtwo \cup \{\g\}$ and $\Ththree \cup \{\h\}$ in previously considered allocations when she was worse off.

Agent $i$, for $i\in\{3,4\}$, does not strongly envy another agent: Notice that $\hX_i \geq_i X_i$. We have shown $X_i>_i \Tgone \cup \Bgtwo \geq_i Z$ in our analysis of allocation $\allocs'$. Thus, $\hX_i>_i\hX_1$. Moreover, in allocation $\allocs$ no agent strongly envies $\Tgtwo \cup \g$ and $\Ththree \cup \h$, by definition of champion.

Agent 2 does not strongly envy 1 by definition of $Z$, and does not strongly envy agents 3 or 4 since $\hX_2=\max_{2} \left\{X_4,\, \Tgtwo \cup \g,\, \Ththree \cup \h,
		\, X_3\right\}\geq_2 \max_{2}\big\{\hX_3,\, \hX_4\big\}$.
\end{proof}
Agent 1 is strictly better off in allocation $\allocshat$. Thus, if agent $\vip$ is agent 1, then $\allocshat$ dominates $\allocs$, and we are done.  Hence, it remains to show that agents 3 and 4 can be made better off.  If $ \hX_2 = X_4$, then both 3 and 4 are better off and we are done.
We split to cases according to the rest of the possibilities for the identity of $\hX_2$:
\begin{center}
	Case \framebox{A}: $\hX_2= \Tgtwo \cup \{\g\}$ , Case \framebox{B}: $\hX_2= \Ththree \cup \{\h\}$, Case \framebox{C}: $\hX_2= X_3$
\end{center}

\shortversion
	We shall need the following claim that holds in all three cases:
\shortversionend
\fullversion
	We shall need the following claim that holds in all three cases (the proof is deferred to Appendix~\ref{apx:4-agents}):%
\fullversionend

\begin{claim}\label{clm:4_MEA_T'_3_cup_h}
	Agent 4 is most envious of $\Ththree \cup \{\h\}$ in $\allocshat$ (in cases \framebox{A}, \framebox{B}, \framebox{C}).
\end{claim}
\shortversion
\begin{proof}
	It suffices to show that only agent 4 envies $\Ththree \cup \h$ in $\allocshat$.
	Agent 4 envies $\Ththree \cup \h$, since $\hX_4=X_4$ in all three cases and $4\champ[\h]3$ in the allocation $\allocs$. 
	Agent $i$, for $i\in\{1,3\}$, does not envy $\Ththree \cup \h$, since $\hX_i \geq_i X_i$ and $i\nchamp[\h] 3$ in allocation $\allocs$ (Observation~\ref{obs:non-champion-doesnt-envy-top-half}).
	Agent $2$ does not envy $\Ththree \cup \h$ in all three cases since $\hX_2 = \max_{2} \left\{ \Tgtwo \cup \g,\, \Ththree \cup \h,\, X_3\right\} \geq_2 \Ththree \cup \h$.
\end{proof}
\shortversionend

\vspace{1ex}
\noindent\textbf{Case \framebox{A}}
\vspace{1ex}

In this case $\allocs''$ is

\begin{center}
\renewcommand{\arraystretch}{1.2}
	$\allocs''$=\bundle{1}{Z}\quad
			\decomp{2}{\Tgtwo}{\g}\quad
			\bundle{3}{X_3}\quad
			\bundle{4}{X_4}
\end{center}


Notice that $\h$ and the set $\Bgone$ remain unallocated. Let $\b\in\Bgone$ be an arbitrary good in $\Bgone$
\fullversion
(note that $\Bgone \neq \emptyset$ since otherwise, as valuations are non-degenerate, we have $X_1 <_1 X_1 \cup \{\g\} = \Tgone \cup \{\g\}$, in contradiction to Observation \ref{obs:non-champion-doesnt-envy-top-half} since $1 \nchamp[\g] 1$ in $\allocs$).
\fullversionend
\shortversion
. Recall that $\Bgone$ is non empty, as otherwise 1 would have been a self $\g$-champion in the original allocation $\allocs$.
\shortversionend

Consider the champion graph $M_{\allocshat}$ restricted to envy edges and champion edges with respect to $\h$ and $\b$. Agent 2 envies agent 1 by definition of $Z$, and agent 3 envies agent 2 since $3\champ[\g]2$ in allocation $\allocs$.
Thus we have the structure:
\begin{center}
	\includegraphics[scale=\figurescale]{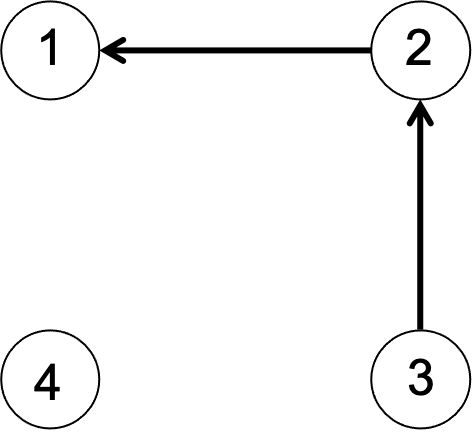}
\end{center}
\vspace{1ex}
\noindent\textbf{Subcase: $\boldsymbol{\vip=3}$.} By Observation~\ref{obs:exists-champion}, agent 3 has a $\b$-champion in $M_{\allocshat}$. If this $\b$-champion is agent 1,2 or 3, we obtain a PI-cycle: $1\champ[\b]3\envies2\envies1$, $2\champ[\b]3\envies 2$ or $3\champ[\b]3$, respectively. By Lemma~\ref{lem:pareto-improvable}, since this PI-cycle always includes agent 3, it follows that there exists a partial EFX allocation $\allocshat'$ that Pareto dominates $\allocshat$, in which agent 3 is better off. Hence, $\allocshat'$ dominates $\allocs$ in case agent $\vip$ is agent 3, as desired. Thus, assume agents 1,2 and 3 are not $\b$-champions of agent 3. This leaves agent 4 as the unique $\b$-champion of agent 3 in allocation $\allocshat$, and we have the structure
\begin{center}
	\includegraphics[scale=\figurescale]{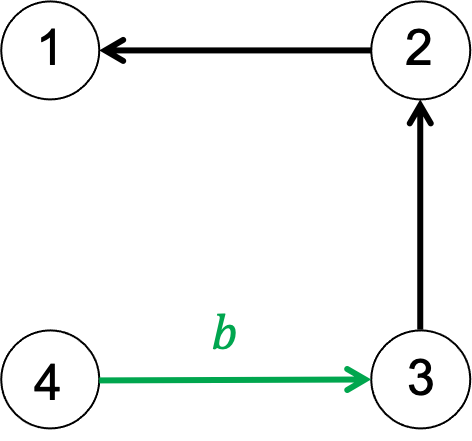}
\end{center}

We next show that in the original allocation $\allocs$, either $4\champ[\b]3$ or $1\champ[\b]3$ must hold. In allocation $\allocshat$,
4 is the most envious agent of $X_3\cup \{\b\}$ while agents 2 and 3 are not. In the allocation $\allocs$ the bundles of 3 and 4 were the same, and agent 2 was better off relative to $\allocshat$. Therefore, neither agent 2 nor 3 can be most envious agents of $X_3\cup \{\b\}$ in allocation $\allocs$. It follows that either 1 or 4 were most envious agents of $X_3\cup \{\b\}$, {\sl i.e.}, $4\champ[\b]3$ or $1\champ[\b]3$  in $\allocs$.

We conclude that there exists a PI-cycle in $\allocs$ in both cases, which we could have applied in hindsight:
$4\champ[\b]3\champ[\h]2\champ[\g]1\envies4$ or
$1\champ[\b]3\champ[\h]2\champ[\g]1$, respectively. Both are indeed PI-cycles, since the edge $2\champ[\g]1$ releases $\Bgone$ of which the item $\b$ is a member. Hence, we are done by Lemma~\ref{lem:pareto-improvable}.

\vspace{1ex}
\noindent\textbf{Subcase: $\boldsymbol{\vip=4}$.}
Agent 4 is the most envious agent of $\Ththree \cup \{\h\}$ by Claim \ref{clm:4_MEA_T'_3_cup_h}, implying that $4\champ[\h]3$ in $\allocshat$ (since $\hX_3 =X_3)$. Thus we have the structure:
\begin{center}
	\includegraphics[scale=\figurescale]{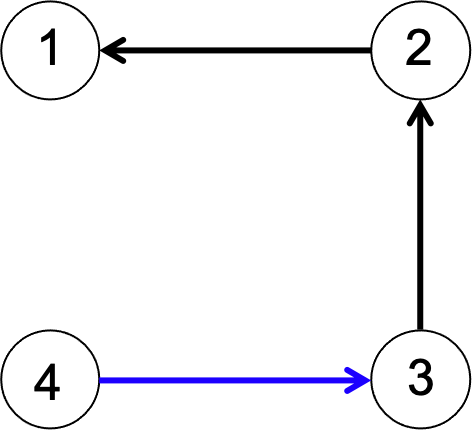}
\end{center}

By Observation~\ref{obs:exists-champion}, there exists a $\b$-champion of agent 4. Regardless of who that champion is
we obtain a PI cycle that includes agent 4:
$$1\champ[\b]4\champ[\h]3\envies 2\envies1,~~
2\champ[\b]4\champ[\h]3\envies 2,~~
3\champ[\b]4\champ[\h]3~~ \text{or}~~
4\champ[\b]4.$$
By Lemma~\ref{lem:pareto-improvable}, since the resulting PI-cycle includes agent 4, it follows that there exists a partial EFX allocation $\allocshat'$ that Pareto dominates $\allocshat$, in which agent 4 is better off. Hence, $\allocshat'$ dominates $\allocs$ in case agent $\vip$ is agent 4, as desired.

\vspace{1ex}
\noindent\textbf{Case \framebox{B} or \framebox{C}}
\vspace{1ex}

In these cases the allocation $\allocs''$ is
\begin{center}
\renewcommand{\arraystretch}{1.2}
\begin{tabular}{cc}
	Case \framebox{B}	& Case \framebox{C}\\[1.3ex]
	$\allocs''$=\bundle{1}{Z}\quad
					\decomp{2}{\Ththree}{\h}\quad
					\decomp{3}{\Tgtwo}{\g}\quad
					\bundle{4}{X_4}
			\qquad&	\qquad
	$\allocs''$=\bundle{1}{Z}\quad
					\bundle{2}{X_3}\quad
					\decomp{3}{\Tgtwo}{\g}\quad
					\bundle{4}{X_4}
\end{tabular}
\end{center}

\vspace{0.1cm}

In both cases agent $3$ is better off ($X_3 <_3 \Tgtwo \cup \{\g\}$ since $3 \champ[\g] 2$ in $\allocs$).  Thus we may assume $\vip = 4$.
We have that $2\envies1$ in allocation $\allocshat$, by definition of $Z$.
In Case \framebox{B} we have $4\envies2$ by Claim \ref{clm:4_MEA_T'_3_cup_h}, since $\hX_2 = \Ththree \cup \{\h\}$.
In Case \framebox{C} we have $4\champ[\h]2$ by Claim \ref{clm:4_MEA_T'_3_cup_h}, since $\hX_2 = X_3$.
Therefore we have the following structure in the allocation $\allocshat$:
\begin{center}
	\includegraphics[scale=\figurescale]{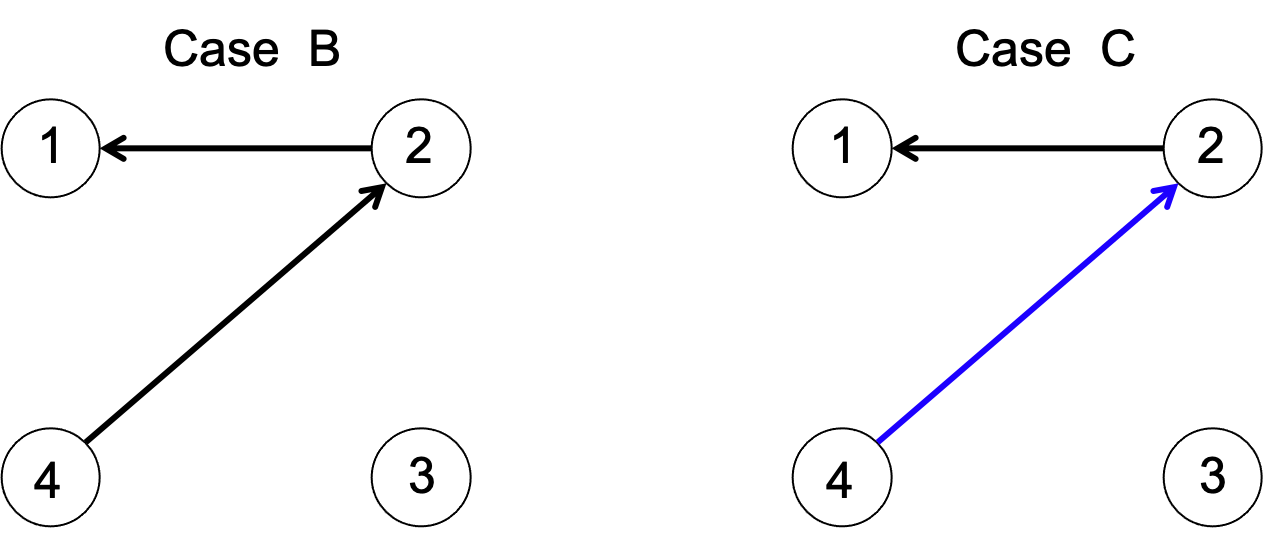}
\end{center}

In both cases, $\Bgone$ remains unallocated. Let $\b\in\Bgone$ be an arbitrary good in $\Bgone$. Recall that $\Bgone$ is non empty, since $1\nchamp[\g]1$ in the original allocation $\allocs$. 
By Observation \ref{obs:exists-champion}, there is some agent $i\in[4]$ such that $i\champ[\b]4$ in allocation $\allocshat$.

If $i$ is 1,2 or 4, then we obtain a PI-cycle in both Cases \framebox{B} and \framebox{C}. In Case \framebox{B} these cycles are
$$
1\champ[\b]4\envies2\envies1,\quad
2\champ[\b]4\envies2,\quad
4\champ[\b]4.
$$
In Case \framebox{C} these cycles are
$$
1\champ[\b]4\champ[\h]2\envies1,\quad
2\champ[\b]4\champ[\h]2,\quad
4\champ[\b]4.
$$
Agent 4 is along each of the PI-cycles described above, hence by Lemma~\ref{lem:pareto-improvable} there exists an EFX allocation in which agent 4 is better off, and we are done. Thus, we assume agent 3 is the unique $\b$-champion of agent 4 in allocation $\allocshat$. We obtain the following structure:
\begin{center}
	\includegraphics[scale=\figurescale]{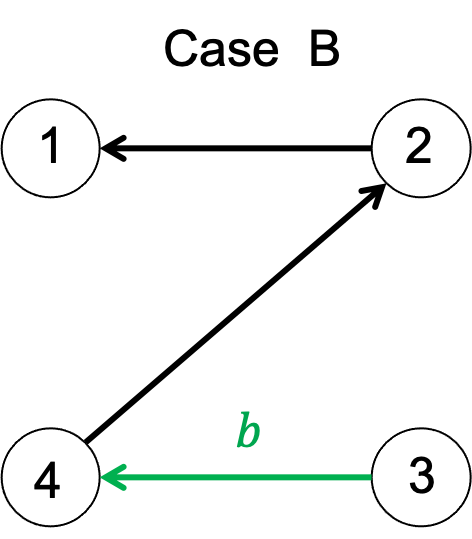}\qquad\qquad\qquad
	\includegraphics[scale=\figurescale]{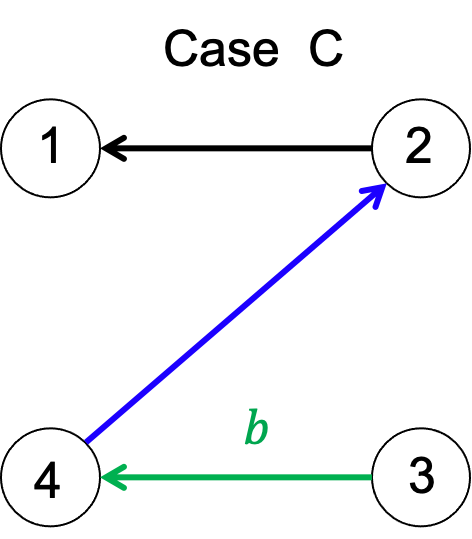}
\end{center}

We split the rest of our analysis of Cases \framebox{B} and \framebox{C} into three subcases:
\begin{center}
\circled{I} $X_4 >_1 \Tgtwo\cup\Bgthree;\quad$ \circled{II} $2\champ[\g]4$ in $\allocs;\quad$ \circled{III} $2\nchamp[\g]4$ in $\allocs$.
\end{center}

\vspace{1ex}
\noindent\textbf{Subcase} \circled{I} : $\boldsymbol{X_4 >_1 \Tgtwo\cup\Bgthree.}$
%
In this case, in any allocation where agent 1 gets $X_4$ or $Z$ (which is worth more than $X_4$ to agent 1), we can give $\Tgtwo\cup\Bgthree$ to another agent knowing that agent 1 will not envy it. (This is somewhat reminiscent of our approach in constructing allocation $X'$. There, the allocation had been EFX if $X_4 >_1 \Tgone\cup\Bgtwo$.)
Note that agent 3 envies $\Tgtwo\cup\Bgthree$ in $\allocs$ since $X_3= \Tgthree \cup \Bgthree <_3 \Tgtwo \cup \Bgthree$, by Observation~\ref{obs:T_k<T_j}
and \cancbty. Thus there exists an agent $i$ which is most envious agent of that set, {\sl i.e.}, $i \champ[\Bgthree \mid\circ]2$.  We ask who $i$ is, \emph{ignoring agent 1} (that is, the most envious agent of $\Tgtwo \cup \Bgthree$ restricted to agents 2,3,4).
Ignoring agent 1 is fine as long as agent 1 receives $X_4$ or $Z$ in the allocations we construct.

$i$ cannot be agent 4.  We show this by essentially using Observation \ref{obs:pred(i)_nchamp_B(i)} since $4 =$ ``$\pred(3)$'' in $M_{\allocs}$. While we do not have a good cycle in $M_{\allocs}$ and cannot really use that observation, its proof still applies:  we have $X_4>_4 X_3 = \Tgthree \cup \Bgthree >_4 \Tgtwo \cup \Bgthree$
where the first inequality is by $4\nenvies 3$ and the second inequality follows from Observation~\ref{obs:T_k<T_j} and \cancbty.  Hence, $4 \nchamp[\Bgthree \mid\circ]2$.  By the next claim, if $i=3$ we are done.

\begin{claim}\label{clm:3_2nd_MEA_implies_PI_cycle}
	If $i = 3$ then $M_\allocs$ contains a PI cycle.
\end{claim}
\shortversion
\begin{proof}
Consider the cycle $1\envies4\champ[\g]3\champ[\Bgthree\mid\circ]2\champ[\h]1$ in $M_\allocs$:
\begin{center}
	\includegraphics[scale=\figurescale]{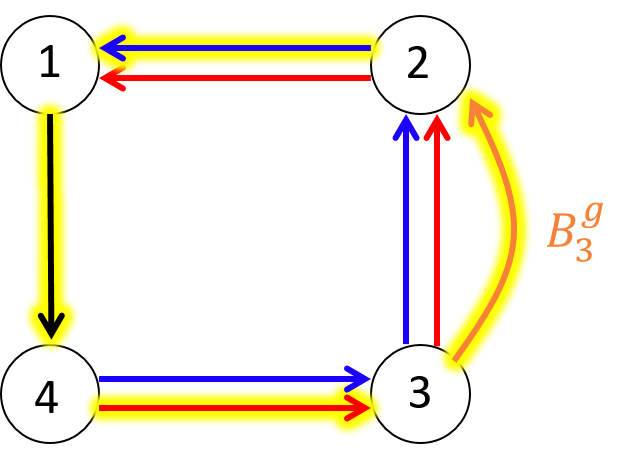}
\end{center}
This cycle might not really exist, since it could be the case that 1 is the real ($\Bgthree \mid \circ$)-champion of agent 2.  Assume for now that this is not the case, {\sl i.e.}, $3\champ[\Bgthree\mid\circ]2$ does exist in $M_{\allocs}$.  Then this is PI cycle, and by Lemma \ref{lem:pareto-improvable}, there is an allocation $\mathbf{Y}$ that Pareto-dominates $\allocs$.  Note that in $\mathbf{Y}$ agent 1 receives $X_4$.  We therefore claim that $\mathbf{Y}$ is EFX even if 1 was the most envious agent of $\Tgtwo \cup \Bgthree$ in $\allocs$.  Since 3 was the most envious agent ignoring agent 1, the only thing that could prevent $\mathbf{Y}$ from being EFX is if 1 strongly envies 3 in $\mathbf{Y}$.  But this is not the case since
$$ Y_1 = X_4 >_1 \Tgtwo \cup \Bgthree \geq_1 Y_3$$ where the last inequality holds since $Y_3 \subseteq \Tgtwo \cup \Bgthree$.  The claim follows.
\end{proof}
\shortversionend
\fullversion
\noindent The proof of Claim~\ref{clm:3_2nd_MEA_implies_PI_cycle} has been deferred to Appendix~\ref{apx:4-agents}.
\fullversionend

We are left with the case $i = 2$.  In the transition to $\allocshat$ from $\allocs$, agent 2 became worse off while the rest of the agents did not.  Thus, ignoring agent 1, agent 2 is also most envious of $\Tgtwo \cup \Bgthree$ in $\allocshat$.  Since $\hX_1 = Z >_1 X_4 >_1 \Tgtwo \cup \Bgthree$, agent 1 cannot be the real most envious agent in $\allocshat$ and thus agent 2 is the real one.  Since $\hX_3 = \Tgtwo \cup \{\g\}$, this means that $2 \champ[\Bgthree \mid \g] 3$ in $\allocshat$ and we obtain a
PI cycle in both cases $\framebox{B}$, $\framebox{C}$:
\begin{center}
\includegraphics[scale=\figurescale]{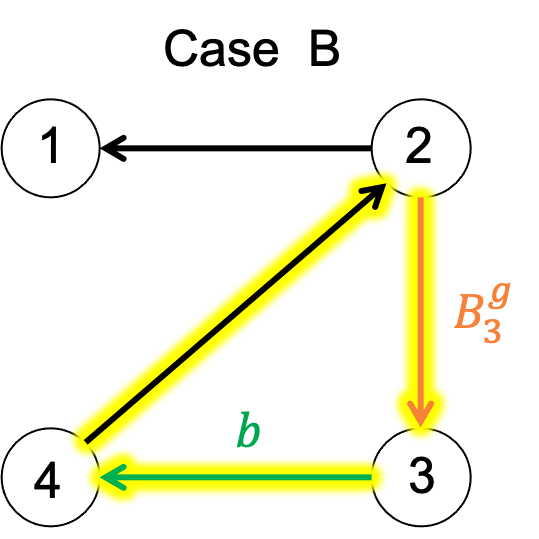}\qquad\qquad\qquad
\includegraphics[scale=\figurescale]{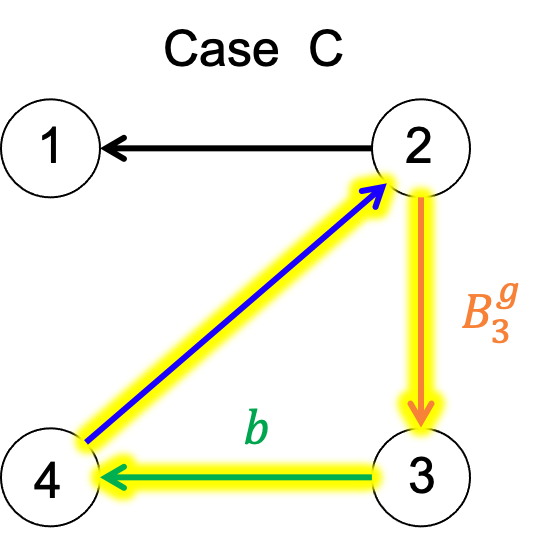}
\end{center}
Note that $\Bgthree$ can indeed be allocated to agent 2, as $\Bgthree\subseteq \Bhthree$ (by Equation \eqref{eq:B_3subsetB'_3}) and $\Bhthree$ is available in both cases (unallocated in Case \framebox{B} and released by the edge $4 \champ[\h] 2$ in Case \framebox{C}).
By Lemma~\ref{lem:pareto-improvable}, since both PI-cycles pass through agent 4, it follows that there exists a partial EFX allocation $\allocshat'$ that Pareto dominates $\allocshat$, in which agent 4 is better off. Hence, $\allocshat'$ dominates $\allocs$ in case agent $\vip$ is agent 4, as desired.

\vspace{1ex}
\noindent\textbf{Subcase} \circled{II} : $\boldsymbol{2\champ[\g]4}$ \textbf{in the allocation} $\boldsymbol{\allocs.}$
In this case, $M_\allocs$ restricted to basic edges has the structure:
\begin{center}
	\includegraphics[scale=\figurescale]{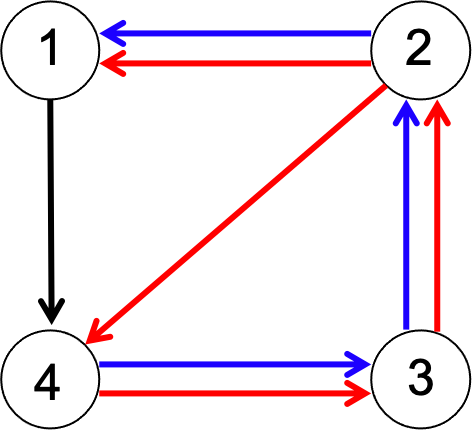}
\end{center}

We start with a proof sketch:  Notice that $2\champ[\g]4\champ[\g]3\champ[\g]2$ is a good $\g$-cycle in $M_{\allocs}$, hence $2$ $\g$-decomposes $4$ into $X_4=\Tgfour\cupdot\Bgfour$.
We first show that $X_4>_1\Tgtwo\cup \Bgfour$ (see Lemma \ref{lem:T_2<_1_T_4}). This allows us to apply a similar approach to the previous case: If agent 1 gets $Z$ (which is better for her than $X_4$), then we can safely give $\Tgtwo\cup \Bgfour$ to another agent knowing that agent 1 will not envy.

As in the previous case, we need to show that there is an agent other than 1 that envies the set $\Tgtwo\cup\Bgfour$. (In the previous case the fact that 3 envies $\Tgtwo\cup\Bgthree$ was immediate.)  Here, the assumption that we are not in Subcase \circled{I} allows us to reach this conclusion (see Lemma \ref{lem:X_3<_3_T_2_cup_B_4}).

Having shown that, we then find a PI-cycle in the sub-allocation $\allocs_{-1}=\langle X_2,X_3,X_4 \rangle$, the application of which leaves $Z$ unallocated (see Lemma \ref{lem:Y_-1_PD_X_-1}).  Finally, we can then allocate $Z$ to agent 1 to obtain an EFX allocation that Pareto-dominates $\allocs$ (see Corollary \ref{cor:Y_PD_X}).

%
\begin{lemma}\label{lem:T_2<_1_T_4}
	$X_4 >_1 \Tgtwo \cup \Bgfour$.
\end{lemma}
\shortversion
\begin{proof}
Consider the good $\g$-cycle $2 \champ[\g] \linebreak[1] 4 \champ[\g] \linebreak[1] 3 \champ[\g] \linebreak[1] 2$ in $\allocs$.
By Lemma~\ref{lem:alg_start}, there exists an agent $i$ such that $i\champ[\Bgtwo\mid\circ]4$.
$i$ cannot be agent 3, because $3=\pred(2)$ in this cycle (Observation~\ref{obs:pred(i)_nchamp_B(i)}). If $i=2$, we obtain the PI-cycle $2\champ[\Bgtwo\mid\circ]4\champ[\h]3\champ[\g]2$ in the original allocation $\allocs$, so we are done. If $i=4$, then
$%
	 X_4<_4 \Tgfour \cup \Bgtwo <_4 \Tgthree \cup \Bgtwo,%
$
where the second inequality is due to Observation~\ref{obs:T_k<T_j} and \cancbty. Therefore, we are done by Claim~\ref{clm:X_4<_4T_3+B_2}.

Thus, we may assume that $i=1$, {\sl i.e.}, $1\champ[\Bgtwo\mid\circ]4$. It follows that $X_1 <_1 \Tgfour \cup \Bgtwo$. Therefore, since $1\nenvies2$,
$$
	\Tgtwo\cup\Bgtwo=X_2<_1 X_1 <_1 \Tgfour \cup \Bgtwo,
$$
and we conclude that $\Tgtwo <_1 \Tgfour$ by \cancbty.  Finally we use \cancbty\ again to conclude $X_4 = \Tgfour \cup \Bgfour >_1 \Tgtwo \cup \Bgfour$, as desired.
\end{proof}
\shortversionend


\begin{lemma}\label{lem:X_3<_3_T_2_cup_B_4}
	Agent 3 envies $\Tgtwo \cup \Bgfour$ in $\allocs$.
\end{lemma}
\shortversion
\begin{proof}
Consider the good $\g$-cycle $2 \champ[\g] \linebreak[1] 4 \champ[\g] \linebreak[1] 3 \champ[\g] \linebreak[1] 2$ in $\allocs$.  We first ask which agent $i$ satisfies $i\champ[\Bgfour\mid\circ]3$ (there exists such an agent by Lemma~\ref{lem:alg_start}).
$i$ cannot be agent 2, because $2=\pred(4)$ in the good $\g$-cycle (Observation~\ref{obs:pred(i)_nchamp_B(i)}). If $i$ is agent 4, we obtain the PI-cycle $2\champ[\g]4\champ[\Bgfour\mid\circ]3\champ[\h]2$ in $M_\allocs$, hence we are done.

If $i = 1$, then in particular we have
$$
	\Tgthree \cup \Bgthree = X_3  <_1 X_1 <_1 \Tgthree\cup\Bgfour,
$$
where the first inequality is due to $1\nenvies3$ and the second holds since $1\champ[\Bgfour\mid\circ]3$.  By \cancbty\ we get $\Bgthree <_1 \Bgfour$, and together with Lemma \ref{lem:T_2<_1_T_4} we conclude that
$
X_4 = \Tgfour\cup\Bgfour>_1 \Tgtwo\cup \Bgfour>_1\Tgtwo\cup\Bgthree,
$
{\sl i.e.}, we are in Case \circled{I} which we already solved.

Thus, we may assume that $i$ is agent 3, i.e., $3\champ[\Bgfour\mid\circ]3$. It follows that
$$
	X_3 <_3 \Tgthree\cup\Bgfour <_3 \Tgtwo\cup\Bgfour,
$$
where the second inequality is by Observation~\ref{obs:T_k<T_j} and \cancbty.  The lemma follows.
\end{proof}
\shortversionend
\fullversion
\noindent The proofs of Lemmata \ref{lem:T_2<_1_T_4} and \ref{lem:X_3<_3_T_2_cup_B_4} have been deferred to Appendix \ref{apx:4-agents}.
\fullversionend

Consider the 3-agent suballocation $\allocs_{-1}=\langle X_2,X_3,X_4 \rangle$.

\begin{lemma}\label{lem:Y_-1_PD_X_-1}
	There is an EFX allocation $\mathbf{Y_{-1}}$ over agents 2,3,4 that Pareto-dominates $\allocs_{-1}$, in which $Z$ remains unallocated.
\end{lemma}

\begin{proof}

By Lemma \ref{lem:X_3<_3_T_2_cup_B_4}, agent 3 envies $\Tgtwo \cup \Bgfour$.
In particular,
there exists an agent $i\in\{2,3,4\}$ such that $i\champ[\Bgfour\mid\circ]2$ in $\allocs_{-1}$. $i$ cannot be agent 2, because $2=\pred(4)$ in the good $\g$-cycle (Observation~\ref{obs:pred(i)_nchamp_B(i)}). If $i$ is agent 3, we obtain the PI-cycle $3\champ[\Bgfour\mid\circ]2\champ[\g]4\champ[\h]3$ in $M_{\allocs_{-1}}$. If $i$ is agent 4, we obtain the PI-cycle $4\champ[\Bgfour\mid\circ]2\champ[\g]4$ in $M_{\allocs_{-1}}$.

\begin{center}
	\includegraphics[scale=\figurescale]{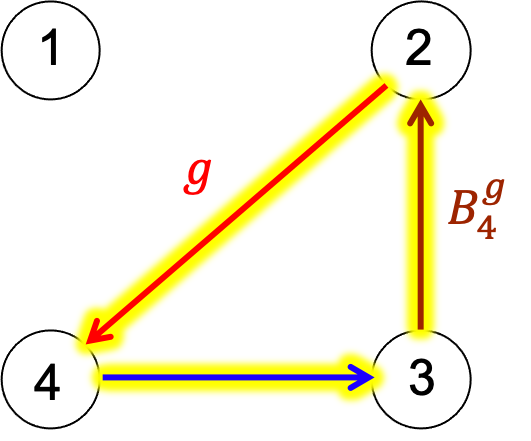}
	\qquad\qquad\qquad
	\includegraphics[scale=\figurescale]{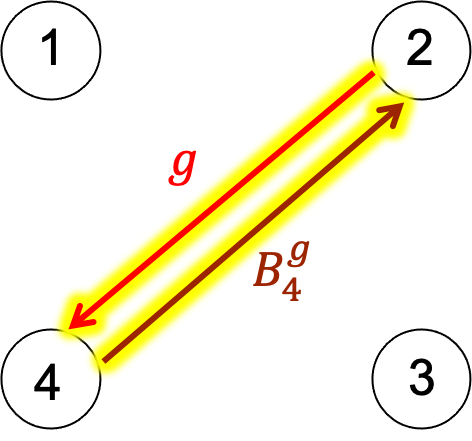}
\end{center}

In both cases, by Lemma~\ref{lem:pareto-improvable} there exists an EFX allocation $\mathbf{Y_{-1}}$ over agents 2,3 and 4 that Pareto dominates the allocation $\allocs_{-1}$.
Moreover, in both cases the set $Z (\subseteq \Tgone\cup\Bgtwo)$ is unallocated in $\mathbf{Y_{-1}}$ and we are done.
\end{proof}

\begin{corollary}\label{cor:Y_PD_X}
	The allocation $\mathbf{Y}$ obtained from $\mathbf{Y_{-1}}$ by allocating $Z$ to agent 1 is EFX and Pareto-dominates $\allocs$.
\end{corollary}

\begin{proof}
%
First, $\mathbf{Y}$ Pareto-dominates $\allocs$ since $\mathbf{Y_{-1}}$ Pareto-dominates $\allocs_{-1}$ and $Y_1 = Z >_1 X_4 >_1 X_1$.
We now show that $\mathbf{Y}$ is EFX.
The allocation $\mathbf{Y_{-1}}$ is EFX by construction, thus it suffices to show that agent 1 does not strongly envy any other agent, and that no agent strongly envies agent 1.

\noindent\textbf{Agent 1 does not strongly envy anyone:} Recall that $Z>_1 X_4 >_1 X_1$.
The possible bundles in $\mathbf{Y_{-1}}$ are $\Tgfour\cup \{\g\}$, $\Ththree\cup\{\h\}$, $X_3$ and a subset of $\Tgtwo\cup\Bgfour$. Agent 1 did not strongly envy the first two bundles in the original allocation $\allocs$, by definition of champion, and did not envy $X_3$. Since she is now better off she does not strongly envy those bundles in $\mathbf{Y}$.
Finally, agent 1 does not envy a subset of $\Tgtwo\cup\Bgfour$, since
$$
	Z>_1 X_4 >_1 \Tgtwo\cup \Bgfour,
$$
where the last inequality is due to
Lemma \ref{lem:T_2<_1_T_4}.
Therefore, agent 1 does not strongly envy any other agent in $\mathbf{Y}$. 

\noindent\textbf{No agent strongly envies agent 1:}
Agent 2 does not strongly envy 1 since she did not strongly envy the bundle $Z$ in allocation $\allocshat$, where she was worse off.
%
We have already shown that $X_3 >_3 \Tgone\cup\Bgtwo (\supseteq Z)$ and $X_4 >_4 \Tgone\cup\Bgtwo (\supseteq Z)$ when we constructed the allocation $\allocs'$. Thus, since neither agent 3 nor 4 became worse off in the transition from allocation $\allocs$ to $\mathbf{Y}$, we conclude that agents 3 and 4 do no envy agent 1.
\end{proof}

\vspace{1ex}
\noindent\textbf{Subcase} \circled{III} : $\boldsymbol{2\nchamp[\g]4}$ \textbf{in the allocation} $\boldsymbol{\allocs.}$
In this case we must have $1 \champ[\g] 4$, and $M_\allocs$ restricted to basic edges has the structure:
\begin{center}
	\includegraphics[scale=\figurescale]{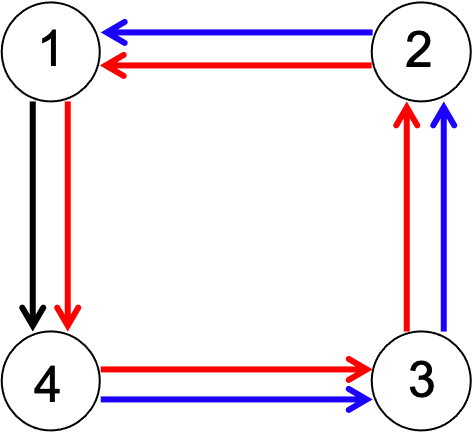}
\end{center}

We now briefly sketch our strategy to find a PI cycle in $M_{\allocshat}$ that includes agent 4.  Since agent 1 is not only a $\g$-champion of agent 4 in $\allocs$, but also envies her, it might be the case that this is a trivial championship, {\sl i.e.}, $\discard[\g]{1}{4} = \{\g\}$.  Such a championship relationship is not very useful, since no items are released from this edge that can used in other generalized championship edges.  Our first step is to show that this is not necessarily the case:  we show that agent 1 envies a strict subset of $X_4$ together with $\g$, that is,
we show that agent 1 $\g$-decomposes $4$ into top and bottom half-bundles $\Tgfour$ and $\Bgfour$ (unless there exists a PI cycle in $\allocs$). This is established in Lemma \ref{lem:1_decomposes_4}.

Given this assumption, our next step is Lemma \ref{lem:i-champB_3|->4} in which we show that in the allocation $\allocshat$ there is an agent $i\neq 4$ such that $i \champ[\Bgthree \mid \Bgfour] 4$.  When $i \in \{1,2\}$, we immediately a PI cycle that includes agent 4 in both cases $\framebox{B}$ and $\framebox{C}$, as depicted below:
\begin{center}
	\includegraphics[scale=\figurescale]{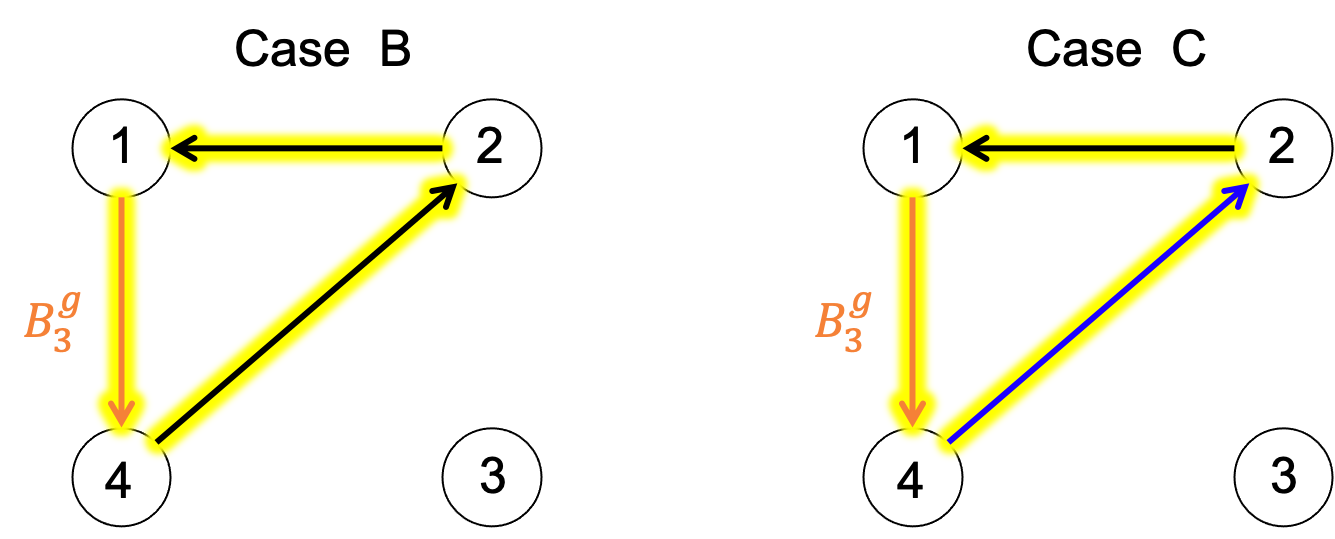}
\end{center}

\begin{center}
	\includegraphics[scale=\figurescale]{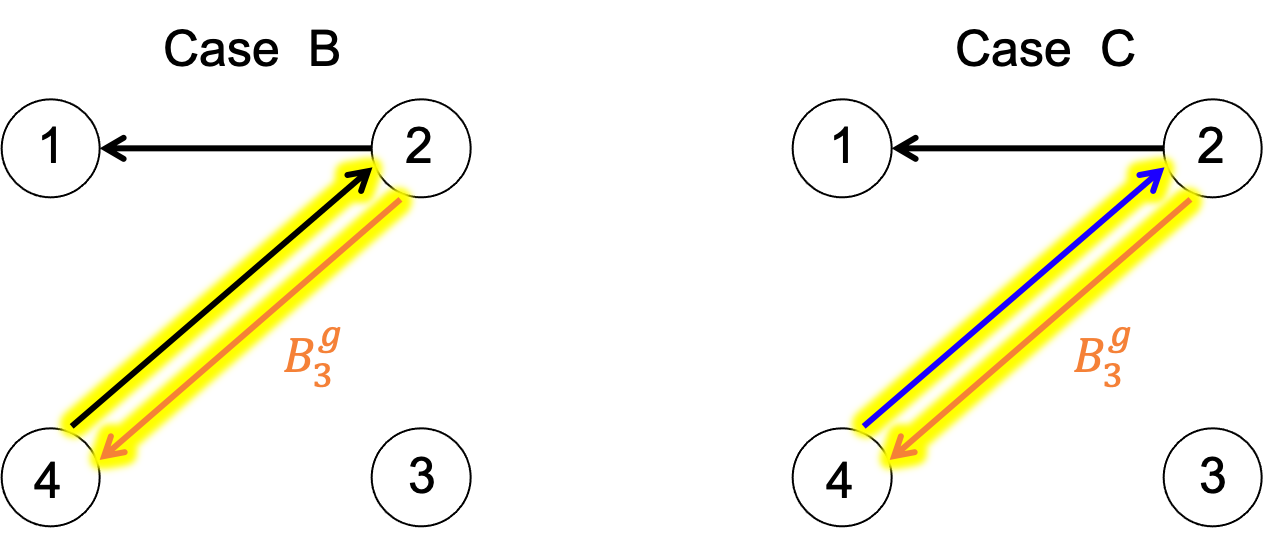}
\end{center}

These are indeed PI cycles since $\Bhthree$ is unallocated in Case $\framebox{B}$ and released by the $4 \champ[\h]2$ edge in Case $\framebox{C}$, and $\Bgthree \subseteq \Bhthree$ by Equation \eqref{eq:B_3subsetB'_3}.  The problematic case is when $i=3$, giving us the structure
\begin{center}
	\includegraphics[scale=\figurescale]{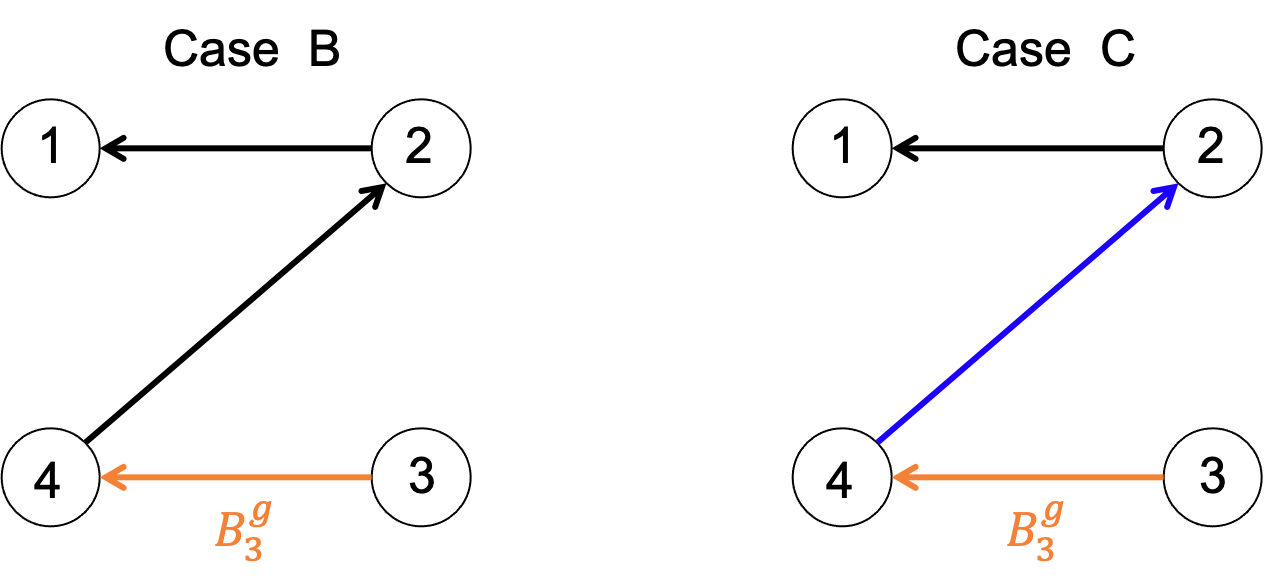}
\end{center}

To solve this case we show, in Lemma \ref{lem:2-champB_4|->2}, that $2\champ[\Bgfour \mid \g]3$ in $\allocshat$.  Since $\Bgfour$ is released by the $3 \champ[\Bgthree \mid \Bgfour]4$ edge, the new edge indeed closes a PI cycle that includes agent 4 in both cases $\framebox{B}$ and $\framebox{C}$:
\begin{center}
	\includegraphics[scale=\figurescale]{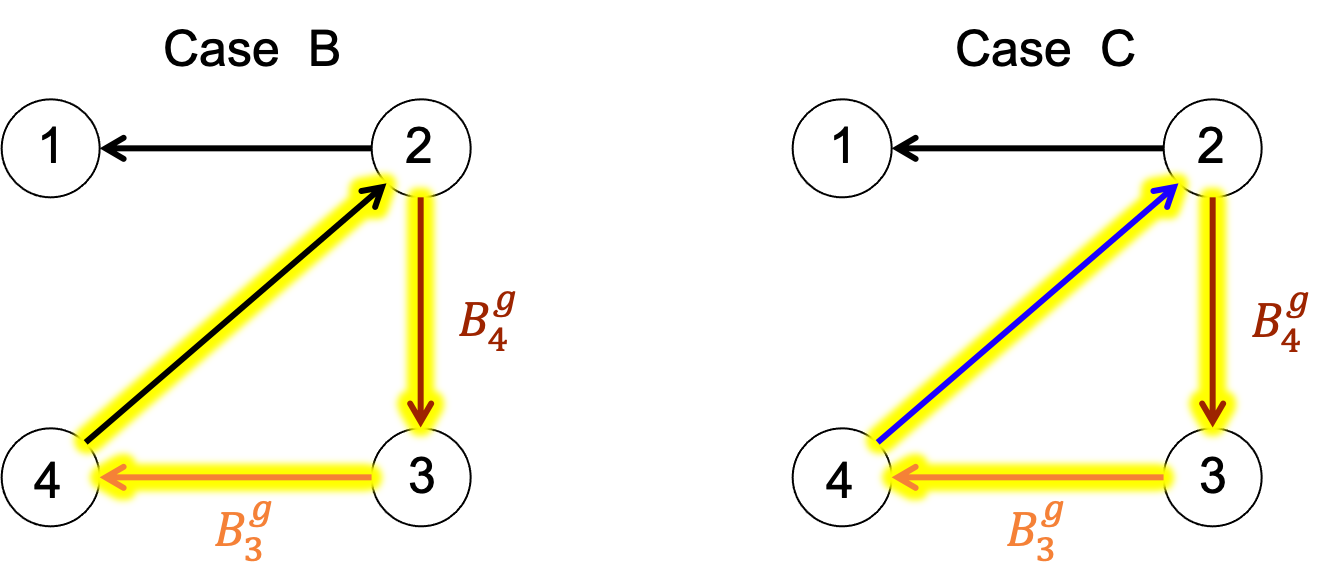}
\end{center}


\begin{lemma}\label{lem:1_decomposes_4}
	Agent 1 $\g$-decomposes agent 4 in the allocation $\allocs$.
\end{lemma}
\begin{proof}
Consider the allocation $\allocshat$. Recall that in the allocation $\allocshat$ agent 3 is the unique $\b$-champion of agent 4. In particular we have $4\nchamp[\b]4$, implying $\discard[\b]{3}{4} \neq \emptyset$.  Furthermore $3\champ[\b]4$, and together with the fact that $3 \nenvies 4$ (since $\hX_3 = \Tgtwo \cup \{\g\} >_3 X_3 >_3 X_4 = \hX_4$) this implies $\b \notin \discard[\b]{3}{4}$ by Observation \ref{obs:g_notin_bottom}.
We conclude that there exists some item $d \in \discard[\b]{3}{4}$ which is not $\b$, {\sl i.e.}, $d \in \hX_4 = X_4$.
We obtain
\begin{equation}\label{eq:X_3<(X_3-d)+b}
  X_3 \leq_3 \hX_3 	<_3 (X_4\cup\{\b\})\setminus \discard[\b]{3}{4}
  					\leq_3 (X_4\cup \{\b\})\setminus \{d\}
  					= (X_4\setminus \{d\})\cup \{\b\},
\end{equation}

Where the first inequality holds since agent 3 is better off in allocation $\allocshat$, the second holds by definition of basic championship, the third holds since valuations are monotone and the equality holds since $\b \neq d$.

By the following claim we may assume $\b <_3 \g$.
\fullversion
The proof of the claim has been deferred to Appendix \ref{apx:4-agents}.
\fullversionend
\begin{claim}\label{clm:b<_3_g} 
	$\b <_3 \g$, unless there exists a PI cycle in $M_{\allocs}$.
\end{claim}
\shortversion
\begin{proof}
Assuming that $\g <_3 \b$, We get
$X_3 <_3 \Tgtwo \cup \{\g\} <_3 \Tgtwo \cup \{\b\}$, where the first inequality holds since $3 \champ[\g] 2$ in $\allocs$ and the second by \cancbty.
In particular there exists a most envious agent of $\Tgtwo \cup \{\b\}$ in $\allocs$, {\sl i.e.}, there exists an agent $i$ such that $i\champ[\b]2$.
$i$ cannot be 2,
since $X_2 >_2 X_1 = \Tgone \cup \Bgone >_2 \Tgtwo \cup \Bgone >_2 \Tgtwo \cup \{\b\}$, where the first inequality holds by $2\nenvies1$, the second by Observation~\ref{obs:T_k<T_j} and \cancbty,
and the third by monotonicity.
Thus $i = 1, 3$ or $4$ and in all cases we get a PI cycle in $M_\allocs$:
$$1\champ[\b]2\champ[\g]1~,~3\champ[\b]2\champ[\g]1\envies4\champ[\h]3~\text{ or }~4\champ[\b]2\champ[\g]1\envies4,$$
respectively.  Note that these are indeed PI cycles since the $\g$-champion edge releases $\Bgone$ which contains $\b$.
\end{proof}
\shortversionend

By Claim \ref{clm:b<_3_g} and Equation \eqref{eq:X_3<(X_3-d)+b} we obtain
$%
	X_3	<_3(X_4\setminus\{d\})\cup \{\b\}
		<_3(X_4\setminus\{d\})\cup \{\g\},%
$
where the second inequality holds by \cancbty. Since agent 3 envies $(X_4\setminus\{d\})\cup\{\g\}$, there exists a most envious agent of this set. Let $i\in[4]$ be that agent. Clearly, $i\champ[\g]4$.
Since we are in Subcase \circled{III}, $i\neq 2$.
If $i = 3$ or $i =4$ we obtain a PI-cycle: $3\champ[\g]4\champ[\h]3$, $4\champ[\g]4$, respectively. Hence, $i=1$. Therefore, there exists a choice of $\discard[\g]{1}{4}$ that contains $d$. 
Since the set $\discard[\g]{1}{4}$ has not been used in this scenario of our proof thus far, we may define it such that $d\in\discard[\g]{1}{4}$, rather than arbitrarily.
Finally, we have that $\g\notin\discard[\g]{1}{4}$, as otherwise $X_1<_1(X_4\cup\{\g\})\setminus\{d,\g\}= X_4\setminus \{d\}$, in contradiction to allocation $\allocs$ being EFX.
This shows that agent 1 $\g$-decomposes agent 4, as claimed.
\end{proof}

By Lemma~\ref{lem:1_decomposes_4} agent 1 $\g$-decomposes $4$ into top and bottom half-bundles denoted $\Tgfour$ and $\Bgfour$. 
The next lemma reveals a generalized champion edge in $M_{\allocshat}$ that in most cases immediately closes a PI-cycle that includes agent 4.
\fullversion
The proof of the lemma has been deferred to Appendix \ref{apx:4-agents}.
\fullversionend

\begin{lemma}\label{lem:i-champB_3|->4}
	There exists $i\in[4]$ such that $i\champ[\Bgthree\mid \Bgfour]4$ in $M_{\allocshat}$.
\end{lemma}
\shortversion
\begin{proof}
	Recalling that $\hX_4 = X_4$, we need to show that there is a most envious agent of $\Tgfour \cup \Bgthree$ in $\allocshat$. We do this by showing $\hX_1	<_1 \Tgfour\cup \Bgthree $.
	If $\Bgtwo >_1 \Bgthree$, then
	$$
	X_4 >_1 X_1 >_1 X_2 = \Tgtwo \cup \Bgtwo >_1 \Tgtwo \cup \Bgthree,
	$$
	where the first and second inequalities are by $1\envies4$ and $1\nenvies2$, respectively, in $M_{\allocs}$,
	and the final inequality follows from \cancbty. This implies that we are in Case \circled{I} that has been solved earlier.

	Thus we may assume $\Bgtwo <_1 \Bgthree$.  Moreover $\Tgone <_1 \Tgfour$ by Observation \ref{obs:T_k<T_j}. Therefore, by \cancbty
	$$
	\Tgone\cup \Bgtwo <_1 \Tgfour \cup \Bgtwo <_1 \Tgfour \cup \Bgthree.
	$$

	Seeing as $\hX_1=Z\subseteq \Tgone\cup\Bgtwo$ and valuations are monotone, the above equation implies that agent 1 envies $\Tgfour \cup \Bgthree$ in allocation $\allocshat$, as desired.
\end{proof}
\shortversionend
Consider the agent $i$ that satisfies $i\champ[\Bgthree\mid \Bgfour]4$ in $M_{\allocshat}$ which is guaranteed by Lemma~\ref{lem:i-champB_3|->4}. First, $i \neq 4$
since
$$
\hX_4 = X_4 >_4 X_3 = \Tgthree \cup \Bgthree >_4 \Tgfour \cup \Bgthree,
$$
where the first inequality holds by $4 \nenvies 3$ in $M_\allocs$ and the second holds by Observation \ref{obs:T_k<T_j} and \cancbty.

If $i = 1$ or $i = 2$ then we obtain a PI-cycle that includes agent 4.  To see that the cycles we propose are indeed PI-cycles, recall that $\Bgthree\subseteq \Bhthree$ by Equation~\eqref{eq:B_3subsetB'_3} and in Case \framebox{B} the set $\Bhthree$ remains unallocated while in Case \framebox{C} the set $\Bhthree$ is released by the edge $4\champ[\h]2$ (which is included in all of the PI-cycles we describe next).
If $1\champ[\Bgthree\mid \Bgfour]4$ in $M_{\allocshat}$, we obtain
the PI-cycle $1\champ[\Bgthree\mid\Bgfour]4\envies2\envies1$ in Case \framebox{B}, and
the PI-cycle $1\champ[\Bgthree\mid\Bgfour]4\champ[\h]2\envies1$ in Case \framebox{C}, as illustrated below.
\begin{center}
	\includegraphics[scale=\figurescale]{figures/With_envy_p16-5}
\end{center}
If $2\champ[\Bgthree\mid \Bgfour]4$ in $M_{\allocshat}$, we obtain
the PI-cycle $2\champ[\Bgthree\mid\Bgfour]4\envies2$ in Case \framebox{B}, and
the PI-cycle $2\champ[\Bgthree\mid\Bgfour]4\linebreak[1]\champ[\h]2$ in Case \framebox{C}, as illustrated below.
\begin{center}
	\includegraphics[scale=\figurescale]{figures/With_envy_p16-6}
\end{center}

The remaining case is $i=3$, giving us the structure
\begin{center}
	\includegraphics[scale=\figurescale]{figures/With_envy_p16-7}
\end{center}
To solve this case we next establish (see Lemma \ref{lem:2-champB_4|->2}) the existence of the edge $2 \champ[\Bgfour \mid \g] 3$ (recall that $\hX_3 = \Tgtwo \cup \{\g\}$), giving us
the PI-cycle $3\champ[\Bgthree\mid\Bgfour]4\envies2\champ[\Bgfour\mid\g]3$ in Case \framebox{B}, and
the PI-cycle $1\champ[\Bgthree\mid\Bgfour]4\champ[\h]2\champ[\Bgfour\mid\g]3$ in Case \framebox{C}, as illustrated below.
\begin{center}
	\includegraphics[scale=\figurescale]{figures/With_envy_p16-2}
\end{center}
In all cases the PI-cycle passes through agent 4, thus by Lemma~\ref{lem:pareto-improvable}, there exists a partial EFX allocation $\allocshat'$ in which agent 4 is strictly better off. Since we assumed that agent $\vip$ is 4, it follows that $\allocshat'$ dominates $\allocs$, as desired.  The following lemma finishes this case.

\begin{lemma}\label{lem:2-champB_4|->2}
	$2\champ[\Bgfour\mid \g] 3$	in $M_{\allocshat}$.
\end{lemma}
\begin{proof}
First we prove the following auxiliary claim:
\fullversion
(its proof is deferred to Appendix \ref{apx:4-agents})
\fullversionend
\begin{claim}\label{clm:1-B_2|->4}
	$1\champ[\Bgtwo\mid\circ]4$ in $M_{\allocs}$, unless there exists a PI-cycle in $M_{\allocs}$.
\end{claim}
\shortversion
\begin{proof}
	Recall that $X_1 <_1 X_4 <_1 Z \subseteq \Tgone\cup\Bgtwo$.
	We thus have
	$%
	X_1 <_1 \Tgone\cup\Bgtwo <_1 \Tgfour\cup\Bgtwo,%
	$
	by Observation~\ref{obs:T_k<T_j} and \cancbty. Since agent 1 envies $\Tgfour\cup\Bgtwo$, there exists a most envious agent of this set. That is, there exists $j\in[4]$, such that $j\champ[\Bgtwo\mid\circ]4$.

	If $2\champ[\Bgtwo\mid\circ]4$, then we obtain the PI-cycle
	$2\champ[\Bgtwo\mid\circ]4\champ[\h] 3 \champ[\g] 2$.
	%
%
	Note that $3\nchamp[\Bgtwo\mid\circ]4$,
	since
	$$%
	X_3>_3 X_2 = \Tgtwo\cup\Bgtwo >_3 \Tgfour\cup \Bgtwo,%
	$$
	where the first inequality is by $3\nenvies2$ and the second is by Observation~\ref{obs:T_k<T_j} and \cancbty.

	Similarly, $4\nchamp[\Bgtwo\mid\circ]4$,
	since $X_4>_4 \Tgthree\cup\Bgtwo$ (otherwise we are done by Claim~\ref{clm:X_4<_4T_3+B_2}), and therefore, by Observation~\ref{obs:T_k<T_j} and \cancbty,
	$%
	X_4>_4 \Tgthree \cup\Bgtwo >_3 \Tgfour \cup\Bgtwo.%
	$
	The claim now follows.
\end{proof}
\shortversionend
Now, since $\hX_3 = \Tgtwo \cup \{\g\}$, we need to show that agent 2 is the most envious agent of $\Tgtwo \cup \Bgfour$.  In the transition from allocation $\allocs$ to $\allocshat$, agent 2 is worse off while all other agents are not, and thus it is enough to show that agent 2 was most envious of $\Tgtwo \cup \Bgfour$ also in $\allocs$.

We first show that there is \emph{some} agent who envies $\Tgtwo \cup \Bgfour$.  Note that if $\Bgtwo <_2 \Bgfour$ then $X_2 = \Tgtwo \cup \Bgtwo <_2 \Tgtwo\cup \Bgfour$, {\sl i.e.}, agent 2 envies $\Tgtwo \cup \Bgfour$.
The following claim helps us deal with the case $\Bgfour <_2 \Bgtwo$.
\fullversion
\noindent Its proof has been deferred to Appendix \ref{apx:4-agents}.
\fullversionend

\begin{claim}\label{clm:3_champ_B_4_3_V_2_champ_B_4_3}
	If  $\Bgfour <_2 \Bgtwo$ then $3\champ[\Bgfour\mid\circ]3$ in $M_\allocs$, unless there exists a PI-cycle in $M_{\allocs}$.
\end{claim}

\shortversion
\begin{proof}
	By Claim \ref{clm:1-B_2|->4} we have the following structure in $M_\allocs$:
	\begin{center}
		\includegraphics[scale=\figurescale]{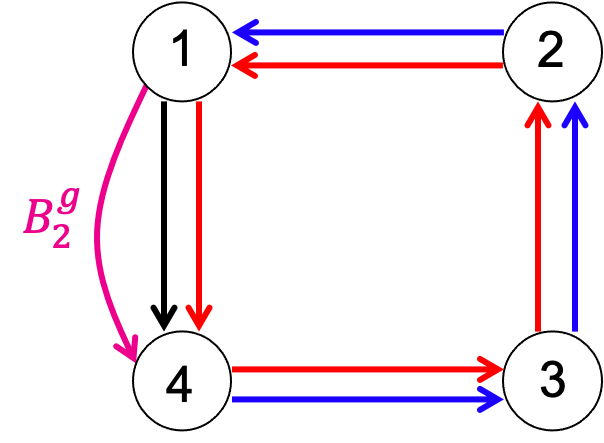}
	\end{center}

	Now, we have
	$%
	X_4 = \Tgfour \cup \Bgfour <_4 \Tgthree \cup \Bgfour,%
	$
	by Observation~\ref{obs:T_k<T_j} and \cancbty. Since agent 4 envies $\Tgthree\cup\Bgfour$, there exists a most envious agent of this set. That is, there exists $i\in[4]$, such that $i\champ[\Bgfour\mid\circ]3$.
	If $i=1$, then
	$%
	\Tgthree \cup \Bgfour >_1 X_1 >_1 X_3 = \Tgthree \cup \Bgthree,%
	$
	where the first inequality holds by definition of $i$ and the second inequality follows from $1\nenvies3$. By \cancbty, we conclude that
	$
	\Bgfour >_1 \Bgthree,
	$ 
	and together with $\Tgfour>_1 \Tgtwo$ (Observation~\ref{obs:T_k<T_j}) we get
	$$
	X_4=\Tgfour\cup\Bgfour >_1 \Tgfour\cup\Bgthree >_1 \Tgtwo\cup\Bgthree.
	$$
	by \cancbty. Thus we are in Subcase~\circled{I} which has been solved earlier.

	Suppose next that $i=4$. 
	By Claim~\ref{clm:1-B_2|->4}, $1\champ[\Bgtwo\mid\circ]4$, thus we obtain the PI-cycle
	$1\champ[\Bgtwo\mid\circ]4\champ[\Bgfour\mid\circ]3\champ[\g]2\champ[\h]1$
	depicted below.
	\begin{center}
		\includegraphics[scale=\figurescale]{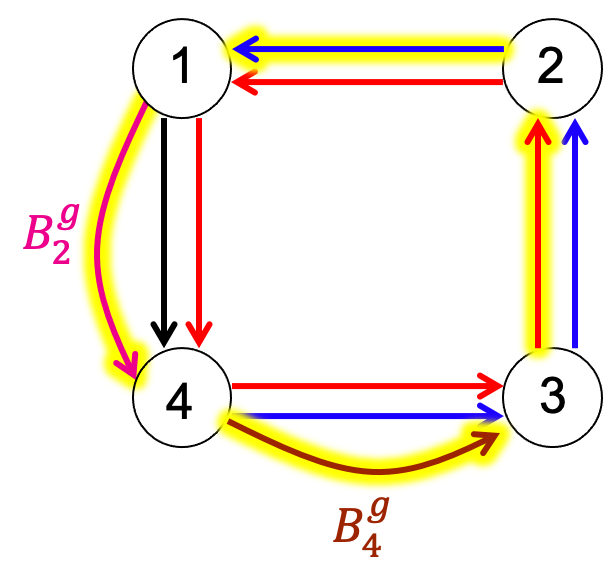}
	\end{center}
	To see that it is indeed a PI-cycle, notice that the set $\Bgtwo$ is released by $3\champ[\g]2$ and the set $\Bgfour$ is released by the edge $1\champ[\Bgtwo\mid\circ]4$.

	Assume now that $i=2$.
	Since we assume that $\Bgfour <_2\Bgtwo$, we have
	$%
	X_2<_2 \Tgthree \cup \Bgfour <_2 \Tgthree \cup \Bgtwo,%
	$
	where the first inequality follows from $2\champ[\Bgfour\mid\circ]3$ and the second is by \cancbty. Therefore, agent 2 envies $\Tgthree \cup \Bgtwo$ and thus there exists an agent $k\in[4]$ such that $k\champ[\Bgtwo\mid\circ]3$.
	Note that $3\nchamp[\Bgtwo\mid\circ]3$, since
	$%
	X_3>_3 X_2 =\Tgtwo \cup\Bgtwo >_3 \Tgthree \cup\Bgtwo,%
	$
	where the first inequality is by $3\nenvies2$ and the second inequality is by Observation~\ref{obs:T_k<T_j} and \cancbty.
	Thus, we may assume that either $1\champ[\Bgtwo\mid\circ]3$, $2\champ[\Bgtwo\mid\circ]3$ or $4\champ[\Bgtwo\mid\circ]3$, in each of these cases we obtain a PI-cycle:

	$$1\champ[\Bgtwo\mid\circ]3\champ[\g]2\champ[\h]1~,~2\champ[\Bgtwo\mid\circ]3\champ[\g]2~ \text{ or }~4\champ[\Bgtwo\mid\circ]3\champ[\g]2\champ[\h]1\envies4, $$
	respectively.  We are left with the case $i=3$, and we are done.
\end{proof}
\shortversionend

Assume $\Bgfour <_2 \Bgtwo$.  Claim \ref{clm:3_champ_B_4_3_V_2_champ_B_4_3} implies that 
$%
	X_3 <_3 \Tgthree\cup\Bgfour <_3 \Tgtwo\cup\Bgfour,%
$
where the second inequality is by Observation~\ref{obs:T_k<T_j} and \cancbty. Hence, the set $\Tgtwo\cup\Bgfour$ is envied by agent 3,
and in particular there is a most envious agent of $\Tgtwo\cup\Bgfour$ in the allocation $\allocs$. To complete the proof of Lemma \ref{lem:2-champB_4|->2} we show that this agent must be 2.

Assume towards contradiction that 1 is the most envious agent of $\Tgtwo\cup\Bgfour$, then $\Tgtwo\cup\Bgtwo = X_2 <_1 X_1 <_1 \Tgtwo\cup\Bgfour$ (the first inequality follows by $1\nenvies2$). Thus,
$
  \Bgtwo<_1\Bgfour
$ 
by \cancbty.
Moreover, by Observation~\ref{obs:T_k<T_j},
$
  	\Tgone<_1\Tgfour.
$

Therefore, by \cancbty
$$
	Z\subseteq \Tgone\cup\Bgtwo <_1 \Tgone\cup\Bgfour <_1 \Tgfour \cup \Bgfour = X_4,
$$
in contradiction to the assumption $Z>_1 X_4$.

If agent 3 (resp., 4) is the most envious agent of $\Tgtwo\cup\Bgfour$, then $3\champ[\Bgfour\mid\circ]2$ (resp., $4\champ[\Bgfour\mid\circ]2$) closes the PI-cycle $3\champ[\Bgfour\mid\circ]2\champ[\h]1\champ[\Bgtwo\mid\circ]4\champ[\g]3$
(resp., $4\champ[\Bgfour\mid\circ]2\champ[\h]1\champ[\Bgtwo\mid\circ]4$) in $M_{\allocs}$. Recall that $1\champ[\Bgtwo\mid\circ]4$ by Claim~\ref{clm:1-B_2|->4}.

Thus, we may assume that agent 2 is the most envious agent of $\Tgtwo\cup\Bgfour$ and the lemma now follows.
\end{proof}
%


\fullversionend
\shortversion
\subsection{$\allocs$ is Not Envy-Free}\label{sec:not-ef-case}

In this case, $M_\allocs$ must contain at least one envy edge. In the following lemma we derive the unique non-trivial configuration of the basic edges ({\sl i.e.}, envy edges, $\g$-edges and $\h$-edges) in $M_\allocs$:

\begin{lemma}[$\star$]\label{lem:envy_unique_structure}
	 Assuming $M_\allocs$ contains no PI cycle, and up to renaming of the agents, the possible $\g$ and $\h$ champions of 4 are only agents 1 and 2, and all other basic edges in $M_\allocs$ are exactly as depicted in the following figure:
\vspace{-1em}
	 \begin{center}
		\includegraphics[scale=\figurescale]{figures/With_envy_p02-4}
	\end{center}
\end{lemma}


By Observation~\ref{obs:g_notin_bottom},
	agents 1,2 and 3 are $\g$ and $\h$ decomposed by their unique champion, 2,3 and 4, respectively.
In the following we denote the $\g$ (resp., $\h$) -decomposition by $X_j = T^g_j\cupdot B^g_j$ (resp. $X_j = T^h_j\cupdot B^h_j$), for $j\in\{1,2,3\}$.
%
Since we renamed the agents depending on the graph structure, agent $\vip$ can be any one of the agents. 
In what follows we assume that $\vip = 2$.  The case $\vip \neq 2$ is deferred to Appendix \ref{sec:4-agents-Not-EF} due to space limitations.
%
We start with the following simple observation.
\shortversionend

\fullversion
\subsubsection{$\boldsymbol{\vip = 2}$}
Recall that the structure of $M_\allocs$ (restricted to envy, $\g$ and $\h$ edges) is as follows:
\begin{center}
	\includegraphics[scale=\figurescale]{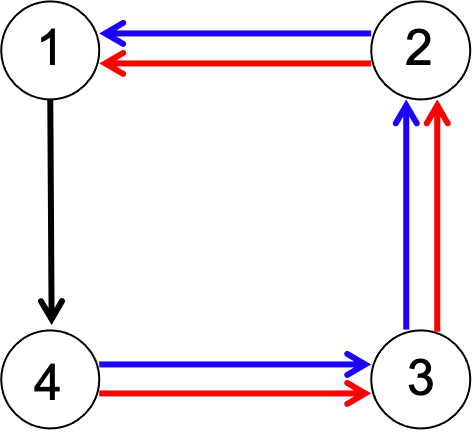}
\end{center}
with only the $\g$ and $\h$ edges going into agent 4 missing (where the source of either of these can be either 1 or 2).
\fullversionend

\begin{observation}\label{obs:1_4_prefer_X_4}
	For $i \in \{1,4\}$ we have $X_4 >_i \max_i\{X_2,\Thtwo\cup \{\h\}, \Tgone \cup \{\g\}, X_3\}$ 
\end{observation}

\begin{proof}
	We have $X_1 >_1 X_2$, $\Thtwo\cup \{\h\}$, $\Tgone \cup \{\g\}$, $X_3$ since $1 \nenvies 2$, by Observation \ref{obs:non-champion-doesnt-envy-top-half} (since $1 \nchamp[\h] 2$), again by Observation \ref{obs:non-champion-doesnt-envy-top-half} (since $1 \nchamp[\g] 1$), $1 \nenvies 3$, respectively. The symmetric claim for agent 4 holds for the same reasons.
\end{proof}
By the above observation, we can define for $i \in \{1,4\}$ the bundle $Z_i$ to be
the subset of $X_4$ obtained by iterative removal of the item of least marginal value to agent $i$ as long as the leftover bundle is still better than $\max_i\{X_2,\Thtwo\cup \{\h\}, \Tgone \cup \{\g\}, X_3\}$ for $i$.

By Lemma \ref{lem:remove_least_marginal_item_1_by_1}, $Z_i$ is a {\bf smallest size} subset of $X_4$ s.t. $Z_i >_i \max_i\{X_2,\Thtwo\cup \{\h\}, \Tgone \cup \{\g\}, X_3\}$.
\fullversion
Note that the fact that $Z_i$ is of minimal cardinality holds trivially under additive valuations. Lemma \ref{lem:remove_least_marginal_item_1_by_1} generalizes this fact to all \valclass\ valuations.

\fullversionend
Denote $w = 1, \ell = 4$ if $\left|Z_1\right| < \left|Z_4\right|$ and $w = 4, \ell = 1$ if  $\left|Z_1\right| \geq \left|Z_4\right|$.  Define the allocation $\allocs'$ as:

\begin{center}
\begin{minipage}{\textwidth}
\begin{center}
\renewcommand{\arraystretch}{1.2}
$\allocs'$=
	\bundle{w}{X_4}\quad
	\bundle{\ell}{\max_\ell\{X_2,\, \Thtwo \cup \h,\, \Tgone\cup \g,\, X_3 \}}\quad
	\bundle{2}{\begin{array}{c}
		\Tgone \cup \g \\
		\mbox{or}^\ast\; X_2
		\end{array} }\quad
	\bundle{3}{\begin{array}{c}
		X_3 \mbox{ or}^\dagger \\
		\Thtwo \cup \h
		\end{array} }\quad
\end{center}
\noindent{\scriptsize $\ast$ $X'_2=\Tgone \cup \{\g\}$ unless $X'_\ell=\Tgone \cup \{\g\}$, in which case $X'_2=X_2$;}\hfill
\noindent{\scriptsize $\dagger$ $X'_3=X_3$ unless $X'_\ell=X_3$, in which case $X'_3=\Thtwo \cup \{\h\}$.}
\end{minipage}
\end{center}


%

\noindent\textbf{The Only Possible Strong Envy in $\mathbf{X'}$ is From $\boldsymbol{\ell}$ to $\boldsymbol{w}$:} For $i\in\{2,3,w\}$, agent $i$ does not strongly envy any of the sets  $X_4, \Tgone \cup \{\g\}, \Thtwo \cup \{\h\}, X_2, X_3$, as agent $i$ is not worse off and she did not strongly envy  any of those before (by definition of basic championship, and the fact that $\allocs$ is EFX). Since this list includes all possible sets in $\allocs'$, it follows that agents $2,3$ and $w$ envy no other agent in $\allocs'$.
Finally, $\ell$ does not envy agents 2 and 3 since she receives her most valued bundle out of 4 options, (some of) the leftovers of which go to agents $2,3$.

\noindent\textbf{If $\boldsymbol{X'_2 = \Tgone \cup \{\g\}}$:}
In this case agent 2 is better off relative to $\allocs$.  Update $\mathbf{X'}$ by replacing $X'_w$ with $Z_w$.  Since $\left|Z_w\right| \leq \left|Z_\ell\right|$,
and since
$Z_\ell$ is a smallest size bundle such that $Z_\ell >_\ell \max_\ell\{X_2,\Thtwo\cup \{\h\}, \Tgone \cup \{\g\}, X_3\} = X'_\ell$ it follows immediately that $Z_w \setminus \{h\} <_\ell X'_\ell$ for any $h \in Z_w$.  Thus $\ell$ does not strongly envy $w$ and we obtain an EFX allocation where agent 2 is better off 
 and we are done.

\noindent\textbf{If $\boldsymbol{X'_2 = X_2}$:}  In this case the allocation $\mathbf{X'}$ is
\begin{center}
	\renewcommand{\arraystretch}{1.2}
	$\allocs'$=\bundle{w}{X_4}\quad
	\decomp{\ell}{\Tgone}{\g}\quad
	\bundle{2}{X_2}\quad
	\bundle{3}{X_3}	
\end{center}
and observe that $\h$ and all items of $\Bgone$ are unallocated.  Take some arbitrary item $\b \in \Bgone$
\shortversion 
(note that $\Bgone \neq \emptyset$ since otherwise, as valuations are non-degenerate, we have $X_1 <_1 X_1 \cup \{\g\} = \Tgone \cup \{\g\}$, in contradiction to Observation \ref{obs:non-champion-doesnt-envy-top-half} since $1 \nchamp[\g] 1$ in $\allocs$).
\shortversionend
\fullversion
. Recall that $\Bgone$ is non empty, as otherwise 1 would have been a self $\g$-champion in the original allocation $\allocs$.
\fullversionend 
We split to cases according to whether $\mathbf{X'}$ is EFX or not, {\sl i.e.}, whether $\ell$ strongly envies $w$ or does not. Note that even if $\mathbf{X'}$ is EFX we are not done yet since agent 2 is not better off relative to $\allocs$.

\vspace{1ex}
\noindent\textbf{Case: $\mathbf{X'}$ is EFX.}
\vspace{1ex}

In $\mathbf{X'}$ we have $2 \envies \ell$ since $2 \champ[\g] 1$ in $\allocs$.  We also have $\ell \envies w$ by Observation \ref{obs:1_4_prefer_X_4}.  Besides that there is no more envy in $\mathbf{X'}$.

Consider a $\b$-champion of agent 2 in $\mathbf{X'}$ (such exists by Observation \ref{obs:exists-champion}).  If $2$, $\ell$ or $w$ are such, then we get the following respective PI cycles that include agent 2:
$$2 \champ[\b] 2, 2 \envies \ell \champ[\b] 2, 2 \envies \ell \envies w \champ[\b] 2$$
and thus by Lemma \ref{lem:pareto-improvable} there is an EFX allocation $\mathbf{Y}$ in which agent 2 is better off relative to $\allocs$, and we are done (recall that $X'_2 = X_2$).  Thus, we can assume that $3$ is the unique $\b$-champion of 2,
and $M_{\mathbf{X'}}$ currently has the structure:
\shortversion
\vspace{-1em}
\shortversionend
\begin{center}
	\includegraphics[scale=\figurescale]{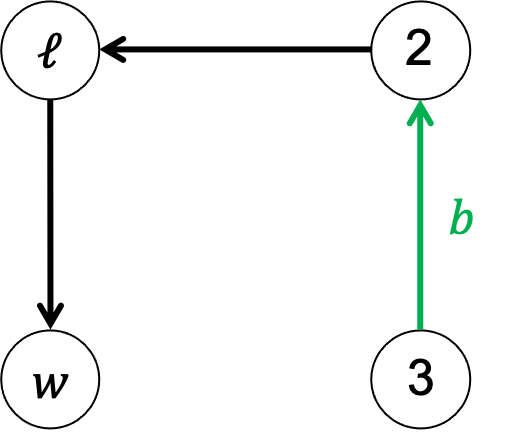}
\end{center}

\begin{claim}
Restricted to agents 2,3,4, the most envious agent of $X_2 \cup \{\b\}$ is agent 3.
\end{claim}

\begin{proof}
	Note that in the transition from $\mathbf{X'}$ to $\allocs$, agents 2 and 3 stay the same, and agent 4 is either better off or stays the same as well (depending on whether $\ell = 4$ or $w = 4$).  Thus, since 3 is most envious of $X_2 \cup \{\b\}$ in $\mathbf{X'}$, we get that 3 is also most envious of that set in $\allocs$ (restricted to agents 2,3,4).  
\end{proof}

Consider the cycle $1 \envies 4 \champ[\h] 3 \champ[\b] 2 \champ[\g] 1$ in $M_\allocs$: (note that $2\champ[\g]1$ releases $\b\in\Bgone$)
\begin{center}
	\includegraphics[scale=\figurescale]{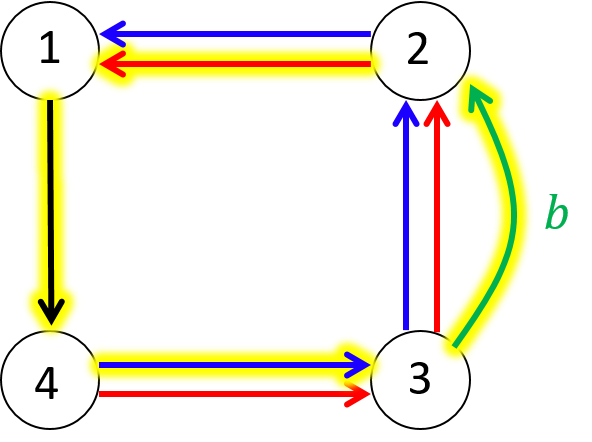}
\end{center}

This cycle might not really exist, since it could be the case that 1 is the real $\b$-champion of agent 2.  Assume for now that this is not the case, {\sl i.e.}, $3\champ[\b]2$ does exist in $M_{\allocs}$.  Then this is a PI cycle, and by Lemma \ref{lem:pareto-improvable}, there is an allocation $\mathbf{Y}$ that Pareto-dominates $\allocs$.  Note that in $\mathbf{Y}$ agent 1 receives $X_4$.  We therefore claim that $\mathbf{Y}$ is EFX even if 1 had been the most envious agent of $X_2 \cup \{\b\}$, and this finishes the proof.  Since 3 is the most envious agent in $\allocs$ restricted to agents 2,3,4, the only issue that could prevent $\mathbf{Y}$ from being EFX is if 1 strongly envies 3 in $\mathbf{Y}$.  But this is not the case since
$$ Y_1 = X_4 \geq_1 X'_1 >_1 (X_2 \cup \{\b\}) \setminus \discard[\b]{3}{2} = Y_3 $$ where the second inequality holds since $1 \nchamp[\b] 2$ in $\mathbf{X'}$.

\vspace{1ex}
\noindent\textbf{Case: $\mathbf{X'}$ is Not EFX.}   
\vspace{1ex}
  
Recall that in this case $\ell$ strongly envies $w$.  Define the set $Z \subseteq X_4$ as the subset of $X_4$ obtained by iterative removal of the item of least marginal value to agent $w$, until agent $\ell$ does not strongly envy $w$ anymore.
We obtain the allocation $\mathbf{X''}$ defined as follows:

\begin{center}
\renewcommand{\arraystretch}{1.2}
	$\allocs''$=\bundle{w}{Z}\quad
			\decomp{\ell}{\Tgone}{\g}\quad
			\bundle{2}{X_2}\quad
			\bundle{3}{X_3}	
\end{center}


Clearly, $\ell$ does not strongly envy $w$ in $\mathbf{X''}$.  
Note that by definition of $Z_\ell$ and $Z$, it must be the case that $\left|Z_\ell\right| \leq \left|Z\right|$ and consequently $\left|Z_w\right| \leq \left|Z_\ell\right| \leq \left|Z\right|$.  Since both $Z$ and $Z_w$ are obtained from $X_4$ by iterative removal of the item with least marginal contribution, we have $Z_w \subseteq Z$, and
we conclude that $w$ does not envy any other agent in $\allocs''$ by monotonicity of $v_w$ and the definition of $Z_w$.
Thus $\allocs''$ is an EFX allocation. 

Note however that $\ell$ might still envy $w$.  If this is indeed the case, then 
from here we can proceed by essentially repeating the previous case with $\mathbf{X''}$ instead of $\mathbf{X'}$. 
We thus assume that $\ell \nenvies w$. Let $\c \in X_4 \setminus Z$ be the last item to be removed from $X_4$ to obtain $Z$.  By definition of $Z$, $\ell$ strongly envies $Z\cup \{\c\}$.  Recall however that $\c$ is the least valued item in $Z\cup \{\c\}$ according to $w$.  Thus, for any item $x \in Z \cup \{\c\}$ we have
$$Z = (Z\cup \{\c\})\setminus \{\c\} >_w (Z\cup \{\c\})\setminus \{x\}$$
establishing that $w$ does not strongly envy $Z \cup \{\c\}$ in $\mathbf{X''}$.  We conclude that $\ell$ is the most envious agent of $X''_w \cup \{\c\}$ (recall that $2,3$ do not envy $X_4$ which is a superset), i.e., $\ell \champ[\c] w$.

From here we can, as above, essentially repeat the proof of the previous section with $\mathbf{X''}$ instead of $\mathbf{X'}$ and the champion edge $\ell \champ[\c] w$ instead of the envy edge $\ell \envies w$.
\section{Simplification and Extension of Known Results}\label{sec:extensions}

In this section we simplify the proofs of full EFX existence for 3 additive agents \cite{chaudhury2020efx} and for $n$ agents with one of two fixed additive valuations \cite{Mahara2020}. 
Moreover, our proofs extend beyond additive to all \valclass\ valuations.
Our proofs demonstrate the versatility of our techniques.
\shortversion
The case of two fixed additive valuations is deferred to Appendix~\ref{sec:2-simple},
and the rest of the section is devoted to the result by \cite{chaudhury2020efx}; namely EFX existence for three agents.
\shortversionend

\fullversion
\subsection{EFX for 3 \valclass\ valuations}
\label{sec:3-simple}
\fullversionend
 
By Lemma \ref{lem:dominate-implies-progress} it is sufficient to prove the following theorem.
\begin{theorem}
	Let $\allocs$ be an EFX allocation for 3 agents with nice cancelable valuations, with at least one unallocated item.  Then, there exists an EFX allocation $\mathbf{Y}$ that dominates $\allocs$.
\end{theorem}

By assumption there is an item $g$ that is unallocated in $\allocs$.
The original proof in \cite{chaudhury2020efx} distinguishes between two main cases according to whether $\allocs$ is envy-free or not.
In the case where $\allocs$ is not envy-free the original proof extends almost immediately to \valclass\ valuations (with one exception, see below).  The property of an additive valuation $v$ that is applied there is 
that for any bundles $S,T,R$ such that $R$ is disjoint from both $S$ and $T$,
$ v(S \cup R) > v(T \cup R) \Leftrightarrow v(S) > v(T)$, 
and this property also holds for \valclass\ valuations as pointed out in Section~\ref{sec:preliminaries}.

The one exception is in Section 4.2 in \cite{chaudhury2020efx} (right after Observation 16), in which their proof requires a subtle adjustment, as follows.
In their sub-case ``$a=2$'' they define the set $Z_i$, for $i =1,3$, to be a smallest subset of $X_3$ such that $Z_i >_i \max_i(X_1 \setminus G_{21} \cup g, X_2)$.
In our case, the set $Z_i$ needs to be defined as the subset of $X_3$ obtained by iteratively removing the item of least marginal value to agent $i$, as long as the leftover bundle is still better than $\max_i(X_1 \setminus G_{21} \cup g, X_2)$ from agent $i$'s point of view.  Lemma \ref{lem:remove_least_marginal_item_1_by_1} then guarantees that the new $Z_i$ is indeed a smallest-size subset of $X_3$ as in the original proof.  The rest of the argument then follows as in the original proof.

We now turn to the envy-free case.  In the original proof the following property of an additive valuation $v$ is used to prove the sub-case where $M_\allocs$ contains a $g$-cycle of size 2:  if $v(S) > v(T)$ and $v(S) > v(R)$, where $R$ and $T$ are disjoint, then $2\cdot v(S) > v(T \cup R)$. This property does \emph{not} hold in general for \valclass\ valuations ({\sl e.g.}, $v$ is multiplicative and $v(S)=5, v(R)=v(T)=4$), thus their proof does not extend to \valclass\ valuations.

In what follows we present a proof, based on our new techniques, that applies to all \valclass\ valuations, and is also substantially simpler than the original proof.
\begin{lemma}\label{lem:3-players-no-envy}
	Let $\allocs$ be a partial envy-free allocation on 3 agents. Then there exists an EFX allocation $\mathbf{Y}$ that Pareto dominates it. 
\end{lemma}

\begin{proof}
	By Lemma \ref{lem:pareto-improvable}, it suffices to find a Pareto-improvable (PI) cycle in the champion graph $M_\allocs$.
	Since every agent has an incoming $g$-champion edge (Observation \ref{obs:exists-champion}), $M_\allocs$ has a cycle of $g$ edges.
	Let $C$ be such a cycle of minimal length. If $C$ is a self loop, then we are done by Corollary \ref{cor:trivial_cycles}. Hence, it remains to consider the cases where $C$ has length two or three.  Since $\allocs$ is envy-free, $C$ is a good $g$-cycle, and as such it induces a $g$-decomposition of $X_j$ into top and bottom half-bundles $X_j = T_j \cup B_j$ for every agent $j$ along the cycle.
	
	{\em Case 1:} $|C|=2$. Assume w.l.o.g. that $C=1 \champ \linebreak[1] 2\champ 1$.
	By Theorem \ref{thm:good-edge-or-external}, $M_\allocs$ contains a good or external $B_1$-edge and a good or external $B_2$-edge.  By Corollary \ref{cor:no_single_external_source} they cannot both be external (since their source would be the same, agent 3).  Thus, w.l.o.g. $M_\allocs$ contains a good $B_1$-edge, which can only be $1 \champ[B_1 \mid\circ] 2$, and we obtain the PI cycle $1 \champ[B_1 \mid\circ] 2 \champ[g] 1$.
	
	{\em Case 2:} $|C|=3$. W.l.o.g. let $C=1\champ\linebreak[1]2\champ\linebreak[1]3\champ1$. 
	As $C$ contains all 3 agents, there are no external edges going into $C$. Thus by Theorem \ref{thm:good-edge-or-external}, $M_\allocs$ contains a good $B_i$-edge, for each $i\in[3]$. 
	
	If for some $i\in[3]$ the good $B_i$-edge is $i \champ[B_i \mid\circ]\pred(i)$, then we get the PI cycle $i \champ[B_i \mid\circ]\pred(i) \champ i$, and we are done. Thus we can assume that all good edges are parallel to the edges of $C$, i.e., from $j$ to $\succ(j)$. W.l.o.g., assume there is a good edge from agent 1 to agent 2. This good edge cannot be the good $B_2$-edge, because $1\nchamp[B_2 \mid\circ]2$ by Observation \ref{obs:no_B_j_champ_of_j}. Thus, this good edge is $1\champ[B_i \mid \circ]2$, for some $i\in \{1,3\}$, and the good $B_2$-edge must be either $2\champ[B_2 \mid\circ]3$ or $3\champ[B_2 \mid\circ]1$. 
	The former case admits the PI cycle $ 1\champ[B_i \mid\circ] 2 \champ[B_2 \mid\circ]3 \champ 1$; the latter admits the PI cycle $1\champ[B_i \mid\circ]2 \champ 3 \champ[B_2 \mid\circ]1$.
\end{proof}
\fullversion
	\fullversion
\subsection {EFX for agents with one of two valuations}
\label{sec:2-simple}
\fullversionend

\shortversion
\section{Proofs for Section~\ref{sec:extensions}: EFX for Agents with One of Two Valuations}
\label{sec:2-simple}
\shortversionend

Consider a setting with $n$ agents, where any agent has one of two valuations $v_a$, $v_b$. 
Let $a_0, \ldots, a_t$ denote the agents with valuation $v_a$, and $b_0, \ldots, b_{\ell}$ denote the agents with valuation $v_b$,
ordered such that
$$
X_{a_0} \leq_a X_{a_1} \leq_a \ldots \leq_a X_{a_t} \mbox{\qquad and \qquad}  X_{b_0} \leq_b X_{b_1} \leq_b \ldots \leq_b X_{b_\ell}.
$$

The following theorem shows that if $v_a$ and $v_b$ are \valclass\ valuations, then given any partial EFX allocation, there exists an EFX allocation that Pareto dominates it. This implies (by Lemmata~\ref{lem:dominate-implies-progress}, \ref{lem:non-deg}) that every instance in this setting admits a full EFX allocation.

\begin{theorem}
	\label{thm:two-valuations}
	In every setting with two \valclass\ valuations, given any partial EFX allocation, there exists an EFX allocation that Pareto dominates it.
\end{theorem}

Before presenting the proof, we present a useful observation.
We say that envy (resp., most envious) {\em propagates backward} within the valuation class $a$ if whenever some agent $a_i$ envies a set $S$ (resp., is most envious of a set $S$), then for every $j<i$, agent $a_j$ envies $S$ (resp., is most envious of $S$) as well. We say that championship propagates backward within the valuation class $a$ in an analogous way.
We define backward propagation within the valuation class $b$ analogously.
One can easily verify that envy propagates backward. The following observation shows that so does championship.

\begin{observation}
	Championship propagates backward within the same valuation class.
\end{observation}
\begin{proof}
	We prove the claim for valuation class $a$. The proof for valuation class $b$ is analogous. 
	We show that the relation most envious propagates backward; by extension, championship propagates backward as well.
	 Suppose that for some $i\in \{0,\ldots,\ell \}$, agent $a_i$ is most envious of a set $S$, and let $D$ be the discard set of $S$ with respect to $a_i$. This means that $a_i$ envies $S \setminus D$. Since envy propagates backward, so does agent $a_j$. By definition of a discard set, no agent strongly envies $S \setminus D$. Therefore, $a_j$ is most envious of $S$.
\end{proof}

We are now ready to prove Theorem~\ref{thm:two-valuations}.

\begin{proof}
Fix a partial EFX allocation, and let $g$ be an unallocated item.
By Corollary~\ref{cor:trivial_cycles} we may assume that no agent is a $g$-self champion.
We first claim that
\begin{equation}
	\label{eq:a-claim}
a_0 \champ b_0 \mbox{  and  } b_0 \champ a_0.
\end{equation}

Indeed, if $b_j \champ b_0$ for some $j$, then $b_0$ is a $g$-self champion by backward propagation. 
Thus, since $b_0$ must have a $g$-champion (Observation \ref{obs:exists-champion}), then $a_j \champ b_0$ for some $j$.
Since championship propagates backward, $a_0 \champ b_0$. 
By symmetry,  $b_0 \champ a_0$.

We may also assume that no agent envies $a_0$ (and similarly, $b_0$).
Clearly, no $a_j$ envies $a_0$.	It remains to show that no $b_j$ envies $a_0$. 
Indeed, if some agent $b_j$ envies $a_0$, then $b_0$ envies $a_0$, and together with the fact that $a_0 \champ b_0$, we have a Pareto-improvable cycle, so we are done by Lemma~\ref{lem:pareto-improvable}. Similarly, if any agent envies $b_0$, we have a Pareto-improvable cycle.

By Equation~(\ref{eq:a-claim}) and the assumption that no agent envies $a_0$ or $b_0$, $a_0 \champ b_0 \champ a_0$ is a good $g$-cycle, thus by Observation~\ref{obs:g_notin_bottom} the bundles of $a_0$ and $b_0$ decompose into top and bottom half-bundles. Let $T_{a_0}$ and $B_{a_0}$ (resp., $T_{b_0}$ and $B_{b_0}$) be the top and bottom half-bundles of $a_0$ (resp., $b_0$), respectively.

	We next argue that $a_0 \champ[B_{a_0}\mid \circ] b_0$. Since $a_0 \champ b_0 \champ a_0$ is a good cycle, by Theorem~\ref{thm:good-edge-or-external} there exists a good or external $B_{a_0}$-edge that goes into $b_0$. If this is a good $B_{a_0}$-edge then it can only be $a_0 \champ[B_{a_0}\mid \circ] b_0$. If this is an external $B_{a_0}$-edge, it can't be $b_j \champ[B_{a_0}\mid \circ] b_0$ since $b_0 \nchamp[B_{a_0}\mid \circ] b_0$ (by Observation~\ref{obs:pred(i)_nchamp_B(i)}) and championship propagates backward. Hence, the external $B_{a_0}$-edge must be $a_j \champ[B_{a_0}\mid \circ] b_0$ for some $j\in \{1,\ldots,t\}$. Again, since championship propagates backward, $a_0 \champ[B_{a_0}\mid \circ] b_0$.
	
	It follows that we have a Pareto-improvable cycle consisting of $a_0 \champ[B_{a_0}\mid \circ] b_0$ and $b_0 \champ a_0$.
	We may now apply Lemma~\ref{lem:pareto-improvable} to conclude the proof.
	
\end{proof}

\fullversionend

\section*{Acknowledgment} 
We are very grateful to Karin Egri for her great help in producing the figures used in this paper. 
We are deeply in debt to Bhaskar Ray Chaudhury, Jugal Garg, Kurt Mehlhorn, Ruta Mehta and Misra Pranabendu for their wonderful papers and for sharing unpublished ongoing results.

\bibliographystyle{alpha}
\bibliography{bibliography}

\newcommand{\etalchar}[1]{$^{#1}$}
\begin{thebibliography}{ABFR{\etalchar{+}}21}

\bibitem[ABFR{\etalchar{+}}21]{Amanatidis20}
Georgios Amanatidis, Georgios Birmpas, Aris Filos-Ratsikas, Alexandros
  Hollender, and Alexandros~A. Voudouris.
\newblock Maximum {Nash} welfare and other stories about {EFX}.
\newblock {\em Theoretical Computer Science}, 2021.

\bibitem[BKK13]{Brams13}
Steven~J. Brams, D.~Marc Kilgour, and Christian Klamler.
\newblock Two-person fair division of indivisible items: An efficient,
  envy-free algorithm.
\newblock MPRA Paper 47400, University Library of Munich, Germany, June 2013.

\bibitem[BT00]{BramsTaylor2000}
Steven Brams and Alan Taylor.
\newblock {\em The win-win solution - guaranteeing fair shares to everybody.}
\newblock WW Norton \& Company, 2000.

\bibitem[Bud11]{Budish2011}
Eric Budish.
\newblock The combinatorial assignment problem: {Approximate Competitive
  Equilibrium from Equal Incomes}.
\newblock {\em Journal of Political Economy}, 119(6):1061--1103, 2011.

\bibitem[CG15]{ColeGkatzelis15}
Richard Cole and Vasilis Gkatzelis.
\newblock Approximating the {Nash} social welfare with indivisible items.
\newblock In {\em Proceedings of the Forty-Seventh Annual ACM Symposium on
  Theory of Computing}, STOC '15, page 371–380, New York, NY, USA, 2015.
  Association for Computing Machinery.

\bibitem[CGG13]{Cole2013}
Richard Cole, Vasilis Gkatzelis, and Gagan Goel.
\newblock Mechanism design for fair division: Allocating divisible items
  without payments.
\newblock In {\em Proceedings of the Fourteenth ACM Conference on Electronic
  Commerce}, EC '13, page 251–268, New York, NY, USA, 2013. Association for
  Computing Machinery.

\bibitem[CGH19]{caragiannis2019envy}
Ioannis Caragiannis, Nick Gravin, and Xin Huang.
\newblock Envy-freeness up to any item with high {Nash} welfare: The virtue of
  donating items.
\newblock In {\em Proceedings of the 2019 {ACM} Conference on Economics and
  Computation}, pages 527--545, 2019.

\bibitem[CGM20]{chaudhury2020efx}
Bhaskar~Ray Chaudhury, Jugal Garg, and Kurt Mehlhorn.
\newblock {EFX} exists for three agents.
\newblock In {\em Proceedings of the 21st ACM Conference on Economics and
  Computation}, pages 1--19, 2020.

\bibitem[CGM{\etalchar{+}}21]{chaudhury2020example}
Bhaskar~Ray Chaudhury, Jugal Garg, Kurt Mehlhorn, Ruta Mehta, and Misra
  Pranabendu.
\newblock Personal communication.
\newblock 2021.

\bibitem[Chr42]{Chroust42}
Anton-Hermann Chroust.
\newblock Aristotle's conception of equity (epieikeia).
\newblock {\em Notre Dame Law Review}, 18:119--128, 1942.

\bibitem[CKM{\etalchar{+}}16]{caragiannis2016unreasonable}
Ioannis Caragiannis, David Kurokawa, Herv\'{e} Moulin, Ariel~D. Procaccia,
  Nisarg Shah, and Junxing Wang.
\newblock The unreasonable fairness of maximum {Nash} welfare.
\newblock In {\em Proceedings of the 2016 ACM Conference on Economics and
  Computation}, EC '16, page 305–322, New York, NY, USA, 2016. Association
  for Computing Machinery.

\bibitem[CKM{\etalchar{+}}19]{caragiannis2019unreasonable}
Ioannis Caragiannis, David Kurokawa, Herv{\'e} Moulin, Ariel~D Procaccia,
  Nisarg Shah, and Junxing Wang.
\newblock The unreasonable fairness of maximum {Nash} welfare.
\newblock {\em ACM Transactions on Economics and Computation (TEAC)},
  7(3):1--32, 2019.

\bibitem[CKMS20]{chaudhury2020little}
Bhaskar~Ray Chaudhury, Telikepalli Kavitha, Kurt Mehlhorn, and Alkmini
  Sgouritsa.
\newblock A little charity guarantees almost envy-freeness.
\newblock In {\em Proceedings of the Fourteenth Annual ACM-SIAM Symposium on
  Discrete Algorithms}, pages 2658--2672. SIAM, 2020.

\bibitem[Fol67]{Foley67}
Duncan~K. Foley.
\newblock Resource allocation and the public sector.
\newblock {\em Yale economic essays}, 7:45--98, 1967.

\bibitem[GMT14]{Gourvs2014NearFI}
Laurent Gourv{\'e}s, J{\'e}r{\^o}me Monnot, and Lydia Tlilane.
\newblock Near fairness in matroids.
\newblock In {\em ECAI}, pages 393--398, 2014.

\bibitem[GS58]{gamow1958puzzle}
George Gamow and Marvin Stern.
\newblock {\em Puzzle-math}.
\newblock Viking Press, 1958.

\bibitem[Jon02]{jones02}
Michael~A. Jones.
\newblock Equitable, envy-free, and efficient cake cutting for two people and
  its application to divisible goods.
\newblock {\em Mathematics Magazine}, 75(4):275--283, 2002.

\bibitem[KSV20]{Kyropoulou_2020}
Maria Kyropoulou, Warut Suksompong, and Alexandros~A. Voudouris.
\newblock Almost envy-freeness in group resource allocation.
\newblock {\em Theoretical Computer Science}, 841:110–123, Nov 2020.

\bibitem[LMMS04]{Lipton04}
Richard~J. Lipton, Evangelos Markakis, Elchanan Mossel, and Amin Saberi.
\newblock On approximately fair allocations of indivisible goods.
\newblock In {\em Proceedings of the 5th ACM Conference on Electronic
  Commerce}, EC '04, page 125–131, New York, NY, USA, 2004. Association for
  Computing Machinery.

\bibitem[Mah20]{Mahara2020}
Ryoga Mahara.
\newblock Existence of {EFX} for two additive valuations, 2020.
\newblock \url{http://arxiv.org/abs/2008.08798}.

\bibitem[MT10]{Mossel2010}
Elchanan Mossel and Omer Tamuz.
\newblock Truthful fair division.
\newblock In {\em Proceedings of the Third International Conference on
  Algorithmic Game Theory}, SAGT'10, page 288–299, Berlin, Heidelberg, 2010.
  Springer-Verlag.

\bibitem[PR20]{plaut2020almost}
Benjamin Plaut and Tim Roughgarden.
\newblock Almost envy-freeness with general valuations.
\newblock {\em SIAM Journal on Discrete Mathematics}, 34(2):1039--1068, 2020.

\bibitem[Pro20]{ProcacciaCACM}
Ariel~D. Procaccia.
\newblock Technical perspective: An answer to fair division's most enigmatic
  question.
\newblock {\em Commun. ACM}, 63(4):118, March 2020.

\bibitem[Raw99]{Rawls99}
John Rawls.
\newblock {\em A Theory of Justice}.
\newblock Harvard University Press, 1999.

\bibitem[Ste49]{Steinhaus49}
Hugo Steinhaus.
\newblock Sur la division pragmatique.
\newblock {\em Econometrica}, 17:315--319, 1949.

\bibitem[Var74]{Varian74}
Hal Varian.
\newblock Equity, envy, and efficiency.
\newblock {\em Journal of Economic Theory}, 9(1):63--91, 1974.

\end{thebibliography}

\clearpage

\appendix
%
\section{Proofs and Claims from Section~\ref{sec:preliminaries}}\label{apx:Preliminaries}
\begin{lemma}
	\label{lem:non-deg}
	To prove the existence of an EFX allocation for a given valuation profile $\mathbf{v} = (v_1,\ldots,v_n)$ of \valclass\ valuations, it is without loss of generality to assume that all of the valuations are non-degenerate.
\end{lemma}

\begin{proof}
For each $i\in [n]$, $v_i$ is a \valclass\ valuation. Hence, for every $i\in [n]$ there exists a non-degenerate \canc\ valuation $v'_i$ that respects $v_i$.
Let $\mathbf{v'}=(v'_1,\ldots,v'_n)$ be the corresponding valuation profile.

We show that any allocation which is EFX for the profile $\mathbf{v'}$ is also EFX for the profile $\mathbf{v}$. 
Let $\mathbf{X}=(X_1,\dots,X_n)$ be an EFX allocation for the profile $\mathbf{v'}$ and assume towards contradiction that $\mathbf{X}$ is not an EFX allocation for the profile $\mathbf{v}$. Hence, under the profile $\mathbf{v}$ there exists an agent $i$ that strongly envies another agent $j$, \emph{i.e.}, $v_i(X_i)<v_i(X_j\setminus \{h\})$ for some $h\in X_j$. Since $v'_i$ respects $v_i$, it follows that $v'_i(X_i)<v'_i(X_j\setminus \{h\})$, in contradiction to $\mathbf{X}$ being EFX over the profile $\mathbf{v}'$. 
\end{proof}
%
%
The family of \valclass\ valuations contains some well-known classes of valuations. Additive valuations are clearly \canc\ and are shown to be \nice\ in \cite{chaudhury2020efx}. The following lemma shows that this class contains many other classes of valuations, including unit-demand, budget-additive and multiplicative.
%
\begin{lemma}\label{lem:BA_and_UD_are_RC}
	Unit-demand, budget-additive and multiplicative valuations are \valclass. 
\end{lemma}
%
%
%
%
	\begin{proof}

	\vspace{1ex}
	\noindent\textbf{Budget-Additive:} Let $v$ be a budget-additive valuation, \emph{i.e.}, for every $S\subseteq M$,
	\[
	v(S)=\min\left\{\sum\limits_{g\in S} v(g),\; B\right\},
	\]
	for some $B>0$. We start by proving $v$ is \canc. Consider $S,T\subseteq M$ and $g\in M\setminus(S\cup T)$ such that
	\[
	v(S\cup\{g\})>v(T\cup\{g\}).
	\]
	First, $v(S\cup\{g\}) \leq B$, by definition of $v$. Therefore, $v(T\cup\{g\})<B$, so $v$ is additive over $T\cup\{g\}$. It follows that $v(T) = v(T\cup\{g\})-v(\{g\})$.
	Second, since budget-additive valuations are sub-additive, $v(S)\geq v(S\cup\{g\})-v(\{g\})$. Combining these two observations we get
	\[
	v(S)\geq v(S\cup\{g\})-v(\{g\}) > v(T\cup\{g\})-v(\{g\}) = v(T).
	\]
	This proves that $v$ is \canc.
	
	We next prove that $v$ is \nice. Define the valuation $v':2^M\to \mathbb{R}_{\geq 0}$ as the underlying additive valuation of $v$, \emph{i.e.}, for every $S\subseteq M$,
	\[
	v'(S)=\sum\limits_{g\in S} v(g).
	\]
	We now show that $v'$ respects $v$. 
	
	Suppose $v(S)>v(T)$ for some $S,T\subseteq M$. Since $v(S)\leq B$, it follows that $v(T)<B$, and thus $v$ is additive over $T$. By definition of $v'$, this implies that $v(T)=v'(T)$. Furthermore, notice that $v'(S)\geq v(S)$. Thus,
	\[
	v'(S)\geq v(S)>v(T)=v'(T).
	\]
	This proves that $v'$ respects $v$.
	
	Finally, since $v'$ is an additive valuation it is \nice\ and \canc\ (as shown in \cite{chaudhury2020efx}). Therefore, there exists a non-degenerate \canc\ valuation $v''$ that respects $v'$. Because $v'$ respects $v$, it follows by transitivity that $v''$ respects $v$ as well. Since $v''$ is non-degenerate and \canc, this proves that $v$ is \nice.

	\vspace{1ex}
	\noindent\textbf{Multiplicative}: Let $v$ be a multiplicative valuation, \emph{i.e.}, for every $S\subseteq M$,
	\[
	v(S)=\prod_{g\in S} v(g).
	\]
	Multiplicative valuations are trivially \canc.
	Since additive valuations are \nice, and since taking the $\log$ of a multiplicative valuation gives us an additive function, similar arguments as in \cite{chaudhury2020efx} imply that multiplicative is also \nice.

	\vspace{1ex}	
	\noindent\textbf{Unit-Demand:} Let $v$ be a unit-demand valuation, \emph{i.e.}, for every $S\subseteq M$,
	\[
	v(S)=\max\limits_{g\in S} v(g).
	\]
	We first show $v$ is \canc. 
	Consider $S,T\subseteq M$ and $g\in M\setminus(S\cup T)$ such that
	\[
	v(S\cup\{g\})>v(T\cup\{g\}).
	\]
	Clearly, $g$ is not the maximal element in $S\cup\{g\}$, otherwise we would have $v(S\cup\{g\})= v(\{g\}) \leq v(T\cup\{g\})$. Therefore, $v(S\cup\{g\})=v(S)$. We get
	\[
	v(S)=v(S\cup\{g\})>v(T\cup\{g\})\geq v(T).
	\]
	This proves that $v$ is \canc.
	
	We next prove that $v$ is \nice. Define $$\delta = \min\limits_{S,T\subseteq M,\; v(S)\neq v(T)}|v(S)-v(T)|.$$
	That is, $\delta$ is the minimal difference between the value of any two non-equal valued sets of items. Let $g_0,\ldots,g_{m-1}$ be the items in $M$ ordered in non-decreasing value, ties broken arbitrarily. Let $\varepsilon= 2^{-(m+1)}\delta$. Define the valuation $v':2^M\to \mathbb{R}_{\geq 0}$ as follows:
	\[
	v'(S)=v(S)+\varepsilon\sum\limits_{i:g_i\in S} 2^i.
	\]
	To complete the proof we need to show that $v'$ is a non-degenerate \canc\ valuation and that $v'$ respects $v$. We begin with the latter. Suppose $v(S)>v(T)$ for some $S,T\subseteq M$. This implies that $v(S)>v(T)+\delta/2$, by definition of $\delta$. Moreover, since $\sum_{i=0}^{m-1}2^i<2^m$, we have
	$$
	v'(T)=v(T)+\varepsilon\sum\limits_{i:g_i\in T} 2^i
	< v(T)+\varepsilon\cdot2^m 
	=v(T)+\frac{\delta}{2}.
	$$
	We get
	$$
	v'(S)>v(S)>v(T)+\frac{\delta}{2} > v'(T),
	$$
	as required.
	
	We next prove $v'$ is non-degenerate. For all $S,T \subseteq M$ such that $v(S)\neq v(T)$, we have shown above that $v'(S)\neq v'(T)$. So it remains to show that $v'(S)\neq v'(T)$ whenever $v(S) = v(T)$ and $S\neq T$. Since $v(S) = v(T)$, to prove that $v'(S)\neq v'(T)$ it suffices to show that
	$$
	\varepsilon\sum\limits_{i:g_i\in S} 2^i \neq \varepsilon\sum\limits_{i:g_i\in T} 2^i,
	$$
	which clearly holds for every $S\neq T$.
	
	Finally, we prove $v'$ is \canc. Consider $S,T\subseteq M$ and $g_j\in M\setminus(S\cup T)$ such that
	\[
	v'(S\cup\{g_j\})>v'(T\cup\{g_j\}).
	\]
	It is impossible that $v(S\cup\{g_j\})<v(T\cup\{g_j\})$, since $v'$ respects $v$. If $v(S\cup\{g_j\})>v(T\cup\{g_j\})$, then $v(S)>v(T)$ since $v$ is \canc. Since $v'$ respects $v$, this implies $v'(S)>v'(T)$, so we are done. We are left with the case where $v(S\cup\{g_j\})=v(T\cup\{g_j\})$. Since $v'(S\cup\{g_j\})>v'(T\cup\{g_j\})$, this implies
	$$
	\varepsilon\cdot\sum\limits_{i:g_i\in S\cup\{g_j\}} 2^i > \varepsilon\cdot\sum\limits_{i:g_i\in T\cup\{g_j\}} 2^i.
	$$
	Eliminating $\varepsilon \cdot 2^j$ from both sides,
	\begin{equation}\label{eq:epsS>epsT}
	\varepsilon\sum\limits_{i:g_i\in S} 2^i > \varepsilon\sum\limits_{i:g_i\in T} 2^i.
	\end{equation}
	If $v(S) \geq v(T)$ then Equation~\eqref{eq:epsS>epsT} implies $v'(S) > v'(T)$ and we are done. We complete the proof by showing that the case $v(S) < v(T)$ is impossible. Assume towards contradiction that $v(S) < v(T)$. Let $g_s$ and $g_t$ be the highest valued items in $S$ and $T$, respectively (breaking ties according to the ordering we defined on the items). Since $v(S) < v(T)$ and $v$ is unit-demand, $v(g_s)<v(g_t)$. Due to our ordering, it follows that $s<t$. Therefore,
	$$
	\sum\limits_{i:g_i\in S} 2^i
	\leq \sum_{i=0}^{s}2^i
	< 2^{s+1}
	\leq 2^t
	\leq \sum\limits_{i:g_i\in T} 2^i,
	$$
	in contradiction to Equation~\eqref{eq:epsS>epsT}. This shows that $v'$ is \canc.
\end{proof}

On the other hand, not all \canc\ valuations are \nice, as shown in Proposition~\ref{prop:sm-not-tb}.

\begin{proposition}
	\label{prop:sm-not-tb}
	Not all \canc\ valuations are \nice. Even when restricted to submodular \canc\ valuations it need not be \nice.
\end{proposition}
%
%
\begin{proof}
	We define a valuation $v$ over the set of items $M=\{a,b,c,d,e,f\}$ and show that it is \canc\ and submodular but not \nice. The following table defines the value of $v$ over each singleton.
	
	\begin{center}
		\begin{tabular}{c|c|c|c|c|c|c|}
			\cline{2-7}
			& \textit{a}   & \textit{b}   & \textit{c}   & \textit{d}   & \textit{e}   & \textit{f}   \\ \hline
			\multicolumn{1}{|c|}{\textit{v}} & 101          & 102          & 102          & 103          & 103          & 104          \\ \hline
		\end{tabular}
	\end{center}
	
	Notice that $v(a)<v(b)=v(c)<v(d)=v(e)<v(f)$, i.e., the values are non-decreasing from left to right.
	
	The next table defines the value of $v$ over each pair of items. For convenience, we depict this table as a matrix, where the coordinate $(x,y)$ contains the value $v(\{x,y\})$ (the matrix is symmetric).
	
	\begin{center}
		\begin{tabular}{c|c|c|c|c|c|c|}
			\cline{2-7}
			& \textit{a}        & \textit{b}        & \textit{c}        & \textit{d}        & \textit{e}        & \textit{f} \\ \hline
			\multicolumn{1}{|c|}{\textit{a}} & -                 & 152               & 152               & 153               & 153               & 154        \\ \hline
			\multicolumn{1}{|c|}{\textit{b}} & \color{gray}{152} & -                 & 152               & 155               & 155               & 156        \\ \hline
			\multicolumn{1}{|c|}{\textit{c}} & \color{gray}{152} & \color{gray}{152} & -                 & 155               & 155               & 156        \\ \hline
			\multicolumn{1}{|c|}{\textit{d}} & \color{gray}{153} & \color{gray}{155} & \color{gray}{155} & -                 & 155               & 156        \\ \hline
			\multicolumn{1}{|c|}{\textit{e}} & \color{gray}{153} & \color{gray}{155} & \color{gray}{155} & \color{gray}{155} & -                 & 156        \\ \hline
			\multicolumn{1}{|c|}{\textit{f}} & \color{gray}{156} & \color{gray}{156} & \color{gray}{156} & \color{gray}{156} & \color{gray}{156} & -          \\ \hline
		\end{tabular}
	\end{center}
	
	Finally, for any set $S\subseteq M$ containing three or more items we define $v(S)=200$, and we set $v(\emptyset)=0$. This completes the definition of $v$ over all subsets of $M$.

	\vspace{1ex}	
	\noindent\textbf{Cancelability:} Consider a pair of sets $S,T \subseteq M$ and an item $g\in M\setminus (S\cup T)$ such that $v(S\cup\{g\})>v(T\cup\{g\})$. First consider the cases where $|S|=|T|=1$. Therefore, $S\cup\{g\}$ and $T\cup\{g\}$ are both pairs of items whose values are located in row $g$ of the matrix above. It is not hard to verify that within each row of the matrix the values are non-decreasing left to right with equality only between the columns of equal valued items (the columns $b$ and $c$ and the columns $d$ and $e$ contain equal values). Therefore, the fact that $v(S\cup\{g\})>v(T\cup\{g\})$ implies that $v(S\cup\{g\})$ is located to the right of $v(T\cup\{g\})$ in row $g$ of the matrix. Therefore, $v(S)\geq v(T)$. The case $v(S)=v(T)$ implies $v(S\cup\{g\})=v(T\cup\{g\})$ since the columns of equal valued singletons are identical in the matrix of pairs. Thus, $v(S)>v(T)$. This completes the case $|S|=|T|=1$.
	
	If $|T|> 2$, then $v(T\cup\{g\})=200$ and there exists no set $S$ such that $v(S\cup\{g\})>200$. Therefore, we may assume $|T|\leq 1$. The case $|T|=0$ is easy. In this case, $v(S\cup\{g\})>v(T\cup\{g\})$ implies that $S$ is non-empty. Since $T=\emptyset$ and $S\neq \emptyset$, we obtain $v(S)>v(T)$, so \cancbty\ is maintained.
	
	We are left with the case $|T|=1$. The case $|S|=0$ is not possible since no singleton has a higher value than a pair of items. The case $|S|=|T|=1$ has been handled above. So it is left to consider $|S|>1$ and $|T|=1$. In this case, notice that $v(S)>150$ while $v(T)<150$, therefore $v(S)>v(T)$, as desired. This proves that $v$ is \canc.

	\vspace{1ex}	
	\noindent\textbf{Submodularity:} It suffices to show that for every item $g$ and every pair of sets $S$ and $T$ such that $S\subseteq T \subseteq M\setminus \{g\}$ we have that $v(S\cup\{g\})-v(S) \geq v(T\cup\{g\})-v(T)$. This condition clearly holds if $S=T$, so assume $S\subsetneq T$. Therefore, $|S|<|T|$. Notice that the following holds for any set $Z\subseteq M$ and any item $x\in M\setminus Z$:
	\begin{align*}
	\text{if $|Z|=0$,}&& 100<&v(Z\cup\{x\})-v(Z)\leq 104;\\
	\text{if $|Z|=1$,}&& 50\leq& v(Z\cup\{x\})-v(Z)<60;\\
	\text{if $|Z|=2$,}&& 40<&v(Z\cup\{x\})-v(Z)<50;\\
	\text{if $|Z|\geq 3$,}&&  &v(Z\cup\{x\})-v(Z)=0.
	\end{align*}
	Therefore, the fact that $|S|<|T|$ directly implies that $v(S\cup\{g\})-v(S) \geq v(T\cup\{g\})-v(T)$. This proves that $v$ is submodular.

	\vspace{1ex}	
	\noindent\textbf{Not \nice:} Assume towards contradiction that there exists a non-degenerate \canc\ valuation $v'$ that respects $v$. Notice that the following inequalities hold:
	\begin{align}
	154&=v\left(\{a,f\}\right)<v\left(\{b,e\}\right)=155,		\label{eq:RC_af<be}\\
	155&=v\left(\{d,e\}\right)<v\left(\{c,f\}\right)=156,		\label{eq:RC_de<cf}\\	
	152&=v\left(\{b,c\}\right)<v\left(\{a,d\}\right)=153.		\label{eq:RC_bc<ad}	
	\end{align}
	Therefore, the analogous inequalities must hold for $v'$. Consider the comparison between $\{a,d,f\}$ and $\{b,d,e\}$ in $v'$. $v'$ is non-degenerate, so $v'(\{a,d,f\}) \neq v'(\{b,d,e\})$. If $v'(\{a,d,f\}) > v'(\{b,d,e\})$, then by \cancbty\ $v'(\{a,f\}) > v'(\{b,e\})$, which contradicts Equation~\eqref{eq:RC_af<be}, so
	\begin{equation}\label{eq:RC_adf<bde}
	v'(\{a,d,f\}) < v'(\{b,d,e\}).
	\end{equation}
	Similarly, due to Equation~\eqref{eq:RC_de<cf} we obtain 
	\begin{equation}\label{eq:RC_bde<bcf}
	v'(\{b,d,e\}) < v'(\{b,c,f\}),
	\end{equation}
	and from Equation~\eqref{eq:RC_bc<ad} we get
	\begin{equation}\label{eq:RC_bcf<adf}
	v'(\{b,c,f\}) < v'(\{a,d,f\}).
	\end{equation}
	However, combining Equations~\eqref{eq:RC_adf<bde} and \eqref{eq:RC_bde<bcf} we get $v'(\{a,d,f\}) < v'(\{b,c,f\})$, in contradiction to Equation~\eqref{eq:RC_bcf<adf}. This proves that there exists no non-degenerate \canc\ valuation $v'$ such that $v'$ respects $v$, and thus $v$ is not \nice.
%
\end{proof}

If an additive agent strongly envies some bundle $S$, then, iterative removal of the least valued item until strong envy is eliminated results in a smallest size subset of $S$ that the agent envies.
The next lemma shows that this property extends to cancelable valuations.

\begin{lemma}\label{lem:remove_least_marginal_item_1_by_1}
	Let $v$ be a \canc\ valuation.  Let $T$ be some bundle, and let $S$ a subset of $T$.  Let $Z$ be the subset obtained from $T$ by iteratively removing the item with least marginal contribution until the leftover bundle has size $\left|S\right|$.  Then $v(Z) \geq v(S)$.
\end{lemma}

\begin{proof}
	Define $T_0 =T$, and for $j\geq 0$ define $T_{j+1} = T_j \setminus \{c\}$, where $c \in T_j$ is the item with least marginal contribution to $T_j$.  It suffices to prove that for every $0 \leq j \leq \left|T\right|$ we have 
	$$T_{j} = \arg\max_{S\subseteq T : \left|S\right| = \left|T\right| - j} v(S).$$
	We prove by induction on $j$.  For $j=0$ the claim is immediate.  Assume that the claim is true for $j$ and we prove for $j+1$.
	Let $c \in T_j$ be the item with least marginal contribution to $T_j$, hence by definition we have $T_{j+1} = T_j \setminus \{c\}$.
	Let $S \subseteq T$ such that $\left|S\right| = \left|T_{j+1}\right|$.  We need to show that $v(S) \leq v(T_{j+1})$.  If $S = T_{j+1}$, then this is immediate. Therefore, assume $S \neq T_{j+1}$.
	Since $S$ and $T_{j+1}$ have the same size, this means that there is some item $b \in T_{j+1} \setminus S$, and thus $S \cup \{b\}$ and $T_j$ have the same size.  By the induction hypothesis we get
	$$v(S \cup \{b\}) \leq v(T_j) = v(T_j \setminus \{b\} \cup \{b\}),$$
	implying $v(S) \leq  v(T_j \setminus \{b\}) \leq v(T_j \setminus \{c\}) = v(T_{j+1})$ where the first inequality holds by \cancbty, and the second by definition of $c$.  The claim follows.
\end{proof}

\putmaybeappendix{Lemma_dominate_implies_progress_proof}
\putmaybeappendix{obs_exists-champion}
\putmaybeappendix{obs_g_notin_bottom}
\putmaybeappendix{obs_non-champion-doesnt-envy-top-half}
\putmaybeappendix{obs_T_k<T_j}

\shortversion
\section{Proofs from Section~\ref{sec:generalized_championship}}\label{apx:techniques}
\shortversionend
\putmaybeappendix{lem_pareto-improvable}
\putmaybeappendix{obs_good_cycle_g_decomposition}
\putmaybeappendix{obs_no_B_j_champ_of_j}
\putmaybeappendix{obs_pred(i)_nchamp_B(i)}
\putmaybeappendix{obs_bottom_bundle_ineq}
\putmaybeappendix{cor_no_single_external_source}

\section{Proofs for Section \ref{sec:4-agents}}\label{apx:4-agents}
\shortversion
\subsection{Missing Cases from Section~\ref{sec:ef-case}}\label{apx:4-agents-EF}	
\shortversionend
\putmaybeappendix{4p-NoEnvy-Cases-3+4}

\putmaybeappendix{4p-NoEnvy-Case-6}

\shortversion
\subsection{Missing Cases and Proofs from Section~\ref{sec:not-ef-case}}\label{sec:4-agents-Not-EF}
\shortversionend

\fullversion
\begin{proof}[Proof of Claim \ref{claim:B_g_subseteq_B_h}]
	Choose some arbitrary discard set $\discard[\h]{i}{j}$.  By Observation \ref{obs:g_notin_bottom}, $\h \notin \discard[\h]{i}{j}$ and thus 
	$$X_i <_i (X_j \cup \{\h\}) \setminus \discard[\h]{i}{j} = (X_j \setminus \discard[\h]{i}{j}) \cup \{\h\} <_i (X_j \setminus \discard[\h]{i}{j}) \cup \{\g\},$$
	where the first inequality holds since $1 \champ[\h] 2$, the equality holds since $\h \notin \discard[\h]{i}{j}$ and the second inequality is by \cancbty\ since $\h <_i \g$.  Therefore, since $i$ is the unique $\g$-champion of $j$, it follows that $i$ is the most envious agent of $(X_j \setminus \discard[\h]{i}{j}) \cup \{\g\}$ (otherwise that subset has another most envious agent and consequently $X_j \cup \{\g\}$ has other most envious agents except $i$). 
	Denote the corresponding minimally envied subset by $S$.  Thus, we can choose $\discard[\g]{i}{j}$ to be $X_j\setminus S$. Clearly $\discard[\h]{i}{j} \subseteq \discard[\g]{i}{j}$.
\end{proof}
\fullversionend

\begin{proof}[Proof of Lemma~\ref{lem:envy_unique_structure}]	
By assumption, there is an envy edge in $M_\allocs$. Rename the agents such that $1 \envies 4$. By Observation~\ref{obs:exists-champion}, agent $1$ has a $\g$-champion. If $4 \champ[\g] 1$, then $1 \envies 4 \champ[\g] 1$ is a PI cycle, so assume that another agent is a $\g$-champion of 1. Rename that agent such that $2\champ[\g] 1$. Agent $2$ has an $\h$ champion. If $4\champ[\h]2$ or $1\champ[\h]2$ then we have a PI cycle: $1\envies 4\champ[\h] 2 \champ[\g] 1$ or $1\champ[\h]2\champ[\g]1$, respectively. Thus $3\champ[\h]2$. We obtain the following structure:

\begin{center}
	\includegraphics[scale=\figurescale]{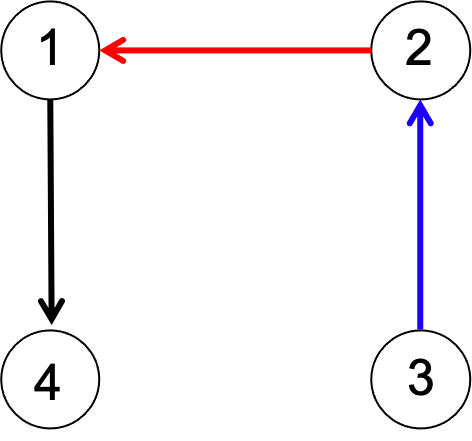}
\end{center}

Agent 3 has a $\g$-champion. If $2\champ[\g]3$ we have a PI cycle: $2\champ[\g]3\champ[\h]2$, thus $1\champ[\g]3$ or $4\champ[\g]3$. It is easy to see that in both cases we must have $2\champ[\h]1$ as any other $\h$-champion of 1 closes a PI cycle. Consequently we must have $3\champ[\g]2$ as well. Therefore, we have one of the following two structures:

\begin{center}
	\includegraphics[scale=\figurescale]{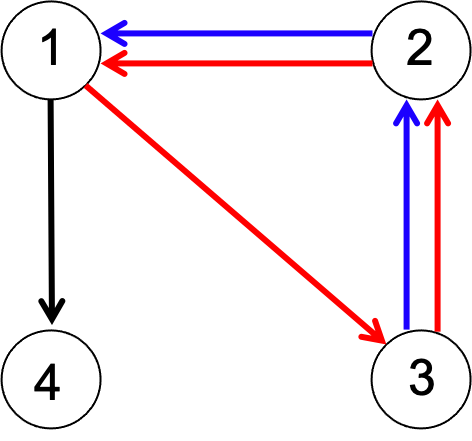}
	\qquad\qquad\qquad
	\includegraphics[scale=\figurescale]{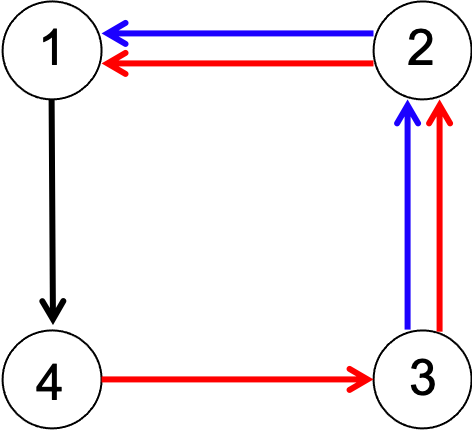}
\end{center}
Moreover, if $2 \champ[\h] 3$ then we are done via the PI cycle $2 \champ[\h] 3 \champ[\g] 2$, hence we may also assume that $1 \champ[\h] 3$ or $4 \champ[\h] 3$ (in addition to $1 \champ[\g] 3$ or $4 \champ[\g] 3$).
\begin{claim}\label{clm:1_nchamp_g,h_3}
	If $1 \champ[\g] 3$ or $1 \champ[\h] 3$ then $M_\allocs$ contains a PI cycle.
\end{claim}
\begin{proof}
	We prove for $1 \champ[\g] 3$ (the proof for $1 \champ[\h] 3$ is analogous).
	It easy to verify that any envy edge going into one of the agents 1,2 or 3 would close a PI cycle. Hence, we may assume there are no such edges. Thus, the cycle $1\champ[\g]2\champ[\g]3\champ[\g]1$ is a good $\g$-cycle,
	which, by Observation \ref{obs:good_cycle_g_decomposition}, induces a $\g$-decomposition $X_j = T^{\g}_j \cupdot B^{\g}_j$ for $j \in \{1,2,3\}.$
	
	By Theorem~\ref{thm:good-edge-or-external} $M_\allocs$ contains a good or external $\Bgone$-edge. If it is external then it is one of $4\champ[\Bgone\mid\circ]2$, $4\champ[\Bgone\mid\circ]3$ (note that $4\nchamp[\Bgone\mid\circ]1$ by Observation~
	\ref{obs:no_B_j_champ_of_j}), and in each case we have a PI cycle:  $4\champ[\Bgone\mid\circ]2\champ[\g]1\envies4$, $4\champ[\Bgone\mid\circ]3\champ[\h]2\champ[\g]1\envies4$, respectively.
	Therefore, there must be a good $\Bgone$-edge which is either one of $1\champ[\Bgone\mid\circ]2$, $1\champ[\Bgone\mid\circ]3$ or $3\champ[\Bgone\mid\circ]2$. In the first two cases we obtain a PI cycle: $1\champ[\Bgone\mid\circ]2\champ[\g]1$ and $1\champ[\Bgone\mid\circ]3\champ[\h]2\champ[\g]1$, respectively. Thus we may assume $3\champ[\Bgone\mid\circ]2$ and we obtain the following structure
	\begin{center}
		\includegraphics[scale=\figurescale]{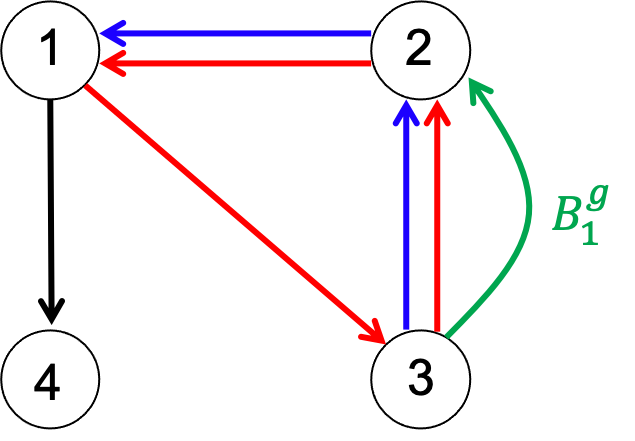}
	\end{center}
	
	By Observation~\ref{obs:exists-champion}, there exists an $\h$-champion of 3 which can be either 1,2 or 4.  In all these cases we get a PI cycle:
	$$1\champ[\h]3 \linebreak[1] \champ[\Bgone\mid\circ]2\linebreak[1]\champ[\g]1,~~2\champ[\h]3\champ[\g]2 ~\text{or}~ 4\champ[\h]3\linebreak[1]\champ[\Bgone\mid\circ]2\linebreak[1]\champ[\g]1\envies 4,$$
	respectively.  We have shown a PI cycle in every case and thus we are done.
\end{proof}
 By Claim \ref{clm:1_nchamp_g,h_3} we may assume $4$ is the unique $\g$ and $\h$ champion of 3. We thus obtain the following structure:
\begin{center}
	\includegraphics[scale=\figurescale]{figures/With_envy_p02-4}
\end{center}

It is not hard to verify that any additional envy edge closes a PI cycle. Moreover, any additional $\h$ or $\g$ champion of 1,2 or 3 closes a PI cycle. This fact is not hard to verify for agents 1 and 2, and we have shown it explicitly for agent 3 in Claim \ref{clm:1_nchamp_g,h_3}.  As for agent 4, the possible $\g$ and $\h$ champions of 4 are agents 1 and 2 ($3 \champ[\g] 4$ also closes an immediate PI cycle). This prove the lemma.
\end{proof}

\fullversion
\begin{proof}[Proof of Claim \ref{clm:X_4<_4T_3+B_2}]
	Since 4 envies $\Tgthree \cup \Bgtwo$, there exists a most envious agent of this bundle. That is, there exists an agent $i$ such that $i\champ[\Bgtwo\mid\circ]3$. 
	$i$ cannot be 3,
	since $X_3 >_3 X_2 = \Tgtwo \cup \Bgtwo >_3 \Tgthree \cup \Bgtwo$, where the first inequality is by $3\nenvies2$ and the second inequality is by Observation~\ref{obs:T_k<T_j}
	and \cancbty.
	
	In the remaining cases, $i = 1$, $i=2$, $i=4$, we obtain the respective PI cycles:
	$$1\champ[\Bgtwo\mid\circ]3\champ[\g]2\champ[\h]1,~2\champ[\Bgtwo\mid\circ]3\champ[\g]2,~4\champ[\Bgtwo\mid\circ]3\champ[\g]2\champ[\h]1\envies4.$$
\end{proof}

\begin{proof}[Proof of Lemma \ref{lem:X_4>_1_Z}]
	
	Consider the allocation $\mathbf{Y}$ obtained from $\allocs'$ by replacing $X'_2$ with $Z$:
	
	
	\begin{center}
		\renewcommand{\arraystretch}{1.2}
		$\mathbf{Y}$=\bundle{1}{X_4}\quad\bundle{2}{Z}\quad\decomp{3}{\Tgtwo}{\g}\quad\decomp{4}{\Ththree}{\h}	
	\end{center}
	%
	%
	We claim this allocation is EFX.
	
	\noindent\textbf{No agent envies agent 2:}
	Agents 3 and 4 do not envy agent 2 since they did not envy her in $\allocs'$ and we only removed items from $X'_2$ in the transition to $\mathbf{Y}$. Agent 1 does not envy agent 2, since we assume $X_4>_1 Z$.
	
	\noindent\textbf{Agent 2 envies no other agent:}
	This follows from the definition of $Z$.
	
	Since only the bundle of agent 2 has been changed in the transition from $\allocs'$ to $\mathbf{Y}$ we conclude that there is no strong envy that does not include agent 2 (as there wasn't any in $\allocs'$), and consequently $\mathbf{Y}$ is EFX.
	Moreover, agents 1,3, and 4 are better off in allocation $\mathbf{Y}$ relative to $\allocs$. Since we assumed that agent $\vip$ is not agent 2, it follows that $\mathbf{Y}$ dominates $\allocs$ and we are done.
\end{proof}

\begin{proof}[Proof of Claim~\ref{clm:4_MEA_T'_3_cup_h}]
	It suffices to show that only agent 4 envies $\Ththree \cup \h$ in $\allocshat$.
	Agent 4 envies $\Ththree \cup \h$, since $\hX_4=X_4$ in all three cases and $4\champ[\h]3$ in the allocation $\allocs$. 
	Agent $i$, for $i\in\{1,3\}$, does not envy $\Ththree \cup \h$, since $\hX_i \geq_i X_i$ and $i\nchamp[\h] 3$ in allocation $\allocs$ (Observation~\ref{obs:non-champion-doesnt-envy-top-half}).
	Agent $2$ does not envy $\Ththree \cup \h$ in all three cases since $\hX_2 = \max_{2} \left\{ \Tgtwo \cup \g,\, \Ththree \cup \h,\, X_3\right\} \geq_2 \Ththree \cup \h$. 
\end{proof}

\begin{proof}[Proof of Claim~\ref{clm:3_2nd_MEA_implies_PI_cycle}]
Consider the cycle $1\envies4\champ[\g]3\champ[\Bgthree\mid\circ]2\champ[\h]1$ in $M_\allocs$:
\begin{center}
	\includegraphics[scale=\figurescale]{figures/With_envy_p10-5.png}
\end{center}
This cycle might not really exist, since it could be the case that 1 is the real ($\Bgthree \mid \circ$)-champion of agent 2.  Assume for now that this is not the case, {\sl i.e.}, $3\champ[\Bgthree\mid\circ]2$ does exist in $M_{\allocs}$.  Then this is PI cycle, and by Lemma \ref{lem:pareto-improvable}, there is an allocation $\mathbf{Y}$ that Pareto-dominates $\allocs$.  Note that in $\mathbf{Y}$ agent 1 receives $X_4$.  We therefore claim that $\mathbf{Y}$ is EFX even if 1 was the most envious agent of $\Tgtwo \cup \Bgthree$ in $\allocs$.  Since 3 was the most envious agent ignoring agent 1, the only thing that could prevent $\mathbf{Y}$ from being EFX is if 1 strongly envies 3 in $\mathbf{Y}$.  But this is not the case since
$$ Y_1 = X_4 >_1 \Tgtwo \cup \Bgthree \geq_1 Y_3$$ where the last inequality holds since $Y_3 \subseteq \Tgtwo \cup \Bgthree$.  The claim follows.
\end{proof}

\begin{proof}[Proof of Lemma~\ref{lem:T_2<_1_T_4}]
Consider the good $\g$-cycle $2 \champ[\g] \linebreak[1] 4 \champ[\g] \linebreak[1] 3 \champ[\g] \linebreak[1] 2$ in $\allocs$.
By Lemma~\ref{lem:alg_start}, there exists an agent $i$ such that $i\champ[\Bgtwo\mid\circ]4$.
$i$ cannot be agent 3, because $3=\pred(2)$ in this cycle (Observation~\ref{obs:pred(i)_nchamp_B(i)}). If $i=2$, we obtain the PI-cycle $2\champ[\Bgtwo\mid\circ]4\champ[\h]3\champ[\g]2$ in the original allocation $\allocs$, so we are done. If $i=4$, then 
$%
	 X_4<_4 \Tgfour \cup \Bgtwo <_4 \Tgthree \cup \Bgtwo,%
$
where the second inequality is due to Observation~\ref{obs:T_k<T_j} and \cancbty. Therefore, we are done by Claim~\ref{clm:X_4<_4T_3+B_2}.

Thus, we may assume that $i=1$, {\sl i.e.}, $1\champ[\Bgtwo\mid\circ]4$. It follows that $X_1 <_1 \Tgfour \cup \Bgtwo$. Therefore, since $1\nenvies2$,
$$
	\Tgtwo\cup\Bgtwo=X_2<_1 X_1 <_1 \Tgfour \cup \Bgtwo,
$$
and we conclude that $\Tgtwo <_1 \Tgfour$ by \cancbty.  Finally we use \cancbty\ again to conclude $X_4 = \Tgfour \cup \Bgfour >_1 \Tgtwo \cup \Bgfour$, as desired. 
\end{proof}

\begin{proof}[Proof of Lemma~\ref{lem:X_3<_3_T_2_cup_B_4}]
Consider the good $\g$-cycle $2 \champ[\g] \linebreak[1] 4 \champ[\g] \linebreak[1] 3 \champ[\g] \linebreak[1] 2$ in $\allocs$.  We first ask which agent $i$ satisfies $i\champ[\Bgfour\mid\circ]3$ (there exists such an agent by Lemma~\ref{lem:alg_start}).
$i$ cannot be agent 2, because $2=\pred(4)$ in the good $\g$-cycle (Observation~\ref{obs:pred(i)_nchamp_B(i)}). If $i$ is agent 4, we obtain the PI-cycle $2\champ[\g]4\champ[\Bgfour\mid\circ]3\champ[\h]2$ in $M_\allocs$, hence we are done.

If $i = 1$, then in particular we have 
$$
	\Tgthree \cup \Bgthree = X_3  <_1 X_1 <_1 \Tgthree\cup\Bgfour,
$$
where the first inequality is due to $1\nenvies3$ and the second holds since $1\champ[\Bgfour\mid\circ]3$.  By \cancbty\ we get $\Bgthree <_1 \Bgfour$, and together with Lemma \ref{lem:T_2<_1_T_4} we conclude that
$
X_4 = \Tgfour\cup\Bgfour>_1 \Tgtwo\cup \Bgfour>_1\Tgtwo\cup\Bgthree,
$
{\sl i.e.}, we are in Case \circled{I} which we already solved. 

Thus, we may assume that $i$ is agent 3, i.e., $3\champ[\Bgfour\mid\circ]3$. It follows that
$$
	X_3 <_3 \Tgthree\cup\Bgfour <_3 \Tgtwo\cup\Bgfour,
$$
where the second inequality is by Observation~\ref{obs:T_k<T_j} and \cancbty.  The lemma follows.
\end{proof}

\begin{proof}[Proof of Claim \ref{clm:b<_3_g}]
Assuming that $\g <_3 \b$, We get
$X_3 <_3 \Tgtwo \cup \{\g\} <_3 \Tgtwo \cup \{\b\}$, where the first inequality holds since $3 \champ[\g] 2$ in $\allocs$ and the second by \cancbty.  
In particular there exists a most envious agent of $\Tgtwo \cup \{\b\}$ in $\allocs$, {\sl i.e.}, there exists an agent $i$ such that $i\champ[\b]2$. 
$i$ cannot be 2,
since $X_2 >_2 X_1 = \Tgone \cup \Bgone >_2 \Tgtwo \cup \Bgone >_2 \Tgtwo \cup \{\b\}$, where the first inequality holds by $2\nenvies1$, the second by Observation~\ref{obs:T_k<T_j} and \cancbty,
and the third by monotonicity.	
Thus $i = 1, 3$ or $4$ and in all cases we get a PI cycle in $M_\allocs$:
$$1\champ[\b]2\champ[\g]1~,~3\champ[\b]2\champ[\g]1\envies4\champ[\h]3~\text{ or }~4\champ[\b]2\champ[\g]1\envies4,$$
respectively.  Note that these are indeed PI cycles since the $\g$-champion edge releases $\Bgone$ which contains $\b$.
\end{proof}

\begin{proof}[Proof of Lemma \ref{lem:i-champB_3|->4}]
	Recalling that $\hX_4 = X_4$, we need to show that there is a most envious agent of $\Tgfour \cup \Bgthree$ in $\allocshat$. We do this by showing $\hX_1	<_1 \Tgfour\cup \Bgthree $.
	If $\Bgtwo >_1 \Bgthree$, then
	$$
	X_4 >_1 X_1 >_1 X_2 = \Tgtwo \cup \Bgtwo >_1 \Tgtwo \cup \Bgthree,
	$$
	where the first and second inequalities are by $1\envies4$ and $1\nenvies2$, respectively, in $M_{\allocs}$,
	and the final inequality follows from \cancbty. This implies that we are in Case \circled{I} that has been solved earlier.
	
	Thus we may assume $\Bgtwo <_1 \Bgthree$.  Moreover $\Tgone <_1 \Tgfour$ by Observation \ref{obs:T_k<T_j}. Therefore, by \cancbty
	$$
	\Tgone\cup \Bgtwo <_1 \Tgfour \cup \Bgtwo <_1 \Tgfour \cup \Bgthree.
	$$
	
	Seeing as $\hX_1=Z\subseteq \Tgone\cup\Bgtwo$ and valuations are monotone, the above equation implies that agent 1 envies $\Tgfour \cup \Bgthree$ in allocation $\allocshat$, as desired.
\end{proof}

\begin{proof}[Proof of Claim \ref{clm:1-B_2|->4}]
	Recall that $X_1 <_1 X_4 <_1 Z \subseteq \Tgone\cup\Bgtwo$.
	We thus have
	$%
	X_1 <_1 \Tgone\cup\Bgtwo <_1 \Tgfour\cup\Bgtwo,%
	$
	by Observation~\ref{obs:T_k<T_j} and \cancbty. Since agent 1 envies $\Tgfour\cup\Bgtwo$, there exists a most envious agent of this set. That is, there exists $j\in[4]$, such that $j\champ[\Bgtwo\mid\circ]4$. 
	
	If $2\champ[\Bgtwo\mid\circ]4$, then we obtain the PI-cycle
	$2\champ[\Bgtwo\mid\circ]4\champ[\h] 3 \champ[\g] 2$.
	%
%
	Note that $3\nchamp[\Bgtwo\mid\circ]4$,
	since
	$$%
	X_3>_3 X_2 = \Tgtwo\cup\Bgtwo >_3 \Tgfour\cup \Bgtwo,%
	$$
	where the first inequality is by $3\nenvies2$ and the second is by Observation~\ref{obs:T_k<T_j} and \cancbty.
	
	Similarly, $4\nchamp[\Bgtwo\mid\circ]4$,
	since $X_4>_4 \Tgthree\cup\Bgtwo$ (otherwise we are done by Claim~\ref{clm:X_4<_4T_3+B_2}), and therefore, by Observation~\ref{obs:T_k<T_j} and \cancbty,
	$%
	X_4>_4 \Tgthree \cup\Bgtwo >_3 \Tgfour \cup\Bgtwo.%
	$	
	The claim now follows.
\end{proof}

\begin{proof}[Proof of Claim \ref{clm:3_champ_B_4_3_V_2_champ_B_4_3}]
	By Claim \ref{clm:1-B_2|->4} we have the following structure in $M_\allocs$:
	\begin{center}
		\includegraphics[scale=\figurescale]{figures/With_envy_p14-2}
	\end{center}
	
	Now, we have
	$%
	X_4 = \Tgfour \cup \Bgfour <_4 \Tgthree \cup \Bgfour,%
	$
	by Observation~\ref{obs:T_k<T_j} and \cancbty. Since agent 4 envies $\Tgthree\cup\Bgfour$, there exists a most envious agent of this set. That is, there exists $i\in[4]$, such that $i\champ[\Bgfour\mid\circ]3$.
	If $i=1$, then
	$%
	\Tgthree \cup \Bgfour >_1 X_1 >_1 X_3 = \Tgthree \cup \Bgthree,%
	$
	where the first inequality holds by definition of $i$ and the second inequality follows from $1\nenvies3$. By \cancbty, we conclude that 
	$
	\Bgfour >_1 \Bgthree,
	$ 
	and together with $\Tgfour>_1 \Tgtwo$ (Observation~\ref{obs:T_k<T_j}) we get
	$$
	X_4=\Tgfour\cup\Bgfour >_1 \Tgfour\cup\Bgthree >_1 \Tgtwo\cup\Bgthree.
	$$
	by \cancbty. Thus we are in Subcase~\circled{I} which has been solved earlier.
	
	Suppose next that $i=4$. 
	By Claim~\ref{clm:1-B_2|->4}, $1\champ[\Bgtwo\mid\circ]4$, thus we obtain the PI-cycle
	$1\champ[\Bgtwo\mid\circ]4\champ[\Bgfour\mid\circ]3\champ[\g]2\champ[\h]1$
	depicted below.
	\begin{center}
		\includegraphics[scale=\figurescale]{figures/With_envy_p14-3}
	\end{center}
	To see that it is indeed a PI-cycle, notice that the set $\Bgtwo$ is released by $3\champ[\g]2$ and the set $\Bgfour$ is released by the edge $1\champ[\Bgtwo\mid\circ]4$. 
	
	Assume now that $i=2$.
	Since we assume that $\Bgfour <_2\Bgtwo$, we have 
	$%
	X_2<_2 \Tgthree \cup \Bgfour <_2 \Tgthree \cup \Bgtwo,%
	$
	where the first inequality follows from $2\champ[\Bgfour\mid\circ]3$ and the second is by \cancbty. Therefore, agent 2 envies $\Tgthree \cup \Bgtwo$ and thus there exists an agent $k\in[4]$ such that $k\champ[\Bgtwo\mid\circ]3$. 
	Note that $3\nchamp[\Bgtwo\mid\circ]3$, since
	$%
	X_3>_3 X_2 =\Tgtwo \cup\Bgtwo >_3 \Tgthree \cup\Bgtwo,%
	$
	where the first inequality is by $3\nenvies2$ and the second inequality is by Observation~\ref{obs:T_k<T_j} and \cancbty.
	Thus, we may assume that either $1\champ[\Bgtwo\mid\circ]3$, $2\champ[\Bgtwo\mid\circ]3$ or $4\champ[\Bgtwo\mid\circ]3$, in each of these cases we obtain a PI-cycle:
	
	$$1\champ[\Bgtwo\mid\circ]3\champ[\g]2\champ[\h]1~,~2\champ[\Bgtwo\mid\circ]3\champ[\g]2~ \text{ or }~4\champ[\Bgtwo\mid\circ]3\champ[\g]2\champ[\h]1\envies4, $$
	respectively.  We are left with the case $i=3$, and we are done. 
\end{proof}

\fullversionend

\shortversion

\shortversionend
	

\end{document}